\documentclass[12pt]{article}
\usepackage[margin=1 in]{geometry}
\usepackage{mathptmx} %
\RequirePackage{multicol}
\RequirePackage{setspace}
\RequirePackage[leqno]{amsmath}
\usepackage{mathtools}
\usepackage{graphicx}
\usepackage{xcolor}
\newif\ifshowrevisions
\showrevisionsfalse   %
\ifshowrevisions
  \newcommand{\revision}[1]{\textcolor{red}{#1}}
\else
  \newcommand{\revision}[1]{#1}
\fi

\makeatletter
\renewcommand\section{\@startsection{section}{1}{\z@}%
  {-3.5ex \@plus -1ex \@minus -.2ex}%
  {2.3ex \@plus.2ex}%
  {\bfseries\large\center}} %

\renewcommand\subsection{\@startsection{subsection}{2}{\z@}%
  {-3.25ex\@plus -1ex \@minus -.2ex}%
  {1.5ex \@plus .2ex}%
  {\itshape\large\center}} %

\renewcommand\refname{REFERENCES}

\newcommand\TenPtJEL{\@setfontsize\TenPtJEL{10.5}{12}}
\long\def\JEL#1{\long\gdef\@JEL{#1}}

\makeatother
	\usepackage[english]{babel} 
	\usepackage{authblk}
	\usepackage{amsmath, amssymb, textcomp, amsthm, mathtools,mathrsfs} 
	\allowdisplaybreaks
 \usepackage{longtable}
	\usepackage{wasysym} 
	\usepackage{stmaryrd}
	\usepackage{dsfont}
 \usepackage{setspace}
 \usepackage{pdflscape}

	\usepackage{hyperref}
 \usepackage{fancyhdr}
   	\usepackage{booktabs}
   	\usepackage{enumitem}
   	\usepackage{inputenc}
	\usepackage{csquotes}
	\usepackage{floatrow}
	\usepackage{footnote,footmisc} 
	\usepackage{setspace} 
	\usepackage{threeparttable}
	\usepackage{tikz}
 \usetikzlibrary{calc, decorations.pathreplacing}
\usepackage{comment}
\usepackage{graphicx}
\usepackage{algorithm}
\usepackage{algorithmic}
 \usepackage{subcaption}
 \usepackage{pgfplots}
\usepgfplotslibrary{statistics}
\pgfplotsset{compat=1.16}

	\usepackage{pgfplots}
	\pgfplotsset{compat=1.17}
	\usetikzlibrary{shapes.misc}
	\usetikzlibrary{patterns}
	\usepackage{xcolor}
	\hypersetup{
    	colorlinks,
    	linkcolor={red!50!black},
    	citecolor={blue!50!black},
    	urlcolor={blue!80!black}
	}
	\usetikzlibrary{decorations.pathreplacing}
	\usetikzlibrary{arrows}
	\usepackage{float} 
	\usepackage{pifont}%

\usepackage{floatrow}

\newcommand{\cmark}{\ding{51}}%
	\newcommand*\dd{\mathop{}\!\mathrm{d}}

	\usepackage[toc,page]{appendix}
	
	\newtheorem{theorem}{Theorem}
	 
	\newtheorem{proposition}{Proposition}

	\newtheorem{assumption}{Assumption}

	\newtheorem*{example*}{Example}
	\usepackage{graphicx}  
	\usepackage{marvosym}

    \DeclareMathOperator*{\argmin}{\arg\!\min}

    \newtheorem*{assumptions*}{\assumptionnumber}
\providecommand{\assumptionnumber}{}

\usepackage{caption}
\usepackage{mwe}

\definecolor{cyan}{cmyk}{1, 0.4, 0, 0}

\definecolor{mypink}{RGB}{219, 48, 122}

\def\XXint#1#2#3{{\setbox0=\hbox{$#1{#2#3}{\int}$ }
\vcenter{\hbox{$#2#3$ }}\kern-.6\wd0}}

\usepackage[round]{natbib}
\makeatletter

\allowdisplaybreaks

 \def\thanks#1{\protected@xdef\@thanks{\@thanks
        \protect\footnotetext{#1}}}
\makeatother
\title{\Huge{Return to Office and the Tenure Distribution}}
\author{\normalsize David Van Dijcke$^{1}$\thanks{Contact: \texttt{eju5az@virginia.edu, fgunsil@emory.edu, austinlw@uchicago.edu}. Author order reflects relative contributions. We thank Zach Brown, Emma Harrington, Mark Polyak, Mel Stephens, Germano Wallerstein, and participants at the 2025 Remote Work Conference and the 2026 AEA Meetings for helpful discussion and comments as well as Ewan Rawcliffe for outstanding research assistance. The authors gratefully acknowledge support from the Peter G. Peterson Foundation Pandemic Response Policy Research Fund, the Becker Friedman Institute, and the University of Michigan via a MITRE award. This research was supported in part through computational resources and services provided by Advanced Research Computing at the University of Michigan, Ann Arbor. This work was conceptualized and a working paper version was completed while the first and second authors were visiting the Department of Economics at Princeton University. The authors acknowledge the use of Claude Opus 4.8 and ChatGPT 5.5 for proofreading the final manuscript and helping to prepare the replication package. This is the authors' final manuscript, accepted for publication in \textit{The Review of Economics and Statistics}, published by MIT Press. Copyright held by the President and Fellows of Harvard College and the Massachusetts Institute of Technology.}}
\author{\normalsize Florian Gunsilius$^{2}$}
\author{Austin Wright$^{3}$}
\affil{\footnotesize $^{1}$Department of Economics, University of Virginia \\ 
$^{2}$ Department of Economics, Emory University\\
$^{3}$Harris School of Public Policy, University of Chicago}
\date{\vspace{-1em} \small September 24, 2026}

\newif\ifsubmissionformat
\submissionformatfalse  %

\ifsubmissionformat
  \let\origfigure\figure
  \let\origendfigure\endfigure
  \let\origtable\table
  \let\origendtable\endtable
  \usepackage[nolists]{endfloat}
  \renewcommand{\floatfoot}[1]{\par\vspace{0.5em}\noindent\footnotesize #1}
\fi

\begin{document}

{\ifsubmissionformat\else\singlespacing\fi
\maketitle
\vspace{-2em}
\begin{abstract} 
\noindent

Debates over return-to-office mandates have intensified since the COVID-19 pandemic ended, though their economic implications remain poorly understood. Using large-scale resume data, we analyze how these policies affect employee tenure and seniority at three large U.S. tech companies: Microsoft, SpaceX, and Apple. Employing a distributional synthetic controls framework, we estimate a reduction in tenure and seniority at firms returning to the office. We extend the framework with a functional bootstrap procedure and prove its validity for constructing uniform confidence bands. Return-to-office mandates appear to drive senior employees away, posing a risk to firm productivity and innovation.


\end{abstract}
\thispagestyle{empty} %
\newpage
\setcounter{page}{1} %
}

\ifsubmissionformat\doublespacing\else\onehalfspacing\fi

\section{Introduction} Since the World Health Organization officially declared an end to the COVID-19 health crisis, companies have shifted their focus towards creating workplace policies for a post-pandemic economy. In this, the preferences of employers and employees often seem diametrically opposed. That opposition has expressed itself most notably in the debate around the effects of the return to office (RTO). This debate is far from settled: as of June 2023, an estimated 75\% of tech and 28\% of all companies in the United States were still ``fully flexible'' -- either fully remote or with a voluntary in-office option \citep{flex2023}. Most business executives expect this number to grow \citep{bloom2023survey}. In this paper, we provide causal evidence that RTO mandates at three large tech companies---Microsoft, SpaceX, and Apple---had a negative effect on the tenure and seniority of their respective workforces. In particular, we find the strongest negative effects at the top of the distribution, implying a more pronounced exodus of relatively senior personnel.

To do so, we estimate the composition of the workforce in the counterfactual reality where the company did not implement an RTO mandate. We document changes in the employee composition that can profoundly impact a company, especially since these changes occur at the top of the hierarchy. Senior employees and those that have been at a company for a long time possess invaluable human capital and tend to have elevated productivity levels \citep{lazear2015value}, which they take with them when leaving the company. They also represent a significant investment in terms of hiring and training costs \citep{blatter2012costs}. Indeed, the central tenet of resource-based theory, one of the most prevalent and empirically well-supported human capital theories, is that human capital is key to explaining why some firms outperform others \citep{crook2011does}. Changes in the employee distributions can, moreover, have ripple effects in the larger economy, by affecting the distribution of skilled workers across different types of companies \citep{lise2017macrodynamics, eeckhout2018assortative}. 

We uncover these results by estimating the distributional causal effects of a return-to-office mandate on the makeup of a company's workforce using a synthetic controls approach \citep{abadie2003economic, abadie2010synthetic}. Since our object of interest is the heterogeneity in the composition of the workforce of a company, we employ the distributional synthetic controls estimator proposed in \citet{gunsilius2023distributional}. Our ability to account for the heterogeneity in the organizational structure of companies is of fundamental importance in this setting. This estimator is also ideally suited for this setting, since it does not require observing the same individuals over time and as such allows for the very changes in workforce composition that we are interested in.  Furthermore, it provides a nonparametric estimate of the effect of an RTO on the entire tenure and seniority distribution of a target company. Beyond applying this estimator, we make three methodological contributions. First, we develop bootstrap-based uniform confidence bands and prove their validity under regularity conditions (Appendix \ref{app:bootstrap}), allowing researchers to identify which parts of a distribution show statistically significant effects; mathematically, this contribution is completely novel in that it provides bootstrapped confidence intervals for these types of functional data, namely replicated quantile function processes. Second, we extend the methodology to ordinal categorical outcomes (like job titles) by working with mixtures of distributions rather than mixtures of quantiles. Third, we propose a ratio-based permutation test that accounts for differential pre-treatment fit quality across placebo units.

Our case-study approach focuses on Microsoft, Apple, and SpaceX because they were among the first large American tech companies to implement RTO mandates. In particular, they did so before a wave of layoffs started hitting the tech industry from late 2022 onwards, allowing us to cleanly disentangle the causal effects of the RTO mandates. Our synthetic controls approach, moreover, allows us to study in detail what happens \textit{within} each company after an RTO, which would not be possible with an aggregate cross-company approach. This guarantees the internal validity of our estimates. Their external validity is twofold. First, we consider some of the largest companies in a sector where the discourse over the return to office was most heated. Together, the three companies account for over 2\% of employment in the tech sector and 30\% of its revenue. What happens at these companies matters for the American economy, and sets the precedent for the wider debate around the return to office. Second, we estimate \revision{very similar} effects for all three companies despite their markedly different corporate culture and product gamut, suggesting the effects are driven by common underlying dynamics. \revision{Nonetheless, an important limitation of this case-study design is that the estimates are most directly informative about early RTO mandates at these three large U.S. technology firms under the relatively tight labor-market conditions of 2022. We therefore interpret the results as internally valid evidence for these settings rather than as an unconditional average effect of RTO mandates across all firms, industries, or labor-market regimes.}

Our paper contributes to several branches of the literature. First, it addresses the relatively understudied impact of RTO mandates. For instance, \citet{ding_return--office_2023} analyze the determinants and effects of RTO mandates for a sample of S\&P 500 companies by examining company balance sheets and Glassdoor reviews. They suggest that RTOs are used by managers to reassert control, which in turn reduces employee satisfaction. Using a voluntary survey of managers and their employees, \citet{van_triest_your_2023} demonstrates that work-from-home (WFH) levels are 50\% lower when managers have the discretion to set RTO policies, as opposed to being constrained by organization-wide rules.  \citet{barrero_why_2021} report that four out of ten employees working from home would find a new job under an RTO mandate, based on a large-scale survey of American workers. Several business and industry reports have also studied RTO. For example, a large-scale survey by \citet{unispace2023returning} found that almost half of firms with mandated RTOs experienced higher than usual difficulties in retaining employees and lost ``key'' staff members, in line with our findings. In another survey of over 3,500 employees, one in three executives said they would leave their job if faced with an RTO mandate, compared to 19\% of non-executives \citep{Gartner2024}. Our paper complements and extends this literature by estimating a significant outflow of more senior employees after the implementation of an RTO mandate. Unlike these previous studies, we rely on large-scale resume data rather than surveys, allowing us to study the realized consequences of RTO policies rather than self-reported behavior, intentions, or preferences. Moreover, the richness of our data and estimation approach allows us to speak to causal rather than merely descriptive effects. It also allows us to study what happens to the organizational structure of companies that return to the office, which is a first-order effect of such a return. 

Second, our paper relates to the large literature on WFH. Focusing on the impact of WFH on employee tenure and seniority, this body of work provides several insights relevant to ours.\footnote{For a comprehensive survey of the effects of WFH, we refer to \citet[\S 2.3]{ding_return--office_2023} and \citet{barrero_evolution_2023}.} \citet{mas_valuing_2017} estimate, using an experimental design, that workers had a willingness to pay for hybrid work arrangements of around 8\% of total pay before the pandemic, with a fat tail of high-valuation workers. \citet{emanuel2023power} find that proximity to coworkers in the office increases the feedback software engineers receive but reduces their programming output, especially for senior engineers. \citet{hansen_remote_2023} study hybrid work options in large-scale job postings data and find substantial heterogeneity within occupation and even company. \citet{mischke2023empty} and \citet{barrero_evolution_2023} find that office attendance is lower in larger firms, in line with our finding that employees leave for such firms at higher rates after an RTO. \citet{barrero_shift_2022} find that a shift to remote work puts downward pressure on wages due to lower employee turnover, which may have reduced the college wage premium as highly educated workers are more likely to WFH \citep{barrero_evolution_2023}. Finally, \cite{bloom2022hybrid} document lower attrition rates under a hybrid work model compared to a fully-in-person one.

An important distinction between our focus on RTO mandates and the WFH literature is that we consider the effects of a \textit{mandatory shift} to (partial) office attendance from an optional hybrid model, while the WFH literature has mostly focused on static comparisons between fully-remote, hybrid, and fully-in-person work. Moreover, while several of the papers mentioned have documented a positive relation between employee attrition rates and degree of required office attendance, our results illuminate how this attrition realizes across the employee distribution, with senior and long-tenured employees leaving at higher rates. Considering the pivotal role such employees play in a company's development and competitive edge, along with the significant costs involved in replacing the firm-specific human capital accumulated over their tenure \citep{gerhart2021resource}, our findings indicate that such ``uneven'' attrition should be a major concern for any company considering a return to office.

The rest of this paper is structured as follows. In Section \ref{sec:empirical_approach} we discuss our empirical approach in terms of data and estimation and present some descriptive evidence; in Section \ref{sec:results} we present our causal results and discuss their implications in light of the academic literature; before concluding in Section \ref{sec:conclusion}.

\section{Empirical Approach} \label{sec:empirical_approach}
 
\subsection{Data} \label{sec:data}

\paragraph{People Data Labs} We study the impact of RTO across the employee distribution using data provided by People Data Labs (PDL) (\url{https://www.peopledatalabs.com/}) via the Dewey Data platform (\url{https://deweydata.io/}) \citep{pdl_cite}. The dataset consists of global resume data for nearly 1 billion unique individuals, with 1.2 billion person records in North America from over 300 million resumes. At the companies of interest, Microsoft, Apple, and SpaceX, we observe \revision{49,083, 29,297, and 2,807 resumes in the quarter immediately preceding each RTO}, respectively. The data contain information on the start and end date of a person's position at a given company, as well as the role (trainee, entry, partner, etc.) of the position.

The data are obtained from a wide variety of sources, comprising two main categories:\footnote{See \url{https://docs.peopledatalabs.com/docs/data-sources} for more details.} first, publicly available data, which includes career websites such as LinkedIn (over 700 million profiles globally), social media platforms such as Facebook (over 1 billion profiles) and Twitter (over 50 million profiles), as well as public records and publicly available government data; second, PDL's ``Data Union'', which is a proprietary data-sharing co-op through which some of PDL's clients share their internal data in a confidential and secure manner. These clients include companies in HR and Real Estate tech as well as in the identity and anti-fraud industries. Because records are often obtained directly from employer systems in these industries, they may be more complete than self-reported online profiles, though we cannot verify this directly. Moreover, many of the clients in the Data Union are software companies, which means that, in practice, we observe a substantial share of employees for most companies in our sample of tech firms. PDL cross-references all these sources, deduplicates them, and assesses their quality through a battery of tests.\footnote{See \url{https://docs.peopledatalabs.com/docs/data-accuracy} for more information.} To assess coverage, we manually collected data on the reported US headcount in 2022 at the 32 firms in our main estimation sample. We collected these data from public company statements and SEC filings, using information from company statements about US-specific growth to estimate the US-specific headcount when only global headcount figures were available. Comparing these headcounts to the corresponding resume counts in the PDL data, we find that our dataset covers 52\% of a firm's employment in our sample, on average, with a median of 57\% and a minimum of 8\%. For Microsoft, Apple, and SpaceX specifically, we observe resumes of \revision{42\%, 29\%, and 24\% of their estimated} US workforce in 2022, respectively \revision{(see Table \ref{tab:coverage}).} Importantly, our identification strategy does not require the PDL sample to be perfectly representative. For the estimated effect to recover the corresponding effect for the underlying workforce population, the distortion induced by selection into the PDL data must be replicated by the synthetic control. Thus, even systematic over- or under-sampling of certain worker types need not bias the population contrast when the resulting selection distortion is balanced between the treated firm and its synthetic control. We characterize this requirement formally through an exact decomposition of selection bias in Appendix \ref{app:selection}.

We estimate the tenure distribution within a company during a given quarter by calculating
\begin{equation}
\mathrm{Tenure}_{ijt} \coloneqq \begin{cases} 
t - \mathrm{T_{0, ij}} & \text{if } T_{ij} > t \\ 
T_{ij} - T_{0,ij} & \text{else}
\end{cases}
\end{equation}
where $t$ indexes the end of the quarter, $T_{0,ij}$ is the start date of employee $i$ at company $j$, and $T_{ij}$ is the end date. We further sum this variable across multiple subsequent employment episodes at the same company for the same individual, to calculate total tenure for a single employment spell. This is needed because employees tend to list a single employment spell at the same company under separate resume items to indicate promotions. 

\begin{table}[ht!]
\centering
\footnotesize
\caption{Seniority Titles}
\label{tab:title_ranking}
\begin{tabular}{ccccccccccc}
\hline
\textbf{Title} & Unpaid & Training & Entry & Manager & Senior & Owner & Partner & Director & VP & CXO \\
\hline
\textbf{Ranking} & 1 & 2 & 3 & 4 & 5 & 6 & 7 & 8 & 9 & 10 \\
\hline
\end{tabular}
\floatfoot{\textit{Note}: table displays seniority levels and corresponding numeric code of employees' positions as classified in the People Data Labs data.}
\end{table}

As another proxy for seniority within the company, we use the discrete distribution of job titles in the PDL data. The lowest level is ``unpaid'' and the highest ``CXO'', corresponding to a C-suite title. The full list of titles and their corresponding numeric codes is shown in Table \ref{tab:title_ranking}. We define an employee's title within a given quarter as the highest title they held within the company in that quarter.  

\paragraph{Return to Office}

We collect data on the return to office dates of tech firms from two sources. First, we manually collect information on the month a return-to-office mandate was implemented for 55 large tech firms (e.g. Google, Apple, IBM) based on news reports and leaks on a popular anonymous employee forum, Blind. Firms were selected based on market capitalization and availability of public information on their RTO policies. Of the firms considered, 19 implemented an RTO mandate, while the other 36 continued to allow remote work to varying degrees as of February 2023. Here, ``RTO mandate'' denotes any degree of mandated office attendance, from 1 day a week to full-time. We depict the chronology of the return to office dates we collected in Figure \ref{fig:rto_dates}. Though this is not necessarily a representative sample, we can see that RTO implementations only picked up more rapidly from spring 2023 onwards. 

\begin{figure}[ht!]
\centering
\includegraphics[width=.9\textwidth]{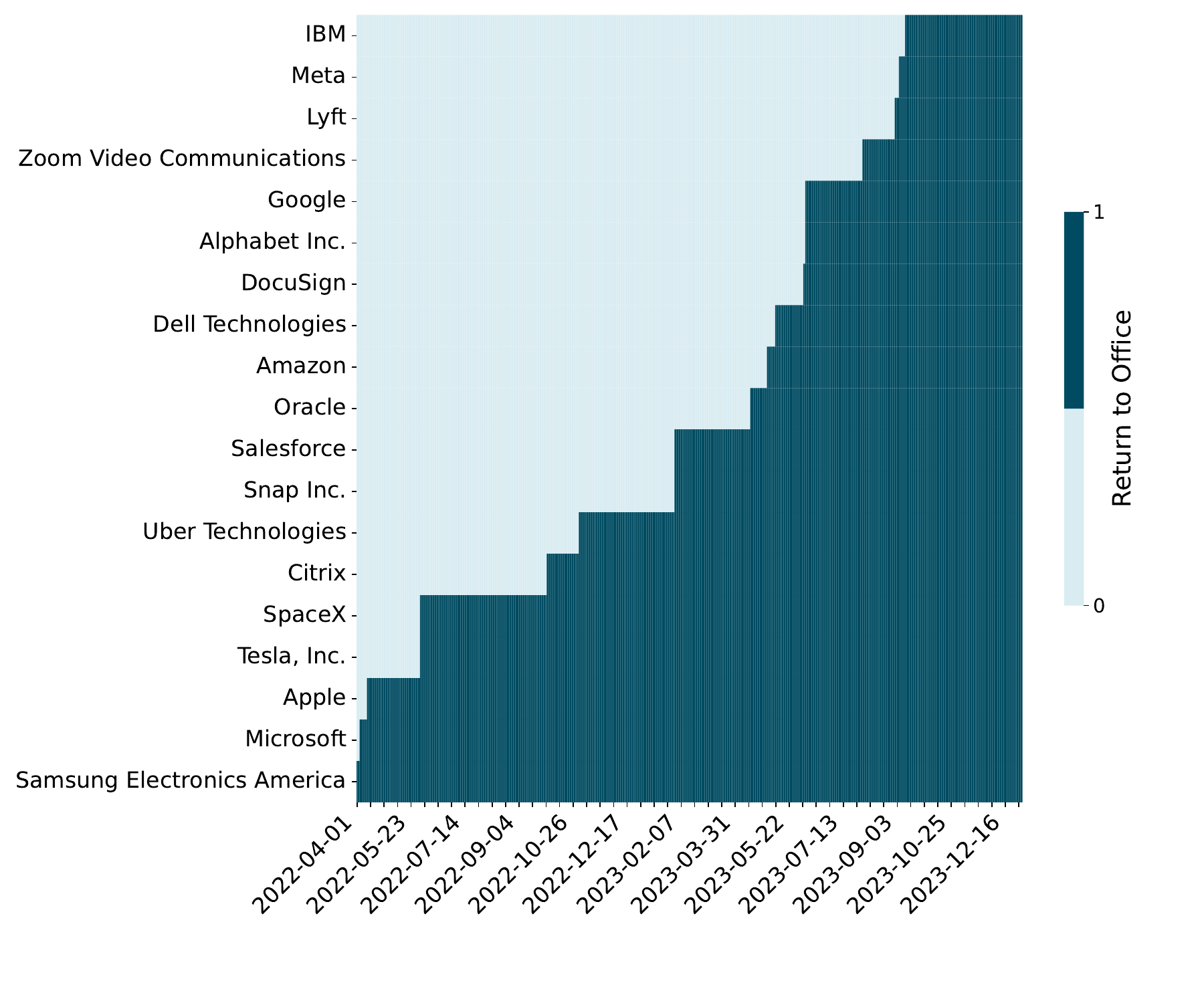}
\caption{Return to Office Dates of Large Tech Firms}
\floatfoot{\textit{Note}: figure shows return to office dates of major tech firms, ordered chronologically at a frequency of 7 weeks. Dark blue color indicates the firm has implemented an RTO mandate.}
\label{fig:rto_dates}
\end{figure}

Second, we supplement this list with data scraped from the Flex Index by Scoop Technologies, a hybrid working platform (\url{https://www.flex.scoopforwork.com/}). At the time we scraped them (Feb 2024), these data provided information on the remote work arrangements of more than 8,000 companies. While they do not contain information on the date these companies returned to work (if at all), we leverage these data to expand our set of ``never-treated'' companies that never returned to the office. These serve as our control units, together with companies that were ``not-yet-treated'' until 3 months after the end of our sample period. We were able to match these data directly to the PDL data based on the company name and website domain, as both fields had accurate coverage in the two datasets. In Table \ref{tab:control_firms}, we list all the control firms used in our main estimation sample for Microsoft, and their corresponding RTO policies at the time of writing. The table also provides verification of these policies against two independent sources, to make sure the Scoop data is accurate: 1) the RTO dates collected by \citet{flynn2024determinants} from Factiva for the Russell 3000;\footnote{We thank the authors for sharing their data with us.} 2) publicly available news reports, for which we provide links. That way, we provide independent verification of the Scoop data from at least one source for every control firm in our main sample.

\paragraph{Tech Layoffs}

The tech sector in the United States has experienced persistent rounds of large-scale layoffs since late 2022 that continued well into 2024. As these layoffs were in general independent of companies' RTO mandates, they are likely to confound our estimates. Of course, an employee's refusal to comply with an RTO mandate could affect their likelihood of being laid off, but this does not affect the independence of the \textit{implementation} of both policies. To address this concern, we leverage data from \url{layoffs.fyi}, a website that has tracked the near-universe of tech sector layoffs since the start of the COVID-19 pandemic. The website aggregates information on layoffs from news reports, company announcements, and crowdsourced employee reports. We match these data to the PDL data with a combination of fuzzy string matching and manual checking to ascertain we capture all layoffs in the data. 

\paragraph{Sample Construction} 

We begin by identifying the relevant set of ``treated units'' used in the main estimation. To avoid concurrent shocks during the estimation window, we restrict the set of firms to companies without layoffs during the two quarters before and one quarter after their RTO mandate. This restriction yields three firms: Microsoft, SpaceX, and Apple. We emphasize that the synthetic control method used in this paper, introduced below, is fundamentally a case-study approach, which emphasizes internal over external validity. As such, it is natural to focus on a small number of influential firms for which we can validly identify causal effects (see the discussion in the introduction). 
To that extent, note that all three companies implemented their RTO mandate relatively early, before both layoffs and the larger return to office debate heated up in the American tech sector. They all implemented several rounds of layoffs as well, but, as mentioned, only did so at least one quarter after their return to office. Moreover, they are all large tech firms that produce both hardware and software and employ several thousands of people. All three companies also shifted from an optional hybrid model to mandatory office attendance. The degree of office presence mandated varied, with Microsoft mandating 50\% attendance \citep{OLoughlin2022, Tilley_Cutter_2022, rosenblatt2022, financialpost2024}, Apple requiring 1 day a week \citep{OLoughlin2023, lyons2022}, and SpaceX demanding full-time attendance (40 hours per week) \citep{mac2022}. We replicate our main results for these three companies, but focus on Microsoft throughout the paper as it has the largest numbers of resumes in our data, it was the first to implement an RTO among the three companies, and its RTO was the ``cleanest'' and most clearly announced -- with Apple postponing its increase from 1 to 3 days a week until September 2022 \citep{OLoughlin2023} and SpaceX announcing its policy through an internal email \citep{mac2022}. Finally, note that enforcement intensity of the mandates likely differed across the three companies. This is hard to know with certainty because it concerns internal company policy. However, it is immaterial to the extent that the policies constituted a credible threat when announced. Indeed, while the media reported that office attendance at Microsoft did not reach its 50\% target in the first few months after their RTO \citep{seattletimes2024}, suggesting lenient enforcement, we nonetheless estimate a significant counterfactual decrease in senior employees following the mandate.

For each of the three companies, we construct a sample of control units by retaining only those firms that 1) had no layoffs in the sample period of three quarters around the RTO; 2) either never implemented an RTO or implemented it at least one quarter after the end of the sample period; 3) fell under the three NAICS codes most commonly associated with the tech industry -- Information (51), Professional, Scientific and Technical Services (54), or Manufacturing (33) --; and 4) had at least 5\% of the total number of resumes observed for the treated unit, to guarantee comparability. Together, these restrictions leave us with a sample size of 1,141,518 employment-by-quarter observations for 32 total companies in our sample for Microsoft.\footnote{Note that \revision{LinkedIn, a subsidiary of Microsoft, is in the sample}, but it only implemented an RTO over a year after Microsoft did \citep{Stewart} and is fully independent.} An ``employment'' here means a single employment of an individual at a given company, where an employee changing job can have two employments within a quarter. 

\subsection{Descriptives} \label{sec:descriptives}

\begin{table}[htbp!]
\centering
\caption{Summary Statistics}
\label{tab:sumstat}
\small
\begin{tabular}{lccccc}
\hline\hline
Statistic & Mean & Pctl(25) & Median & Pctl(75) & Benchmark \\
\hline
Median Tenure (days) & 966 & 594 & 900 & 1,316 & 657--737\textsuperscript{a} \\
Median Title & 4.84 & 5 & 5 & 5 & 3--4\textsuperscript{b} \\
Number of Resumes & 11,787 & 3,326 & 4,554 & 10,370 & --- \\
Quarterly Turnover & 0.062 & 0.044 & 0.048 & 0.087 & 0.06--0.09\textsuperscript{c} \\
Share Female & 0.315 & 0.273 & 0.310 & 0.368 & 0.27\textsuperscript{d} \\
Share to Startups & 0.042 & 0.031 & 0.040 & 0.046 & 0.04--0.06\textsuperscript{e} \\
\hline
\end{tabular}
\floatfoot{\textit{Notes}: Statistics for the 32 firms in our estimation sample at the firm-quarter level. Median tenure: 966 days $\approx$ 2.6 years. Quarterly turnover: share of employees in quarter $t$ not observed at the same firm in $t+1$. Share to startups: share of leavers whose next employer has $<$50 employees. \\[0.3em]
\textit{Benchmarks}: \textsuperscript{a}Average tenure at major tech firms ranges 1.8--2.0 years \citep{Paysa2018}; our higher median likely reflects longer-tenured employees maintaining updated profiles. \textsuperscript{b}Mid-level engineer is modal title in tech \citep{Chip2020}. \textsuperscript{c}BLS JOLTS information sector monthly separation rates of 2--3\% imply quarterly rates of 6--9\%. \textsuperscript{d}Women comprise 26--28\% of computing occupations \citep{BLS2021,CompTIA2022}. \textsuperscript{e}Startup transition rates from \citet{OnDeck2023} and \citet{Kauffman2022}.}
\end{table}

The data reveal several descriptive patterns that speak to our main results. In Table \ref{tab:sumstat}, we report cross-company summary statistics for our main variables of interest in the quarter before Microsoft's RTO. Specifically, we first aggregate several variables across employees within each company in our main sample and then report the dispersion across companies. This shows that the average tenure at the companies in our sample is around 2.5 years, and the title held by the median employee at most of the companies is at either the manager (4) or senior (5) level. We observe between 2,700 and 80,000 resumes at any of the firms, giving us ample statistical power to estimate the empirical tenure distribution. Moreover, the quarterly turnover is in line with estimates for the information and professional sectors from the Bureau of Labor Statistics's Job Openings and Turnover Survey (\url{https://www.bls.gov/jlt/}); while the share of female employees aligns with other estimates of around 30\% \citep{BLS2021}. Finally, 5\% of employees that left their company within the quarter of interest took up positions at startups, which we classify as firms with fewer than 50 employees. \revision{To supplement these statistics, Appendix Table \ref{tab:balance} reports observable characteristics of Microsoft, the control pool, and the synthetic control. These raw comparisons are descriptive and should not be interpreted as the relevant balance condition for the synthetic control design. Under the distributional linear model described below, the identifying condition is stated in terms of pre-treatment outcome distributions: if the treated firm’s pre-treatment tenure or seniority distribution is matched by the distributional synthetic control, the same weights recover the untreated post-treatment counterfactual distribution under the model. Thus, \revision{under this model}, separate balance on observed covariates is not required for identification once the relevant pre-treatment distributions are matched. \revision{This approach is analogous to assuming the linear factor model in \citet{abadie2010synthetic} does not contain observed covariates.} Covariate balance of the synthetic control can nevertheless be useful as an auxiliary diagnostic for the plausibility of the assumption of common dynamics implied by the linear model.}

\begin{figure}[ht!]
\centering
\begin{subfigure}[b]{0.49\textwidth}\includegraphics[width=\textwidth]{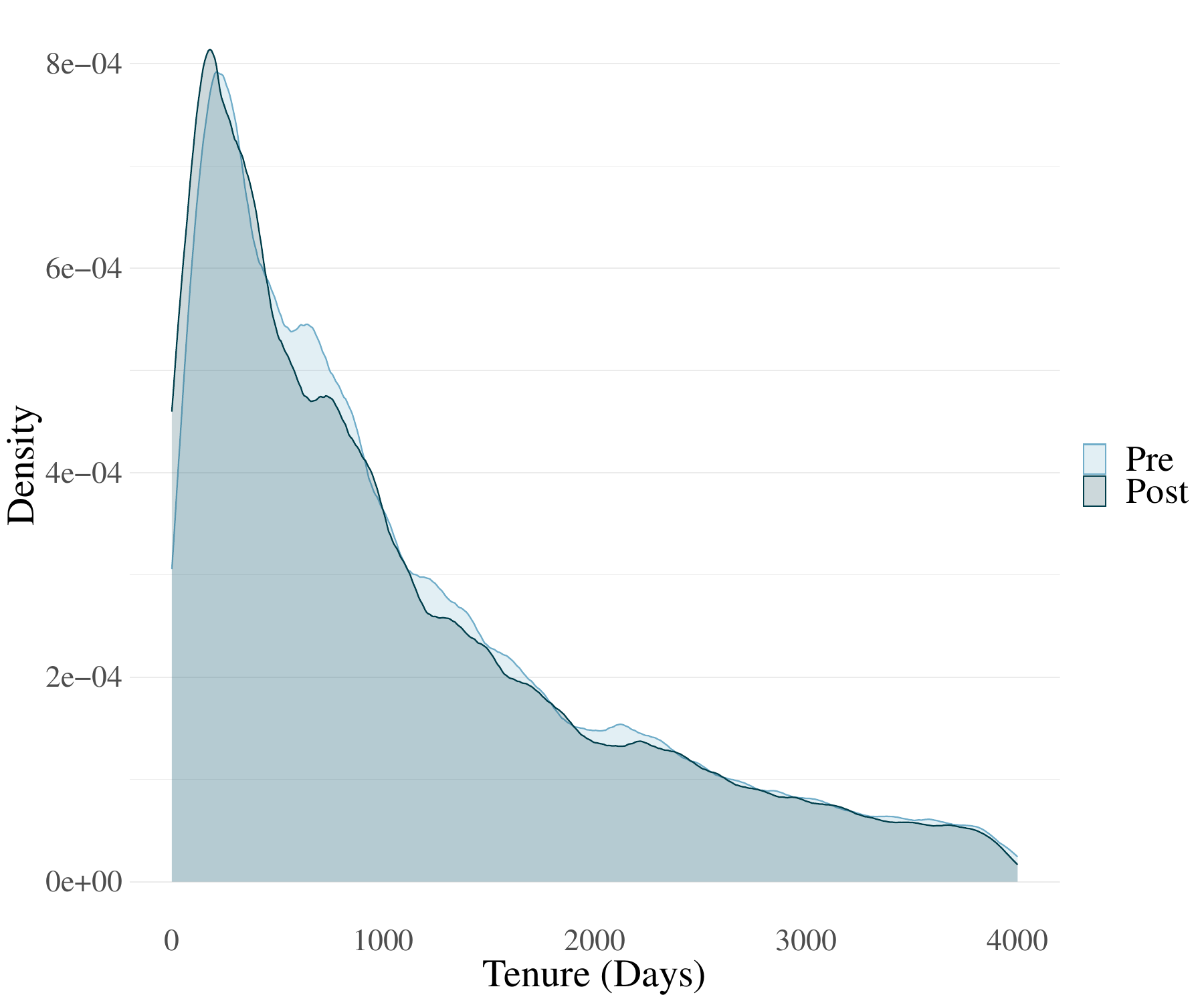}
    \caption{Tenure: Probability Density}
    \label{fig:tenure_pdf}
\end{subfigure}
\begin{subfigure}[b]{0.49\textwidth}
    \includegraphics[width=\textwidth]{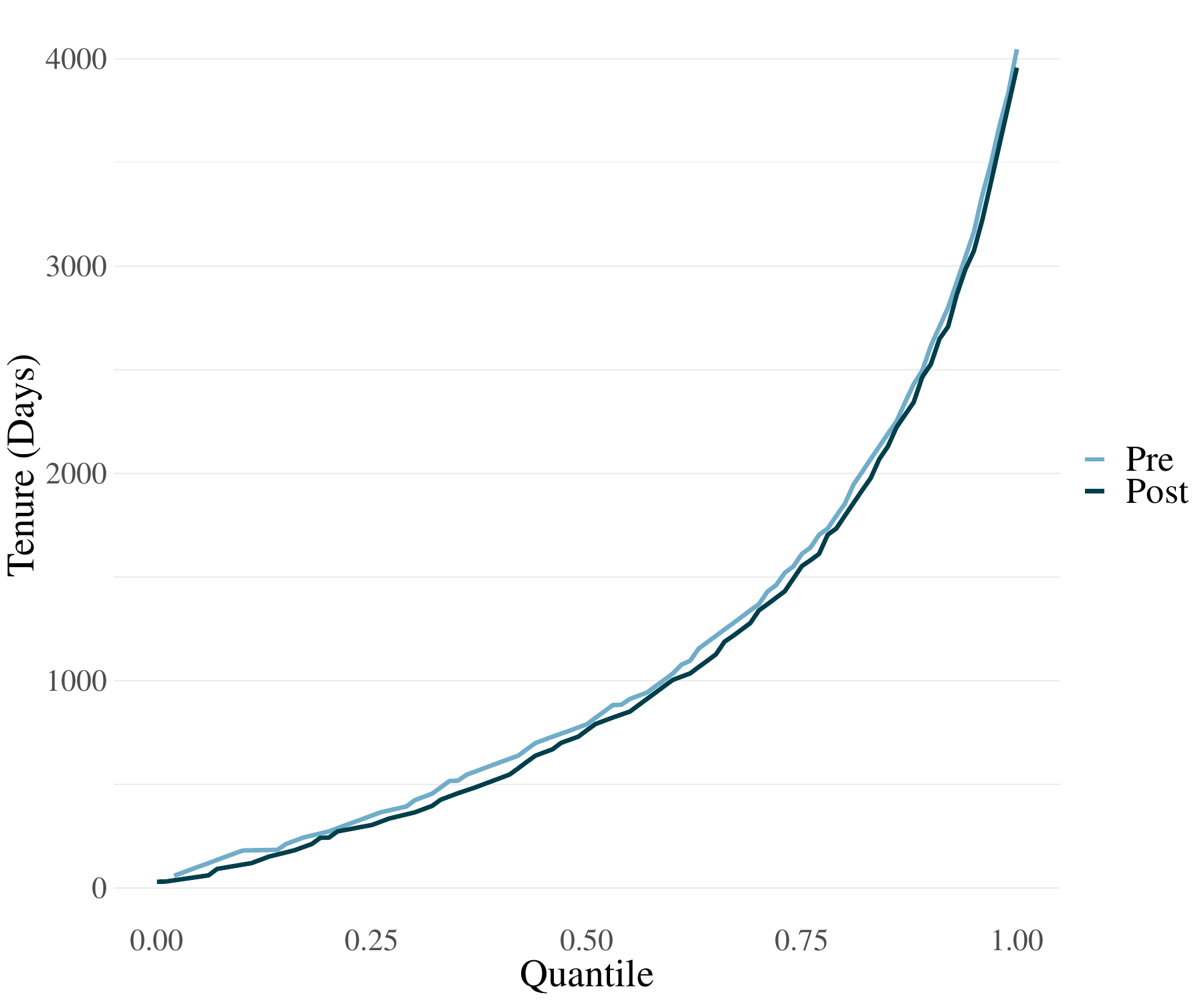}
    \caption{Tenure: Quantile Function}
        \label{fig:tenure_quantile}
\end{subfigure}
\begin{subfigure}[b]{0.49\textwidth}\includegraphics[width=\textwidth]{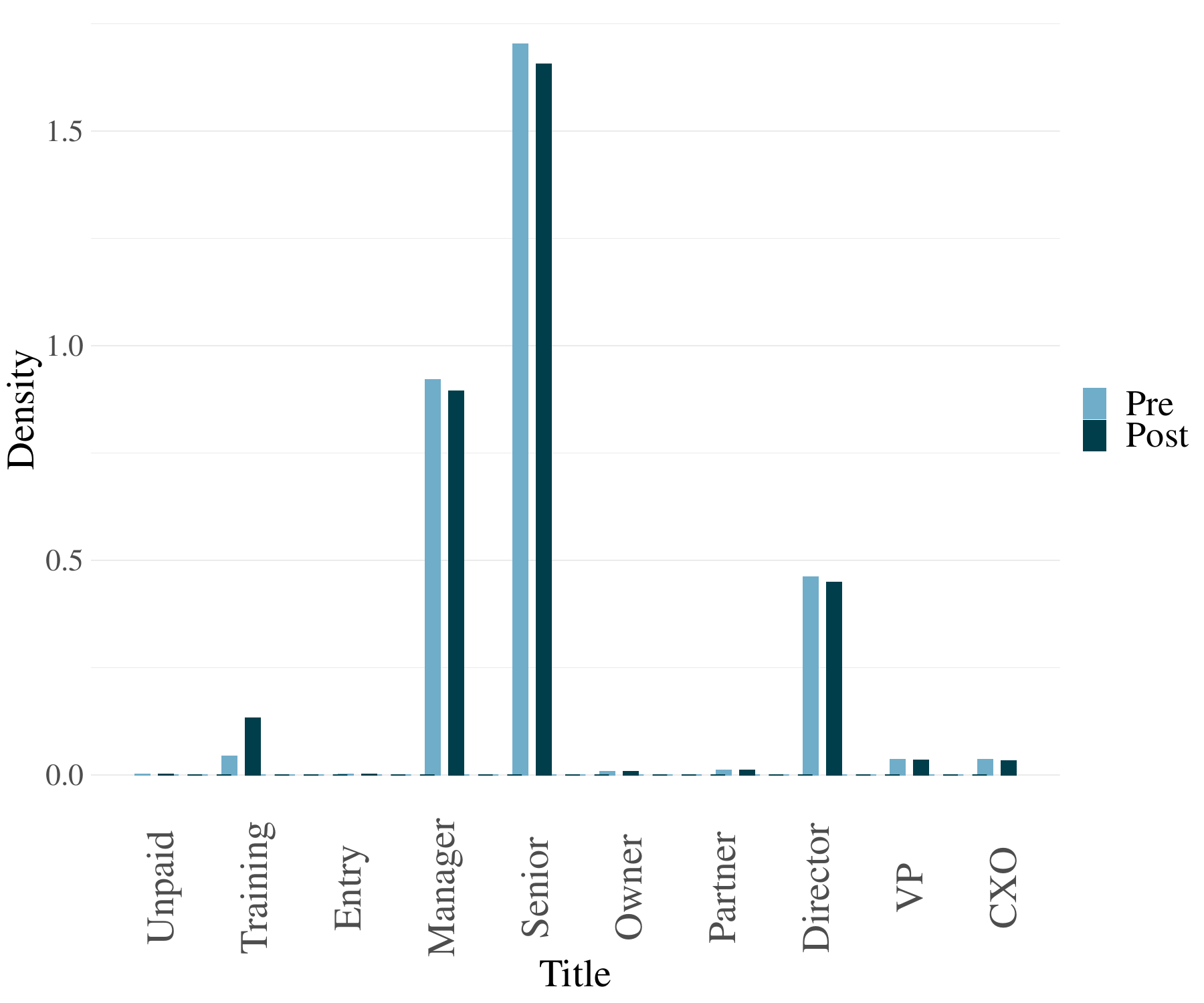}
    \caption{Titles: Histogram}
    \label{fig:title_pdf}
\end{subfigure}
\begin{subfigure}[b]{0.49\textwidth}
    \includegraphics[width=\textwidth]{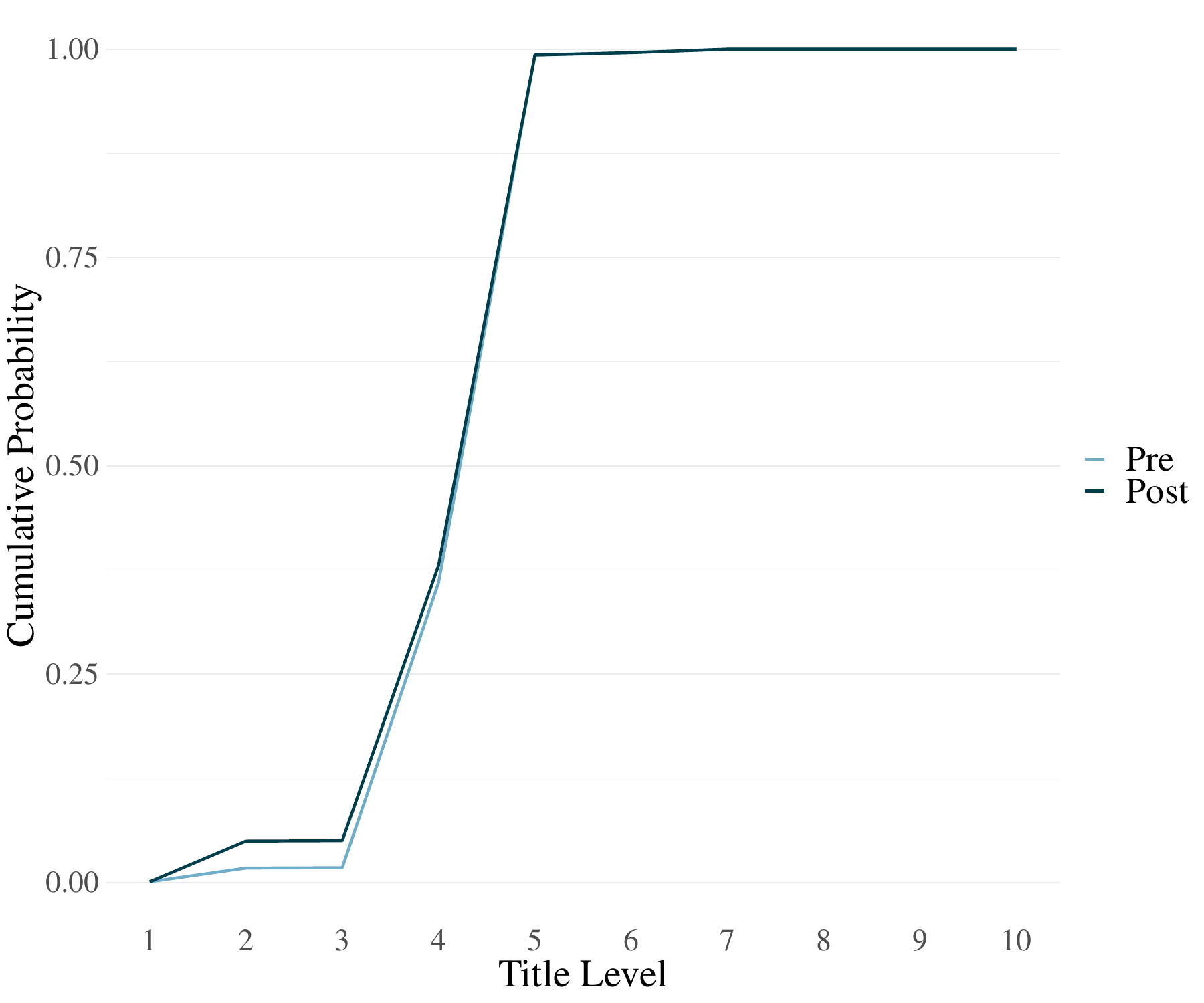}
    \caption{Titles: Cumulative Distribution Function}
        \label{fig:title_cdf}
\end{subfigure}
\floatfoot{\textit{Note}: \ref{fig:tenure_pdf} shows probability density functions (pdf) of tenure (in days) of Microsoft's employees in quarter before (``Pre'', light blue) and after (``Post'', dark blue) RTO, computed using Epanechnikov kernel with bandwidth of 90 days; \ref{fig:tenure_quantile} shows empirical quantile functions of same outcome. \ref{fig:title_pdf} shows histogram of seniority titles; and \ref{fig:title_cdf} shows cumulative distribution functions of the same, with title levels mapped to ordinal variable between 1--10 in same order. Tenure distribution is restricted to bottom 90 quantiles.}
\caption{Distribution of Tenure and Titles}
\label{fig:empirical_distributions}
\end{figure}

A preview of the main descriptive results is shown in Figure \ref{fig:empirical_distributions}. There, we display the observed distributions of our main outcome variables, seniority and tenure, in the quarter before (``Pre'') and after (``Post'') Microsoft's RTO. Looking at the tenure distribution in the top row, we can see that both the probability density function and the quantile function shift leftward after the RTO, suggesting a decrease in the total days tenured across the distribution. Similarly, in the bottom row, we observe an increase in the histogram density for employees at the trainee level and a decrease at the senior level. This corresponds to a leftward shift in the cumulative probability at seniority levels below ``manager'' in \ref{fig:title_cdf}. Overall, the patterns in the raw data already suggest a decrease in tenure and seniority at Microsoft after its RTO. Of course, these patterns are only correlational insofar as they do not account for how the distributions would have evolved in the absence of a return to office. As such, they fail to account for confounding shocks unrelated to the RTO mandate. This is the purpose of our synthetic controls approach, which we describe and estimate in the next section. 

\begin{figure}[ht!]
\centering
\begin{subfigure}[b]{0.49\textwidth}
    \includegraphics[width=\textwidth]{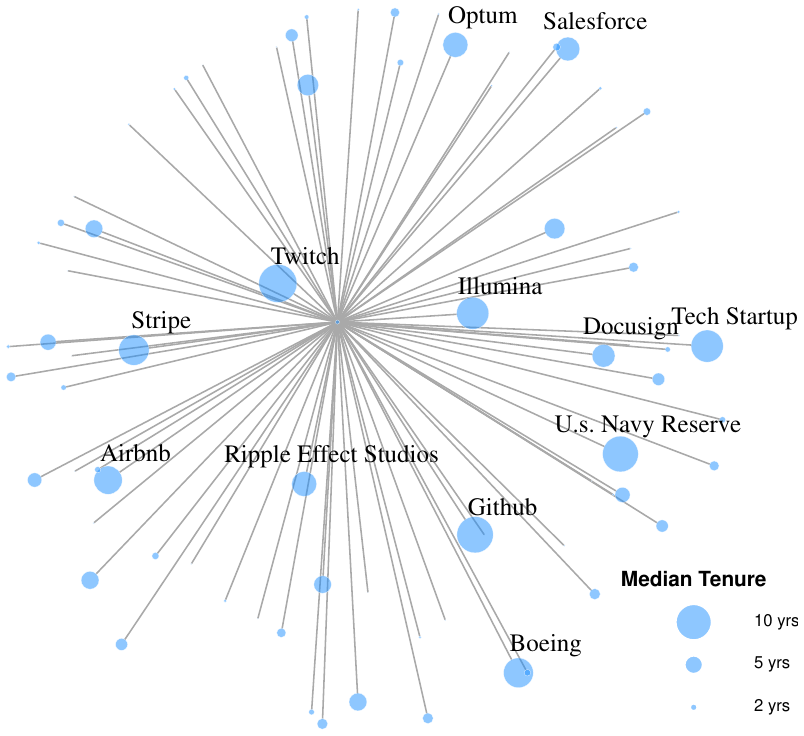}
    \caption{Pre-RTO}
        \label{fig:network_pre}
\end{subfigure}
\begin{subfigure}[b]{0.49\textwidth}
    \includegraphics[width=\textwidth]{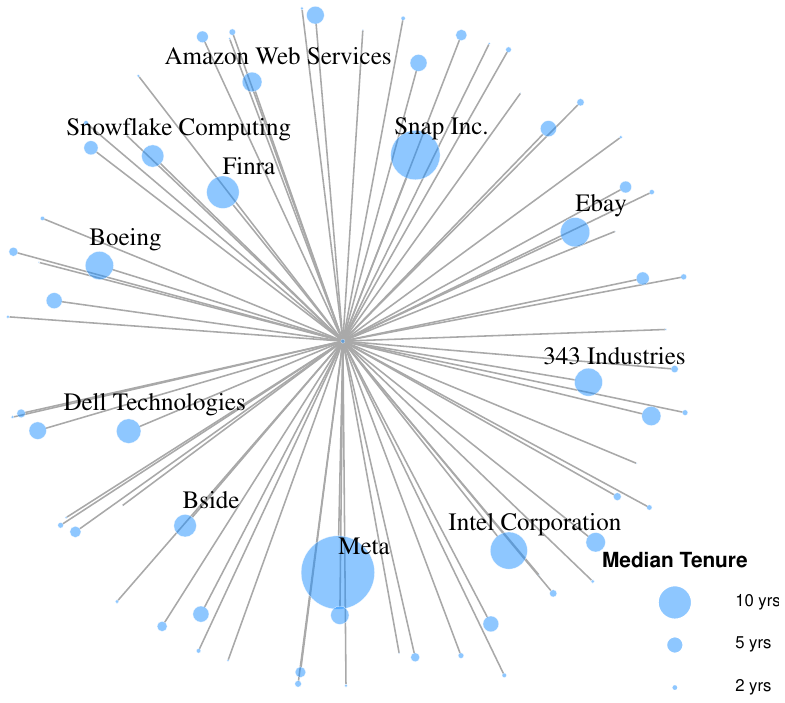}
    \caption{Post-RTO}
        \label{fig:network_post}
\end{subfigure}
\caption{Quarterly Employee Outflow Network from Microsoft, by Median Tenure}
\floatfoot{\textit{Note:} figure depicts flow of employees leaving Microsoft one quarter before and after their RTO. One vertex corresponds to one destination company, with the size of the vertex scaled according to the median tenure of the employees leaving for that company. Destination company names for vertices in top 20\% of vertex sizes indicated by labels. }
\label{fig:network}
\end{figure}

Finally, we show the network of companies to which Microsoft employees left in Figure \ref{fig:network}, for the quarter before and after its RTO. Each vertex indicates a company to which an employee moved. The vertex sizes correspond to the median tenure the newly hired employees had accumulated at Microsoft before leaving. Thus, the networks draw a picture of the outflow of experience from Microsoft to \revision{other large firms in related industries}. That way, one can observe a marked increase in the median tenure of employees hired at several large tech companies that compete directly with Microsoft for labor, such as Meta, Snap Inc., and Intel, \revision{as well as in several other companies that are less obviously direct competitors such as Boeing}. In other words, after Microsoft's RTO, several senior employees appear to have moved to \revision{other large, tech or tech-adjacent firms} that did not have an RTO mandate in place, \revision{some of which were direct competitors of Microsoft}. We find causal support for these descriptive patterns by estimating a counterfactual decrease in employees' tenure and seniority at Microsoft, as well as an increase in the share of employees departing for larger companies instead of startups. 

\subsection{Estimation}

The descriptive patterns documented above can only provide correlational evidence, as they do not account for counterfactual changes that would have affected Microsoft's employee distribution in the absence of an RTO. In particular, one could expect companies to implement an RTO mandate because of pre-existing issues with employee retention, leading to endogeneity concerns. 
In order to estimate the counterfactual distribution, we use the distributional synthetic controls estimator from \citet{gunsilius2023distributional}. The idea is to extend the classical synthetic controls method \citep{abadie2003economic, abadie2010synthetic, abadie2021using} to settings where the policy change of interest is implemented at an aggregate level, but the researcher has access to disaggregated data. This is exactly our setting, where the level of intervention is at the company level, but the heterogeneity of the treatment is of fundamental interest. The distributional synthetic controls estimator is ideally suited for this setting as it estimates the heterogeneous treatment effect for the entire company and not an individual treatment effect for each unique worker. In particular, the method does not make the common assumption that individuals can be tracked over time. Relaxing this assumption allows for changes in the composition of workers due to both entry and exit into the company, which is a crucial aspect of within-firm employee distributions.

The classical synthetic controls method was developed for observational settings where treatment is not randomly assigned to units. Furthermore, the pool of control units is sufficiently dissimilar from the treated unit that simple before-and-after comparisons of outcome trends, as in the difference-in-differences design, are hard to defend. The method proposes to solve these challenges to causal inference by closely replicating the behavior of the treated unit prior to treatment. This becomes feasible when the pool of control units is sufficiently large. This is exactly the setting considered in this paper, where RTO policies are non-random, strategic decisions by firms, and employment trends at firms that do and do not return to the office are likely to differ. The argument is that, if one can construct a ``synthetic Microsoft'' that replicates the employment trends at Microsoft well enough in the months before its RTO, one can use it as a ``synthetic control'' unit afterward.

We briefly review the idea behind the distributional version of synthetic controls. For more detailed information, we refer to \citet{gunsilius2023distributional} and the R implementation \citep{discosR}. Denote by $F_{jt,N}(y)$ the cumulative distribution function of the outcome of interest $Y$---in our case seniority or tenure---for company $j$ at time $t$ in the absence of the intervention. The distribution in the presence of intervention is denoted by $F_{jt,I}$. We write $t^*$ for the last time point before the policy's implementation. By construction, the observed distribution $F_{jt}$ coincides with $F_{jt,N}$ for all $t\leq t^*$. We denote the corresponding quantile function by 
\[F^{-1}_{jt}(q) = \inf_{y\in \mathbb{R}}\left\{F_{jt}(y)\geq q\right\},\quad q\in(0,1).\] In our setting the distributions $F_{jt}$ are estimated from individual level data $Y_{i(t)jt}$, where $i(t)$ is the individual $i$ at company $j$ at time $t$. Individual observations are only used to estimate the distributions $\hat{F}_{jtn}$ and corresponding quantiles $\hat{F}^{-1}_{jtn}$, where $n$ is the number of cross-sectional observations in time period $t$, which is allowed to vary with both $t$ and $j$.

In the synthetic controls setting, we pick a target company $j=0$ that at a point $t^*+1$ implements the RTO mandate. We focus on Microsoft as our target company but replicate our main results for SpaceX and Apple. We only use companies as potential controls that had not implemented a form of RTO three months after the end of the sample period, or equivalently six months after the target company's RTO. Our goal is to estimate the counterfactual distribution $F_{0t,N}$ for time periods $t > t^*$ to see what would have happened to the outcome distribution of the target company had it not implemented the RTO mandate. 

We have two main outcomes of interest, tenure at the company in number of days and seniority in terms of an employee's title. We treat tenure as a continuous variable. Therefore, we estimate the counterfactual tenure distribution by obtaining the optimal weights $\vec{\lambda}^*_t = \left(\lambda_{1t},\ldots, \lambda_{Jt}\right)$ in each time period $t\leq t^*$, by solving 
\[\vec{\lambda}^*_t = \argmin_{\lambda\in \Delta^J} \int_{0}^1 \left\lvert \sum_{j=1}^J \lambda_{jt} F^{-1}_{jt}(q) - F^{-1}_{0t}(q)\right\rvert^2 \dd q.\]
From the resulting time-specific weights, we form one set of time-averaged weights $\vec{\lambda}^* = \frac{1}{t^*}\sum_{t=1}^{t^*} \vec{\lambda}_t^*$, and construct the counterfactual quantile function for the time periods $t>t^*$ via 
\[F^{-1}_{0t,N}(q) = \sum_{j=1}^J \lambda^*_j F^{-1}_{jt}(q).\] Here, $\Delta^J = \left\{(\lambda_{1t},\ldots,\lambda_{Jt})\in\mathbb{R}^J: \lambda_{1t},\ldots, \lambda_{Jt}\geq0, \quad \sum_{j=1}^J\lambda_j=1\right\}$ is the probability unit simplex in $\mathbb{R}^J$. In this sense, the distributional version is the perfect analogue of the classic synthetic controls estimator \revision{of \citet{abadie2003economic}}, which also generates one set of time-invariant weights by averaging. The other outcome variable is seniority, which is a discrete variable with $10$ ordinal values. Since the outcome is ordinal and not cardinal, it makes sense to use mixtures of distributions instead of mixtures of quantiles, as distributions preserve the ordinal structure. Therefore, in line with the arguments in \citet{gunsilius2023distributional}, we compute the optimal weights \revision{$\vec{\lambda}^{*,\mathrm{title}}_t$} in this case via
\begin{equation} \label{eq:main_discrete} \revision{\vec{\lambda}^{*,\mathrm{title}}_t} = \argmin_{\lambda\in \Delta^J} \int_{\mathbb{R}} \left\lvert \sum_{j=1}^J \lambda_{jt} F_{jt}(y) - F_{0t}(y)\right\rvert \dd y,\end{equation} and construct the counterfactual distribution in $t>t^*$ via
\[F_{0t,N}(y) = \sum_{j=1}^J \revision{\lambda^{*,\mathrm{title}}_j} F_{jt}(y).\] We also use this mixture of distributions approach to estimate the counterfactual share of leavers with a certain characteristic in Section \ref{sec:where_do_they_go}. 

To obtain an unbiased causal effect of RTO on the tenure distribution using the synthetic controls, we rely on the causal model for distributional synthetic controls introduced in \citet[Appendix A]{gunsilius2023distributional}. The model assumes that the outcome of an individual $i$ at company $j$ aƒt time $t$ can be written as $Y_{ijt,N} = \alpha_t + \beta_t U_{ijt}$, for $U_{ijt} = \gamma_t + \delta_t U_{ij(t-1)}$, where $i$ indexes the individual, $j$ indexes the company and $t$ indexes time. \revision{The latent variable $U_{ijt}$ summarizes worker- and firm-level determinants of the potential outcome absent the RTO, such as accumulated firm-specific human capital, career stage, and latent retention propensity. The parameters $\alpha_t$ and $\beta_t$ allow the mapping from this latent component to the observed outcome to vary over time, while $\gamma_t$ and $\delta_t$ describe the common evolution and persistence of the latent component.} This model relies on the linearity assumption to identify the treatment effect---analogously to the classic synthetic controls model \citep{abadie2010synthetic}\revision{, but through a distributional reduced-form structure rather than a linear factor model.} \revision{In our setting, the linearity assumption is a tractable approximation over the short estimation window of a few quarters. The latent component captures persistent characteristics such as firm-specific human capital and match quality that plausibly evolve smoothly over short horizons. Moreover, while some such structure is necessary---pre-treatment fit alone does not identify the counterfactual---exact linearity is sufficient but not the only possibility, and weaker distributional factor structures can also suffice. For a} factor model approach at the quantile level for distributional synthetic controls, see also \citet{zhang2024asymptotic}.

\revision{Unlike approaches based on unconfoundedness, this does not require treatment to be independent of potential outcomes given observed covariates: RTO adoption may depend on a firm's unobserved latent type, so long as the dynamics governing it are common across units.} Moreover, unlike changes-in-changes \citep{athey2006identification}, this approach identifies counterfactual distributions at the company level $j$ rather than the individual level $i$, which means the method works with repeated cross-sections where the individuals $i$ observed within each unit $j$ can vary over time. \citet{gunsilius2023distributional} contains a formal proof that the causal effect in this model is identified from the proposed distributional synthetic controls estimator \revision{under the condition that it replicates the pre-treatment treated unit's distribution.} 

To construct confidence intervals, we rely on a version of the bootstrap. The bootstrap algorithm we
use as well as conditions that guarantee its uniform convergence, which we prove as a novel
theoretical contribution, are provided in \revision{Appendix} \ref{app:bootstrap}. It is different from existing results, since we only prove that the bootstrap works almost everywhere under the assumption that the quantile functions are sampled from some random measure that is absolutely continuous with respect to a Gaussian measure. While this assumption is unconventional in the classical asymptotic theory, it is reasonable in our setting since we work with entire quantile functions. In fact, this allows us to address ideas of randomness at the level of the quantile function, which fits our functional data framework. We report uniform confidence bands throughout to account for simultaneous inference. Monte Carlo simulations, reported in Appendix \ref{app:mc}, validate that the confidence bands achieve uniform nominal coverage, barring some slight undercoverage that is common for functional data \citep[e.g.][]{degras2011simultaneous, chernozhukov2014gaussian}, for a data-generating process tailored to mimic the empirical application, and converge approximately at rate $\sqrt{n}$, supporting the theoretical results. For the continuous quantile-function estimator, we report uniform confidence bands to account for simultaneous inference. We also report analogous bootstrap bands for the discrete CDF-mixture analyses, while the formal validity result in Appendix \ref{app:bootstrap} applies to the quantile-function estimator.

Finally, we perform inference via a permutation test, analogous to the classical setting \citep{abadie2010synthetic}. The intuition is the same as in the classic case: we apply the same synthetic control method to each control unit, pretending it is the treated unit. If there is a treatment effect, the actual treated unit should be among the units with the most dramatic post-treatment change compared to these 'placebo' tests. Since we work with distributions instead of vectors, we rely on the same distance we used to replicate the pre-treatment trends: the Wasserstein distance. To account for suboptimal fits in the replications of the pre-treatment distributions, we augment the proposed test in \citet{gunsilius2023distributional} and construct a novel permutation test as the ratio of post-to-pre-intervention distances following \citet{abadie2010synthetic, abadie2021using} by computing the statistic
\begin{equation} 
r_j = \frac{R_j(t^*+1,T)}{R_j(1,t^*)} \label{eq:permutation}
\end{equation} 
when treating company $j$ as the target and calculating the $p$-value for the permutation test as
\[p=\frac{1}{J+1}\sum_{j=0}^J H(r_j-r_0),\] where $H(x)$ is the Heaviside function, which takes the value $1$ if $x\geq0$ and $0$ if $x<0$. Here,
\[
R_j(t_1,t_2) = \left( \frac{1}{t_2-t_1+1} \sum_{t=t_1}^{t_2}d_t^2 \right)^{1/2}
\]
is the root mean squared prediction error in the distance $d_t$. For the continuous quantile-mixture estimator, we use the $2$-Wasserstein distance
\[
d_t^2 = \int_{q_{\min}}^{q_{\max}} \left| \sum_{j=1}^J\lambda_j^*F^{-1}_{jt}(q) - F^{-1}_{0t}(q) \right|^2 \dd q.
\]
For the discrete CDF-mixture estimator, we instead use the corresponding $1$-Wasserstein distance
\[
d_t = \int_{\mathbb{R}} \left| \sum_{j=1}^J\lambda_j^*F_{jt}(y) - F_{0t}(y) \right| \dd y,
\]
which is the same distance used to construct the synthetic control in
\eqref{eq:main_discrete}.

The two inference approaches are complementary. The permutation test provides an overall test of causal significance under the sharp null hypothesis that no unit is affected by treatment; it is nonparametric and valid in finite samples under exchangeability of treatment assignment across units. The bootstrap confidence bands, by contrast, are for analyzing the purely statistical randomness induced by the fact that in practice the quantile functions have to be estimated themselves. For the continuous quantile-function estimator, their validity relies on asymptotic arguments based on Hadamard differentiability of the synthetic controls functional; see Appendix \ref{app:bootstrap}. These arguments require the number of observations within each firm-quarter to be large.

\section{Results} \label{sec:results}
\subsection{The causal effects of an RTO on tenure and seniority}
Using the method described above, we estimate synthetic counterfactual tenure and title distributions for Microsoft. For the tenure distribution, we restrict the distribution to the bottom 90 quantiles as the top 10 quantiles are very sparse, leading to unstable estimates. In Table \ref{tab:weights}, we report the top 10 largest control weights for each of the two outcomes of interest separately, together with the name of the corresponding control unit. Most of the companies receiving the largest weights in both outcomes of interest are highly similar to Microsoft, being large diversified tech firms that produce both hardware and software, such as Amazon, Dell, Motorola, and Cisco. The weights differ somewhat across outcomes because different control firms may better approximate Microsoft's tenure distribution versus its title distribution. This is not a concern for identification: each outcome has its own synthetic control constructed to match that specific outcome in the pre-period. Despite these weight differences, the qualitative conclusions are consistent across both outcomes. 
\begin{table}
\centering
\begin{subtable}{0.45\textwidth}
\begin{tabular}{lr}
  \hline
Company & Weight \\ 
  \hline
Amazon & 0.1925 \\ 
  Autodesk & 0.1282 \\ 
  Dell Technologies & 0.1098 \\ 
  Slalom Consulting & 0.0884 \\ 
  Motorola Solutions & 0.0859 \\ 
  3m & 0.0805 \\ 
  Cisco & 0.0674 \\ 
  Docusign & 0.0497 \\ 
  Citrix & 0.0364 \\ 
  Cox Automotive Inc. & 0.0328 \\ 
   \hline
\end{tabular}

\caption{Tenure}
\end{subtable}
\begin{subtable}{0.45\textwidth}
\begin{tabular}{lr}
  \hline
Company & Weight \\ 
  \hline
Protiviti & 0.3185 \\ 
  Linkedin & 0.2054 \\ 
  Intuit & 0.1753 \\ 
  Dell Technologies & 0.1547 \\ 
  Cisco & 0.0567 \\ 
  Deloitte & 0.0469 \\ 
  Intel Corporation & 0.0422 \\ 
  Hp & 0.0003 \\ 
  Broadridge & 0.0000 \\ 
  Nvidia & 0.0000 \\ 
   \hline
\end{tabular}

\caption{Titles}
\end{subtable}
\caption{Top 10 Synthetic Controls Weights: Microsoft}
\floatfoot{\textit{Note}: table shows names of 10 companies with largest synthetic controls weights for Microsoft's RTO results shown in Figures \ref{fig:main_tenure} and \ref{fig:main_titles}.}
\label{tab:weights}
\end{table}

With these weights in hand, we construct a distributional synthetic control for Microsoft. We find that the implementation of a return-to-office mandate in April 2022 led to \revision{ a significant shift in the composition of employee departures}, with more senior staff leaving at higher rates. Figure \ref{fig:main_tenure} depicts this through the counterfactual changes in employee tenure. 

\begin{figure}[htbp!] 
\begin{subfigure}[b]{0.49\textwidth}
\includegraphics[width=\textwidth]{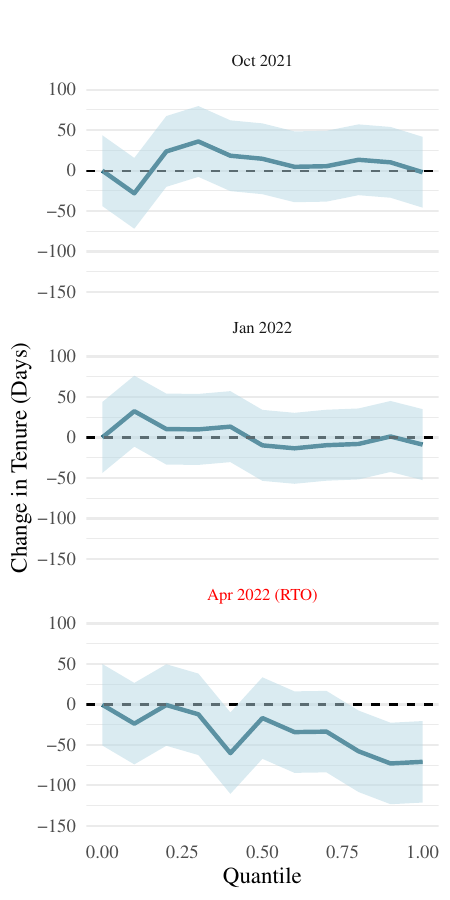}
\caption{Continuous}
\end{subfigure}
\begin{subfigure}[b]{0.49\textwidth}
\includegraphics[width=\textwidth]{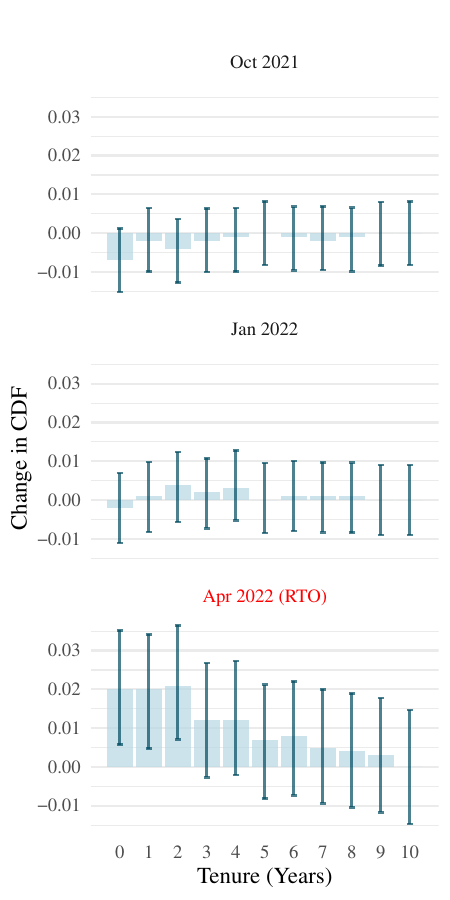}
\caption{Discretized}
\end{subfigure}
\caption{Tenure at Microsoft After RTO}
\floatfoot{\textit{Note}: \textbf{(a)}: counterfactual change in quantile functions of tenure in days at Microsoft for two quarters before its RTO and one quarter after; \textbf{(b)} counterfactual change in CDFs of discretized version of the tenure variable, measured in years of tenure the employee has completed in a given quarter, for two quarters before and one quarter after Microsoft's RTO. Facet titles indicate first month of respective quarters. Estimated change is depicted by dark blue solid line and 95\% bootstrap bands (uniform confidence intervals) by light blue shaded area (1,000 bootstraps). The formal validity results in Appendix \ref{app:bootstrap} correspond to the quantile case (a). Test p-values are (see \eqref{eq:permutation}): left panel \revision{0.06 (0.03 for quantiles $[0.5,1]$) }, right panel 0.06. N=1,141,518, J+1=32.}
\label{fig:main_tenure}
\end{figure}

For the left panel, the y-axis captures the difference between the observed distribution of tenure in number of days and the counterfactual one, estimated using the distributional synthetic controls method. The x-axis represents quantiles of these distributions. The synthetic control can accurately reconstruct the pre-RTO tenure distribution, as the difference between the two distributions is close to and not significantly different from zero at all quantile grid points. Following the return-to-office (RTO) mandate, we estimate a reduction in tenure among Microsoft's employees that escalates with tenure duration, and only becomes statistically significant for the top 2 deciles. With pointwise confidence bands, this becomes the top 5 deciles (not reported). In terms of magnitudes, tenure at Microsoft decreased by two months at the top deciles. Our permutation test supports the validity of this inference, with a p-value of \revision{0.06 (0.03 for the quantiles above the median).}

For an alternative perspective, we discretize the tenure measure into years and estimate the counterfactual change in the CDF of this discrete variable using the mixture of distributions approach in \eqref{eq:main_discrete}. Thus, we can interpret the y-axis as capturing the increase in probability mass at a given year of tenure. The resulting estimates are shown in the right panel of Figure \ref{fig:main_tenure}. We replicate the pre-treatment distributions nearly perfectly. After treatment, we estimate a 2\% increase in the probability mass of employees that have 0 to 3 years of tenure, corresponding to a relative decrease in the share of employees with longer tenure. Again, the permutation p-value of 0.06 supports the validity of the inference. We note that the ends of the support for all categorical variables  in this paper have a 0\% change in the CDF by construction, but nonetheless have confidence bands around them due to the nature of uniform confidence bands.

In Figure \ref{fig:main_titles}, we consider a different measure of the employee distribution: seniority. The figure depicts an increase in the mass of workers with lower-ranking titles at Microsoft after its RTO, using the mixture of distributions approach in \eqref{eq:main_discrete}. We estimate an increase in the mass of employees in training, entry, and manager positions of approximately 1.7\%, with the corresponding bootstrap bands excluding zero. The estimator replicates the pre-treatment distribution well, with the bootstrap bands including zero everywhere. We obtain nearly identical results when restricting the distribution to only the senior level and below, where most of the mass is concentrated. The permutation test once again supports the validity of the inference, with a p-value of 0.03. While the CDF is the fundamental object of analysis in our categorical setting, we also report these results for changes in the corresponding PMF in Figure \ref{fig:title_pmf}, which may be easier to interpret. An additional benefit is that the PMF better captures the decreases in mass at the top of the title distribution, for which the bootstrap bands also exclude zero for SpaceX. Additionally, these results are nearly unchanged when pooling the Owner, Partner, and Director categories together in Figure \ref{fig:title_robustness_pooled}, to account for the fact that not all firms may have Owner or Partner positions.

\begin{figure}[ht!]
\includegraphics[width=0.9\textwidth]{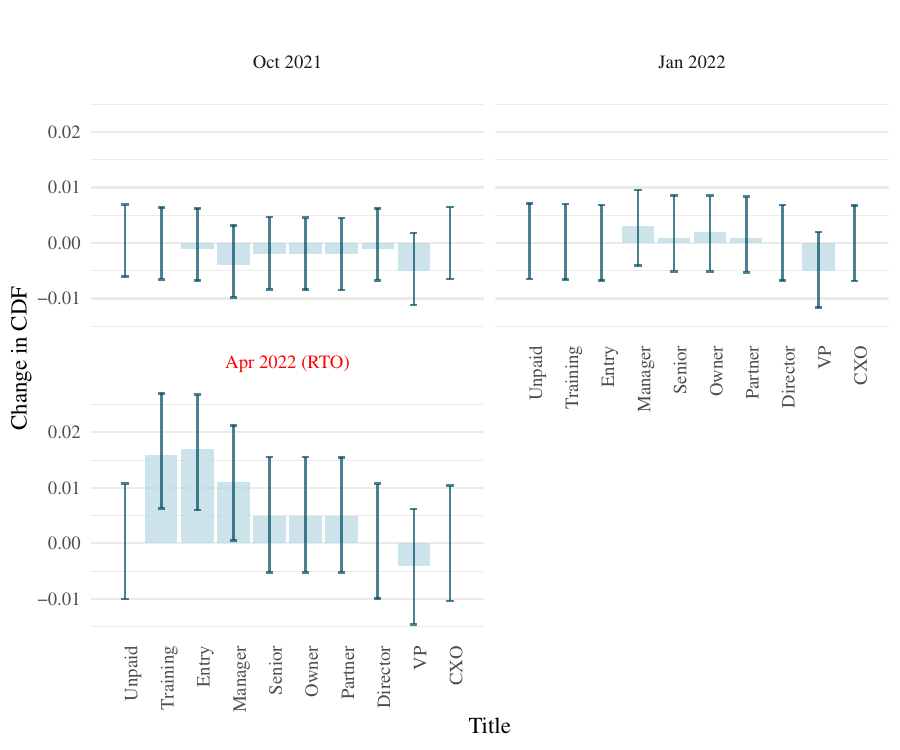}
\caption{Titles at Microsoft After RTO}
\floatfoot{\textit{Note}: figure shows counterfactual change in cumulative distribution functions of title distribution at Microsoft for two quarters before its return to office (RTO) and one quarter after. Facet titles indicate first month of respective quarters. Estimated change is depicted by light blue solid bar and 95\% bootstrap bands by dark blue whiskers (1,000 bootstraps). Permutation test p-value (see \eqref{eq:permutation}) is 0.03. N=1,141,518; J+1=32.}
\label{fig:main_titles}
\end{figure}

We replicate these results almost identically for Apple and SpaceX in Figure \ref{fig:title_robustness}. The result for the continuous tenure distribution does not replicate, likely because we observe fewer resumes at these companies, which makes it harder to accurately estimate the continuous tenure distribution. This explanation is supported by the bad pre-treatment fit of the tenure distribution for these companies (not reported). However, when we discretize it into 5-year bins \revision{in Figure \ref{fig:tenure_robustness}}, we estimate an increase in the mass of workers with less than 5 years of tenure of 1.5\%. To replicate the result for the title distribution, we focus on the titles below the senior level at these companies, as there is very little mass at the higher titles.\footnote{The results for Microsoft replicate identically when restricting the distribution in this way.} Interestingly, the magnitude of the results appears to increase with the stringency of the RTO, with Apple seeing a leftward shift in mass of up to 2\% and no significant effects for managers, while up to 5\% of the SpaceX distribution shifts leftward. This correlates with these companies mandating 1 and 5 days of office attendance each week, respectively, compared to Microsoft's 50\%. 

Several points merit further discussion. First, the distribution of some of our control units may be affected by spillover effects, if one of the employees leaving Microsoft takes up employment at a control unit. We do not believe this will meaningfully bias our estimates, as only 5\% of employees leave a company in a given quarter, on average. Of those, we typically see only a handful of these leaving for the same, usually large, firm. Furthermore, we restricted our control sample to companies with at least 5\% of the number of employees observed at Microsoft, so the number of employees joining any single control unit will be trivial relative to that unit's total number of employees. Together, these facts imply that any bias due to potential spillover effects to control units is likely negligible. To verify this claim, we track where Microsoft employees go after leaving the company post-RTO. Of all Microsoft leavers in the post-treatment quarter, only 78 employees (0.16\% of Microsoft's workforce) joined any of our 31 control firms. The maximum share of any control firm's workforce that came from Microsoft was 0.30\% (Docusign), and no control firm received more than 1\% of its workforce from Microsoft (Table \ref{tab:spillover_magnitude}). Moreover, any such spillovers would bias our estimates toward zero: if senior Microsoft employees join control firms, this mechanically increases the average tenure at those firms, making the synthetic control's tenure distribution shift rightward and reducing the estimated treatment effect. As a bounding exercise, we re-estimate our main specification excluding the five control firms that received the most Microsoft employees post-RTO. The permutation test $p$-value \revision{remains stable with Microsoft ranked 2nd}, confirming that our results are robust to potential spillover effects (Table \ref{tab:spillover_robustness}).

Second, there could be anticipation effects that bias the magnitude of our estimates. For all three companies we consider, however, the RTO mandates were announced less than 2 months prior to their implementation.\footnote{Microsoft announced plans for a phased return to the office in early February 2022 with mandatory office presence from April 2022 onwards \citep{Tilley_Cutter_2022}. Apple announced its RTO mandate, which took effect in April 2022, in early March \citep{haring2022}; while SpaceX's RTO was announced with immediate effect via an internal email from Elon Musk to employees \citep{mac2022}.} Given that a typical job search takes around 4.5 months on average \citep{BLS2024} and resignations usually only take effect at the end of the month, this left essentially no room for anticipation that is not captured by our research design. Of course, employees may have heard rumors before the official announcement. Both Microsoft and Apple, however, delayed several prior RTOs indefinitely due to pushback from employees and pandemic conditions \citep{OLoughlin2023, OLoughlin2022}. Survey evidence from May 2022 indeed suggested that such pushback to office attendance requirements had been common and effective \citep{bloom2022greatresistance}. Thus, it seems likely that employees only acted on the mandates once they were actually implemented due to the uncertainty surrounding them as well as their previous successes in pushing back the dates. This does not contradict the possibility that employees planning to leave had already started their job search before the actual RTO, which would explain the sizeable tenure effects we observe only 3 months after implementation.

Third, due to the interference of layoff dates with RTO mandates, we are only able to study these mandates' effects one quarter after their implementation. It is possible that firms try to mitigate the documented shift in their tenure distribution after their RTO by increased hiring or internal promotions. This does not negate the significant costs associated with this shift and any mitigation thereof, which we discuss in more detail below.

Fourth, an important question is whether the compositional shift we document is driven by senior employees departing or by differential hiring of junior employees. To address this, we apply standard synthetic control methods \citep{abadie2010synthetic} to aggregate separation and hiring rates. \revision{We focus on Microsoft for this exercise, since it has a much larger sample size, providing more power for accurately estimating effects on the much smaller sets of leavers and joiners and their subsequent destinations}. Table \ref{tab:worker_flows} reports the results. Neither Microsoft's separation rate ($p = 0.219$) nor its hiring rate ($p = 1$) differs significantly from the synthetic control. \revision{The precision of this null is limited: the placebo distributions of post-RTO flow-rate effects are wide (10th--90th percentile range $[-1.3,+0.6]$ pp for separations and $[-1.5,+0.8]$ pp for hiring; s.d.\ $\approx 1$ pp), so Table \ref{tab:worker_flows} can rule out flow-rate effects larger than roughly $\pm1.5$ percentage points but cannot exclude smaller changes.} Net employee flows were positive across all seniority levels post-RTO but this reflects the broader tech hiring boom in early 2022 rather than anything specific to Microsoft, as control firms exhibited nearly identical patterns. Importantly, the June spike in entry-level hiring associated with new college graduates is accounted for by the synthetic control; the DiSCo-weighted counterfactual exhibits a \revision{similarly large} June spike \revision{as} Microsoft (46.4\% vs.\ 62.6\% entry-level share among new hires). \revision{The aggregate null should therefore be read as ``no large change in volumes'' while the mechanism is identified through the composition of leavers (Section \ref{sec:separation_composition}) rather than aggregate rates.}

\begin{table}[htbp!]
\centering
\caption{Aggregate Worker Flow Rates: Synthetic Control Analysis}
\label{tab:worker_flows}
\begin{tabular}{lcccc}
\toprule
& \multicolumn{2}{c}{Microsoft} & \multicolumn{2}{c}{Control Average} \\
\cmidrule(lr){2-3} \cmidrule(lr){4-5}
& Pre-RTO & Post-RTO & Pre-RTO & Post-RTO \\
\midrule
Separation Rate (\%) & 1.55 & 1.59 & 2.01 & 2.07 \\
\quad \textit{Placebo p-value} & \multicolumn{2}{c}{\textit{0.219}} & & \\[0.3em]
Hiring Rate (\%) & 1.86 & 2.55 & 2.20 & 2.48 \\
\quad \textit{Placebo p-value} & \multicolumn{2}{c}{\textit{1.000}} & & \\
\bottomrule
\end{tabular}
\floatfoot{\textit{Notes}: Table reports monthly worker flow rates for Microsoft and control firms before and after Microsoft's RTO mandate (April 2022). Separation rate is the share of employees observed in month $t$ who are not observed at the same firm in month $t+1$. Hiring rate is the share of employees in month $t$ who first appear at the firm that month. Pre-RTO period: 6 months; Post-RTO period: 3 months. The donor pool is the 31 control firms of the main distributional synthetic control (Table~\ref{tab:control_firms}). Placebo p-values are computed across these 31 controls: each control is treated as a placebo, and the p-value is the share of units (including Microsoft) whose post-RTO effect is at least as large in absolute value as Microsoft's. Neither rate differs significantly from controls, but this null is imprecise: Microsoft's estimated effects (separation $\approx -1.2$ pp, hiring $\approx +0.0$ pp) lie within wide placebo distributions (10th--90th percentile range $[-1.3,+0.6]$ pp for separations and $[-1.5,+0.8]$ pp for hiring; s.d.\ $\approx 0.8$ and $1.1$ pp). The mechanism is therefore identified through the \emph{composition} of leavers (Table~\ref{tab:separation_composition}), not aggregate volumes.}
\end{table}

Appendix \ref{app:flows} studies this mechanism in greater detail. Using synthetic control methods on the characteristics of separating employees, we find that the mean tenure of leavers at Microsoft increased by 110.8 days relative to the counterfactual ($p = 0.032$), and the share of leavers with senior titles rose by 8.4 percentage points ($p = 0.097$). The composition of new hires, by contrast, \revision{was unchanged, with the same caveat about precision as above}. The RTO thus changed \textit{who} left, not \textit{how many} left, with junior employees staying at higher rates and offsetting the increase in senior departures. This pattern is consistent with \citet{emanuel2023power}, who find that the productivity gains from office proximity are concentrated among junior employees still building firm-specific human capital, who as a result may be more inclined to stay after an RTO mandate.

Finally, the results are robust to extending the pre-treatment period out further to three quarters before the RTO mandates, as shown in Appendix \ref{app:extended_pretreatment}.

\subsection{Where Do Employees Go After an RTO?} \label{sec:where_do_they_go}

Our main results provide robust evidence that long-tenured and senior employees depart from Microsoft at higher rates after it implements a return to office mandate. A natural follow-up question to ask is: where do these senior employees go? To answer this question, we construct a set of distributional synthetic controls for the employees who leave Microsoft in a given quarter. Since the leavers comprise a far smaller sample than the total number of employees, we focus on replicating an aggregate distributional statistic \citep{abadie2010synthetic}, the share of leavers who have characteristic X, rather than the entire distribution. 

This way, we find that Microsoft's RTO led to a counterfactual decrease in the share of its employees departing for startups, as shown in Figure \ref{fig:startups}. In particular, we proxy for whether a firm is a start-up in two ways: whether it has fewer than 50 employees, and whether it is below the 5th quantile in terms of its industry labor share. For both measures, we estimate an increase in the mass of workers leaving to non-startups of around 5\% that is significant at the 95\% level. Moreover, the permutation test for both measures indicates the inference is (marginally) valid, with p-values of 0.03 and 0.09, respectively. Combined with our main finding that more senior employees are leaving at higher rates, this result \revision{aligns with} the descriptive evidence in Figure \ref{fig:network} that the RTO caused a departure of senior employees \revision{to other large firms}. Thus, Microsoft's RTO appears to not only have had important human capital costs due to the hiring and training costs required to replace the senior employees lost, but also potential competitive side effects due to these employees taking operational knowledge with them.

\begin{figure}[ht!]
\begin{subfigure}[b]{0.49\textwidth}
\includegraphics[width=\textwidth]{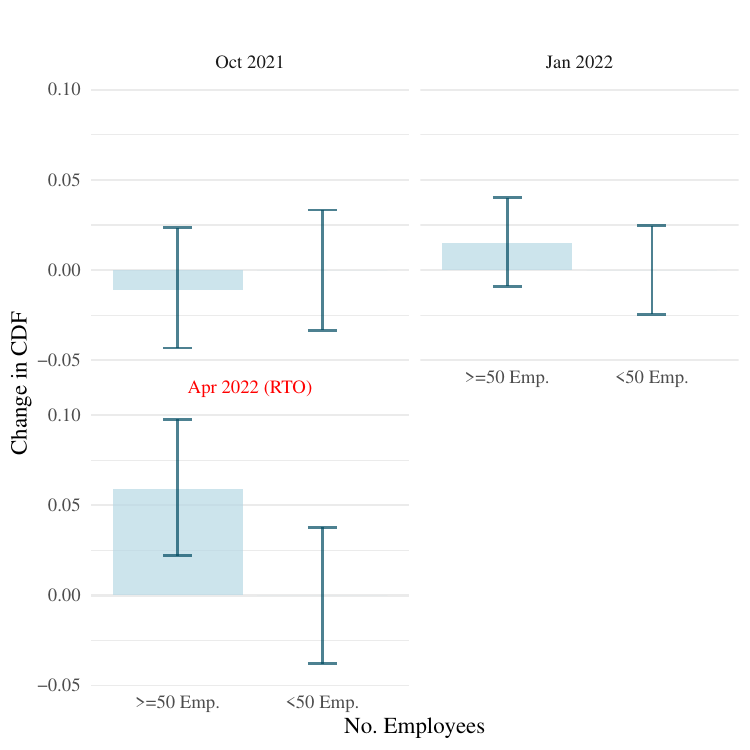}
\end{subfigure}
\begin{subfigure}[b]{0.49\textwidth}
\includegraphics[width=\textwidth]{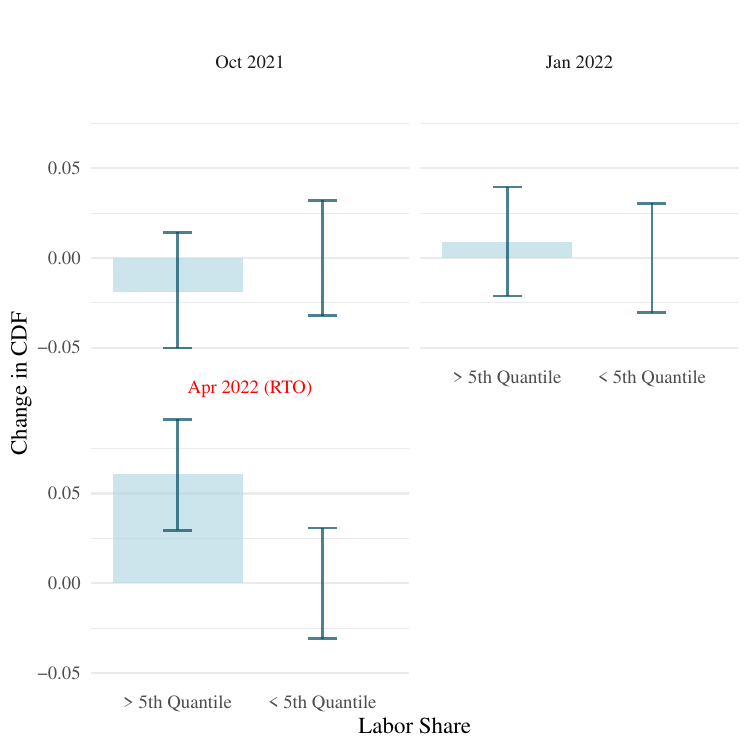}
\end{subfigure}
\caption{Share of Employees Leaving to Startups}
\floatfoot{\textit{Note}: figure shows counterfactual change in share of employees leaving Microsoft by size of the company they leave to, measured by number of employees (larger or smaller than 50, left panel) and the company's labor share (number of employees as a share of total industry employment, right panel). Estimated change is depicted by light blue solid bar and 95\% bootstrap bands by dark blue whiskers (1,000 bootstraps). Pre-treatment fits are depicted in the top two bar plots, while the post-treatment counterfactual is shown in the bottom bar plot. Permutation p-values (see \eqref{eq:permutation}) are 0.03 and 0.09, respectively. N=1,846, J+1=32.}
\label{fig:startups}
\end{figure}

Finally, in Figure \ref{fig:null_results} we further test for, but fail to find any evidence of, distributional differences in: the share of men vs. women leaving Microsoft; the share of leavers that go into unemployment; the share of leavers that accept a demotion in their title at their new job; and the share of leavers changing roles at their new job. \revision{The lack of gender differences is notable given prior survey and experimental evidence that women report stronger preferences for working from home \citep{yougov2022working, flexjobs2021menwomen, mas_valuing_2017}; in our setting, however, we do not find evidence that the RTO induced differential exits by gender.} The lack of counterfactual changes in unemployment status, demotions, and job roles suggests that the senior employees leaving Microsoft have good outside options, in the sense that they appear to find employment in similar roles and levels of seniority. Of course, given employees' well-documented willingness to pay for remote work \citep{mas_valuing_2017}, they may have accepted wage cuts. Unfortunately, the PDL data does not contain accurate salary information, so we cannot test for this.  Nonetheless, our other findings align with the fact that Microsoft implemented its RTO mandate much earlier than most tech companies, which meant there were still many fully remote companies the leavers could go to. As such, we can expect the value of employees' outside option to shrink as more companies implement RTOs. 

At the same time, this also suggests that the share of companies in an industry that have returned to office will determine the negative impact of an RTO mandate on the company that implements it. That is because, all else equal, we can expect that the outflow of senior employees that we have documented will be larger when these employees' outside options are better. This hints at the existence of a tipping point where companies implement RTO mandates at increasing rates when the value of employees' outside options has diminished sufficiently. As a counterpoint, many companies started explicitly coding the degree of hybrid work allowed into new joiners' contracts, which may generate path dependence both for fully remote companies as well as those that returned to the office. We provide suggestive evidence that a more slack labor market attenuates the compositional effects we document in Appendix \ref{app:external_validity}, where we replicate our analysis for Amazon and Google. Both firms implemented their RTO mandate in 2023, in the midst of a sector-wide wave of layoffs, which significantly weakened the tech sector's labor market. Indeed, both firms conducted large-scale layoffs themselves in the months leading up to their RTO. As such, the RTO policy intervention is almost certainly confounded by these layoffs, and these results should be interpreted with caution. \revision{We focus on these two firms as they have by far the largest sample sizes among the set of firms that implemented RTOs early on in this tech layoff wave.} Nonetheless, we find some evidence that the effects we document are indeed fully attenuated at these firms. At the same time, another study by \citet{ding2024return} has replicated our key results of a ``brain drain'' for 54 S\&P 500 firms with RTO mandates between Q1 2020 and Q4 2023 using a difference-in-differences design, providing independent evidence for the external validity of our findings. 

\subsection{Discussion}

This article estimates that a return-to-office mandate at Microsoft led to a significant outflow of senior employees to \revision{large firms, some of which are direct Microsoft competitors}. What does this suggest about the larger ramifications of such mandates for the companies that implement them? 
To start, consider the reasons for implementing an RTO mandate. Managers tend to believe work-from-home decreases productivity \citep{Bloom_Barrero_Davis_Meyer_Mihaylov_2023} and diffusion of knowledge within the firm \citep{barrero_evolution_2023}, and face difficulties monitoring employees in remote-work settings \citep{Microsoft2022HybridWork}. The academic literature provides support for some, but not all, of these concerns -- see \citet{barrero_evolution_2023} for an overview. On the other hand, employee satisfaction and retention have been found to be lower in in-person work settings \citep{bloom_does_2015, choudhury2024hybrid}. Our results suggest that such retention issues may be more serious than previously thought, as senior employees and, indeed, the (senior) \textit{managers themselves} are leaving. Moreover, the sudden shift induced by an RTO mandate may cause a large, temporary \textit{spike} in employee attrition, which can have quantitatively and qualitatively different effects from chronic but gradual attrition. Thus, a firm returning to the office with the aim of increasing productivity or innovation may end up doing so in a narrow sense, improving the output metrics of those who stay, but harming those metrics for the company as a whole.

In support of this, the academic literature provides ample evidence on the value of senior and long-tenured employees to the companies they work at. In broad terms, human capital has been theorized to be a key factor determining firms' competitive advantage and estimated to be strongly related to firm performance \citep{crook2011does}. More specifically, worker output has been found to be monotonically increasing in tenure \citep{shaw2008tenure}, with bosses estimated to be almost two times as productive as the average worker \citep{lazear2015value}. These effects are driven by (costly) employee training \citep{de2012effects, konings2015impact} and the more general accrual of firm-specific human capital with tenure \citep{becker1962investment}. Moreover, senior employees tend to not just \textit{be} productive; they also strongly affect a firm's wider productivity. Managers and management practices have been found to explain a large share of the variation in firm productivity \citep{bender2018management, metcalfe2023managers, fenizia2022managers}, especially in combination with higher-quality human capital \citep{bender2018management}. Employees working in research and development (R\&D) not only increase a firm's innovation rates, but also take their accumulated knowledge with them when they move to competitors \citep{palomeras2010markets, cassiman2006search}. Relatedly, lower employee turnover has been found to increase firm investment and decrease new firm entry \citep{jeffers2024impact}. Replacing skilled employees also incurs substantial hiring costs, equal to around 2--4 months of pay, that increase with the skill level \citep{blatter2012costs}. 

As discussed above, RTO mandates are implemented strategically by firms. Hence, it is natural to ask to what extent the employment consequences we document in this paper were intended. While it is possible that this was partially the case, it seems unlikely that this fully explains the effects we document, for two reasons. One, their magnitude. In the absence of large-scale layoffs, which we manually ascertained did not happen for any of the firms in our sample, it seems unlikely that firms intended their RTO mandate to cause 5\% of senior employees to leave. We validate this by re-estimating our main results for the January 2023 layoffs at Microsoft in Appendix \ref{app:layoff_comparison}. That way, we do not find significant effects of the layoff on either the tenure or title distribution, supporting our argument that the observed results are not driven by concurrent layoffs, which appear to not affect the workforce composition. Two, firms' forecasting abilities. Given the unpredictable nature of employee response, and limited forecasting technology, it seems unlikely that firms were able to predict precisely how their employees would respond to an RTO mandate. This hypothesis appears to be supported by the many articles documenting employee resistance to announced mandates, and subsequent backtracking on the firm side. It thus seems reasonable to consider the employment effects documented here as mostly unintentional side effects of RTO mandates. 

Taken together, this large body of literature implies that increased attrition of senior employees can substantially impede firm output, productivity, innovation, and competitiveness. This provides an important perspective on the implications of the return to office for the broader functioning of a company. Moreover, the downstream outcomes of the ``uneven attrition'' we document may be a fruitful avenue for future research.

\section{Conclusion} \label{sec:conclusion}

The workplace revolution brought on by the COVID-19 pandemic continues to reverberate across offices globally, even as the pandemic itself retreats into history. While some companies have returned to the office, many others have gone fully remote. This paper has provided novel evidence on the implications of this bifurcation for the distribution of employees across and within companies. In particular, we provide causal evidence that three of the largest U.S. tech companies---Microsoft, SpaceX, and Apple---experienced similar shifts toward a less senior workforce after implementing an RTO mandate, with the strongest statistical evidence for Microsoft and SpaceX.
 We do this using a distributional extension of the classical synthetic controls estimator \citep{abadie2003economic, gunsilius2023distributional}, for which we derive conditions for the uniform validity of the bootstrap. Moreover, we provide both descriptive and causal evidence that these senior employees left for larger companies at higher rates than usual. Taken together, our findings imply that return-to-office mandates can lead to significant human capital costs in terms of output, productivity, innovation, and competitiveness for the companies that implement them. Furthermore, the sorting of skilled labor across companies that our findings imply can have important consequences for the employment landscape, which we leave to future work. 

\clearpage

\section*{Data Availability Statement}
\addcontentsline{toc}{section}{Data Availability Statement}

The code that generates every table and figure in this paper and its
appendices, together with full documentation of the construction of the
analysis samples, is available at \url{https://doi.org/10.7910/DVN/FPKORT}. A detailed README
describes the computing environment, the order in which programs must be
run, and the mapping from each program to each exhibit.

The individual-level resume records underlying the analysis are a commercial
product of People Data Labs, licensed by the authors through Dewey Data
(\url{https://www.deweydata.io}), a data marketplace that provides
third-party datasets to academic researchers. The data use agreement
prohibits redistribution of the individual-level records, so they are not
included in the replication package. The package therefore documents the
complete chain from the raw extract to the reported results: all extraction,
cleaning, and analysis code, the query specification and vintage of the
extract used, and the aggregated firm-by-quarter distributions required to
reproduce the headline estimates without the individual data.

Any researcher may license the same data on the terms the authors received;
the authors have no preferential access and no exclusive arrangement.
Researchers at institutions holding a Dewey subscription can register at
\url{https://www.deweydata.io} with an institutional email address and
request the dataset through their library's data services group. Researchers
without institutional access may license it directly via the contact form at
\url{https://www.deweydata.io/contact}. The records are also available from
People Data Labs (\url{https://www.peopledatalabs.com}) through that vendor's
standard commercial process. Dewey licenses these data for academic and
non-commercial research only.

The publicly available inputs used in the paper---including the firm
headcount figures collected from SEC filings and company statements and the
Bureau of Labor Statistics series used in the descriptive comparisons---are
provided in full in the replication package.

\clearpage

\bibliographystyle{agsm}
\bibliography{references}

\ifsubmissionformat
  \processdelayedfloats
  \let\figure\origfigure
  \let\endfigure\origendfigure
  \let\table\origtable
  \let\endtable\origendtable
\fi

\newpage

\begin{center}
{\Large \textbf{Online Appendix}}\\[1em]
{\large Return to Office and the Tenure Distribution}\\[0.5em]
{\normalsize David Van Dijcke, Florian Gunsilius, and Austin Wright}
\end{center}
\vspace{2em}

\appendix

\setcounter{table}{0}
\setcounter{figure}{0}
\setcounter{theorem}{0}
\setcounter{equation}{0}

\renewcommand{\thetable}{A-\arabic{table}}
\renewcommand{\thefigure}{A-\arabic{figure}}
\renewcommand{\thetheorem}{A-\arabic{theorem}}
\renewcommand{\theproposition}{A-\arabic{proposition}}
\renewcommand{\thelemma}{A-\arabic{lemma}}
\renewcommand{\thecorollary}{A-\arabic{corollary}}
\renewcommand{\theequation}{A-\arabic{equation}}
\renewcommand{\thedefinition}{A-\arabic{definition}}

\section{Additional Results}

\subsection{Figures}

\begin{figure}[htbp!]
\begin{subfigure}[b]{0.49\textwidth}
    \includegraphics[width=\textwidth]{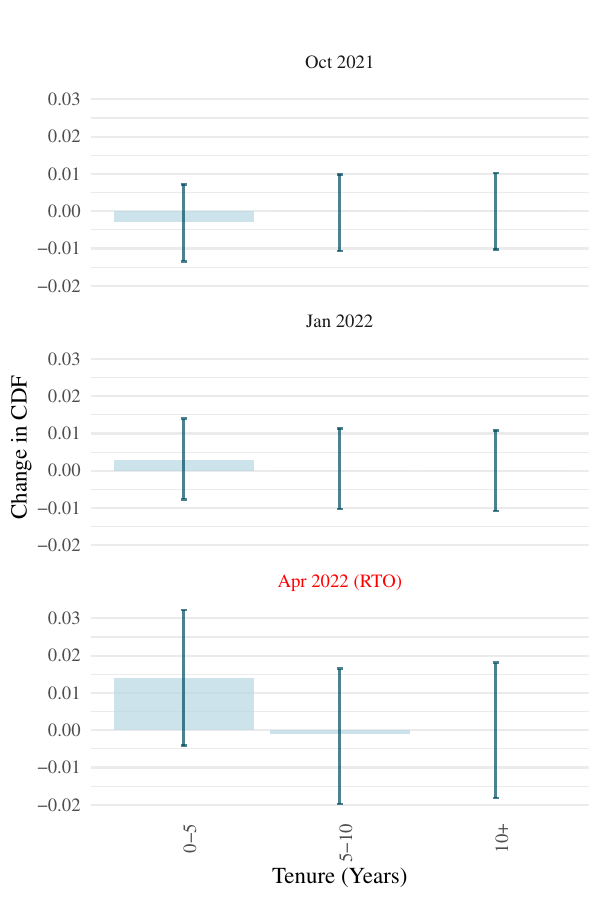}
\caption{Apple}
\end{subfigure}
\begin{subfigure}[b]{0.49\textwidth}
\includegraphics[width=\textwidth]{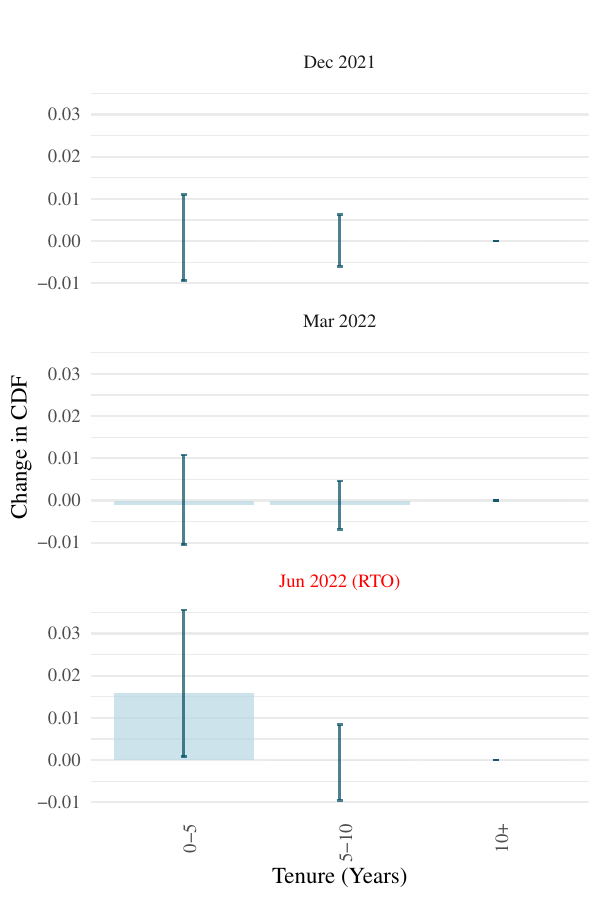}
\caption{SpaceX}
\end{subfigure}
\caption{Tenure After RTO, Robustness}
\label{fig:tenure_robustness}
\floatfoot{\textit{Note}: figure shows counterfactual change in cumulative distribution functions of discretized tenure levels at Apple and SpaceX for two quarters before its return to office (RTO) and one quarter after. Estimated change is depicted by light blue solid bar and 95\% bootstrap bands by dark blue whiskers (1,000 bootstraps). Permutation test p-values (see \eqref{eq:permutation}) are: Apple 0.25, SpaceX 0.04. Apple: N=1,080,159, J+1=44; SpaceX: N=935,060, J+1=220.}
\end{figure}

\begin{figure}[htbp!]
\begin{subfigure}[b]{0.49\textwidth}
    \includegraphics[width=\textwidth]{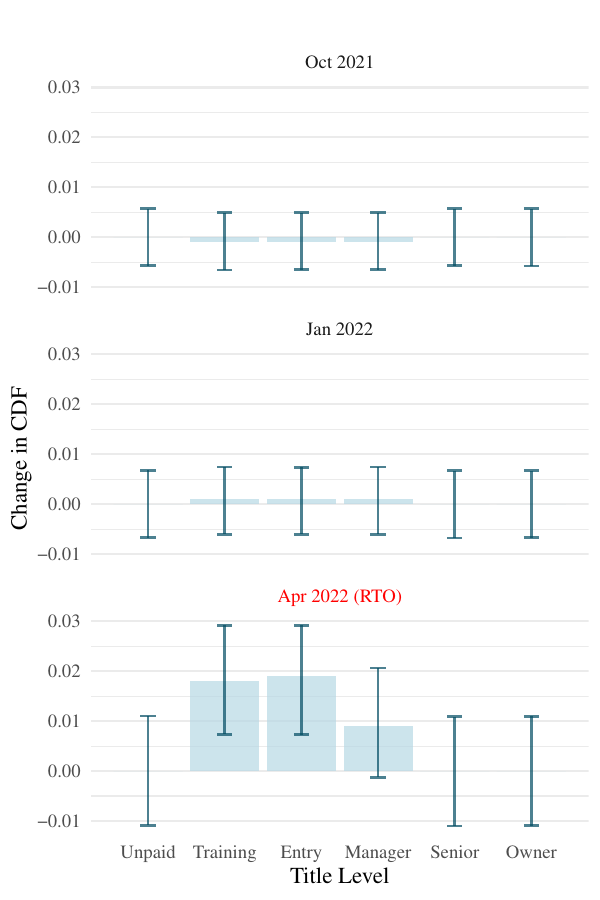}
\caption{Apple}
\end{subfigure}
\begin{subfigure}[b]{0.49\textwidth}
\includegraphics[width=\textwidth]{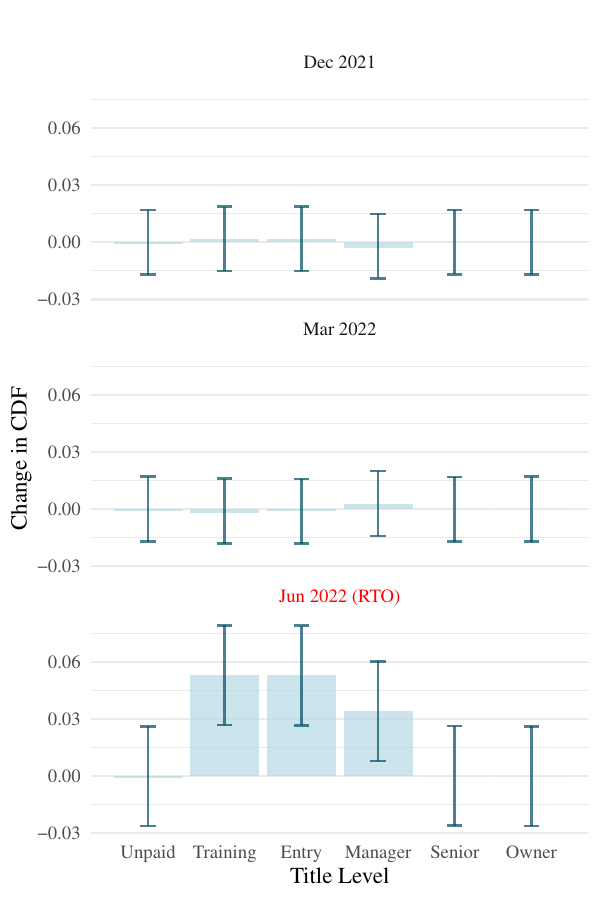}
\caption{SpaceX}
\end{subfigure}
\caption{Titles After RTO, Robustness}
\label{fig:title_robustness}
\floatfoot{\textit{Note}:  Estimated change is depicted by light blue solid bar and 95\% bootstrap bands by dark blue whiskers (1,000 bootstraps). Permutation test p-values (see \eqref{eq:permutation}) are: Apple 0.07, SpaceX 0.046. Apple: N=1,080,159, J+1=44; SpaceX: N=935,060, J+1=220.}
\end{figure}

  \begin{figure}[htbp!]                                                                                           
  \centering                                                                                                      
  \begin{subfigure}[b]{0.32\textwidth}                                                                            
      \includegraphics[width=\textwidth]{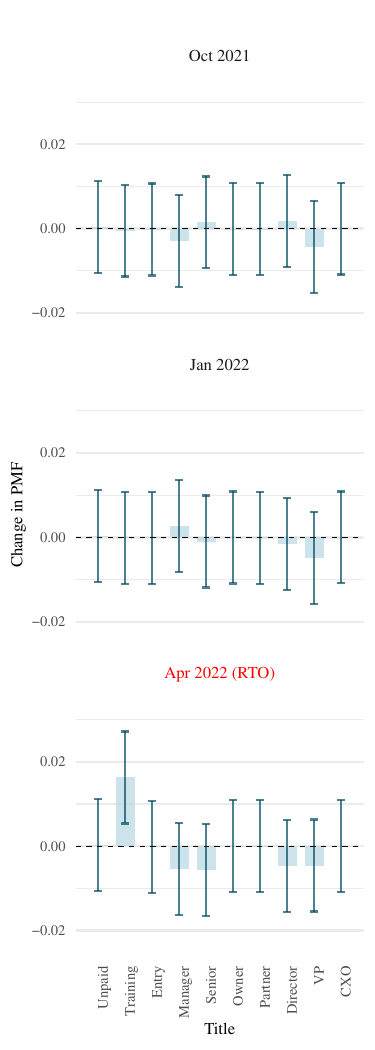}                                            
  \caption{Microsoft}                                                                                             
  \end{subfigure}                                                                                                 
  \begin{subfigure}[b]{0.32\textwidth}                                                                            
      \includegraphics[width=\textwidth]{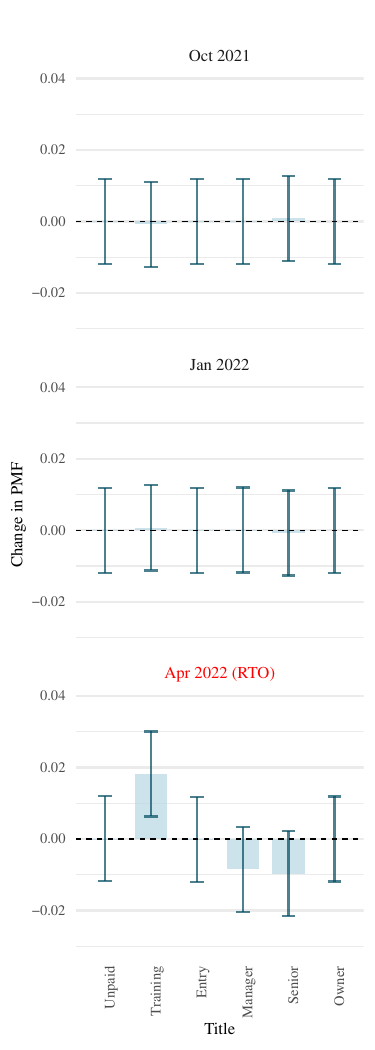}                                                
  \caption{Apple}                                                                                                 
  \end{subfigure}                                                                                                 
  \begin{subfigure}[b]{0.32\textwidth}                                                                            
  \includegraphics[width=\textwidth]{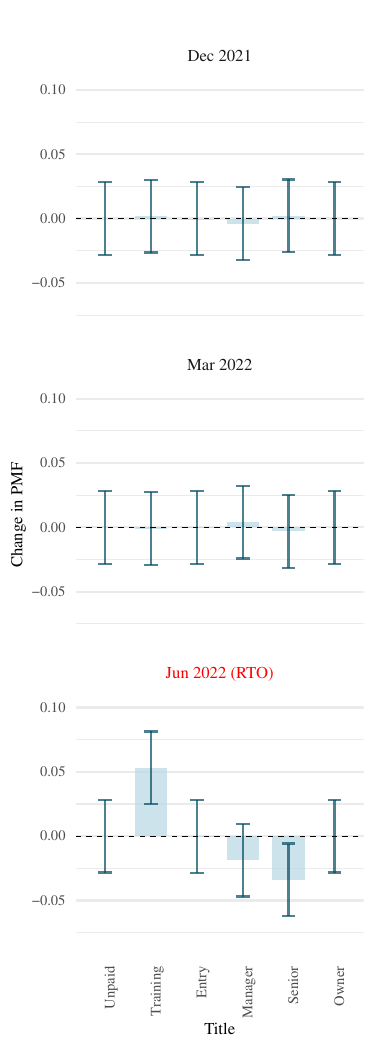}                                                   
  \caption{SpaceX}                                                                                                
  \end{subfigure}                                                                                                 
  \caption{Title Distribution Change: Probability Mass Function}                                                  
  \label{fig:title_pmf}                                                                                           
  \floatfoot{\textit{Note}: Figure shows the change in probability mass function (PMF) at each title level,       
  computed as the first difference of the CDF changes. Light blue bars show point estimates; dark blue whiskers   
  show 95\% bootstrap bands (1,000 bootstraps). Positive values indicate an increase in 
  the share of employees at that title level relative to the synthetic control; negative values indicate a        
  decrease.}                                                                                                      
  \end{figure}
  
\begin{figure}[htbp!]                        
  \centering                      \includegraphics[width=0.65\textwidth]{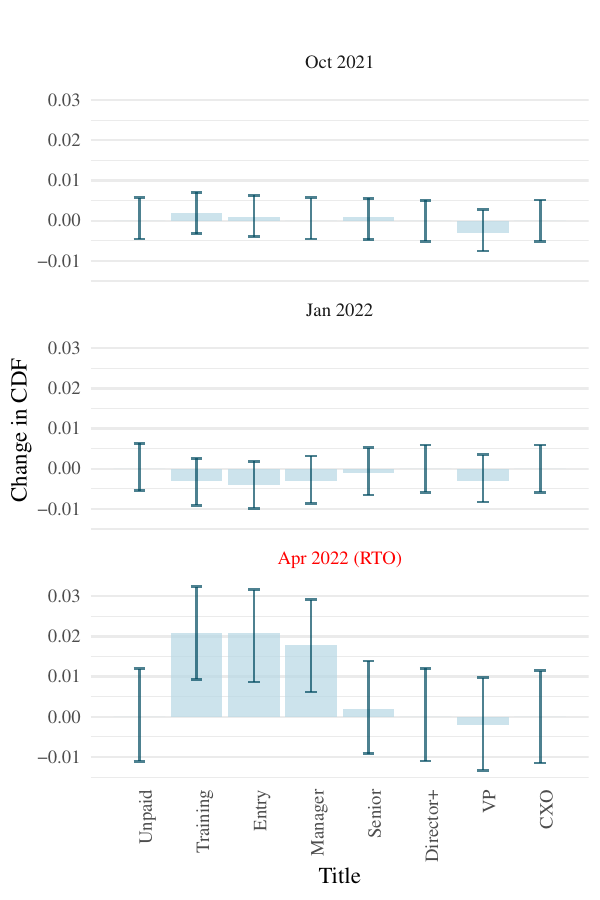}                                        
  \caption{Microsoft Title Robustness: Owner/Partner Merge into Director}                                                 
  \floatfoot{\textit{Note}: Figure replicates the title
  distribution analysis (Figure \ref{fig:main_titles}) after
  merging the Owner and Partner categories into Director (labeled
  ``Director+''). These categories together comprise 0.60\% of
  Microsoft profiles. Permutation test p-value: 0.031. Bars show
  CDF differences between Microsoft and its synthetic control;
  whiskers show 95\% bootstrap confidence intervals.}
  \label{fig:title_robustness_pooled}
  \end{figure}

\begin{figure}[hb!]
\centering
\begin{subfigure}[b]{0.49\textwidth}
    \includegraphics[width=\textwidth]{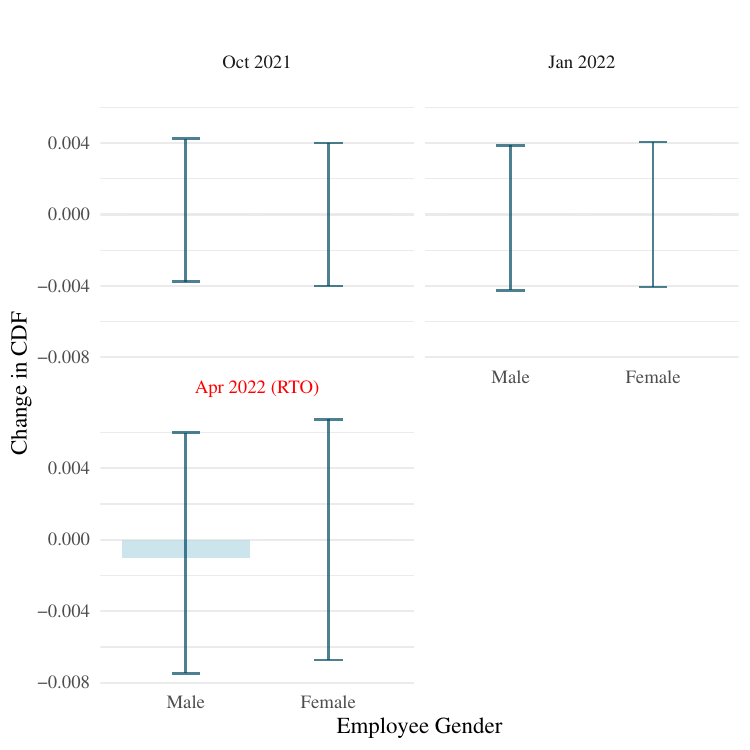}
\caption{Gender}
\end{subfigure}
\begin{subfigure}[b]{0.49\textwidth}
\includegraphics[width=\textwidth]{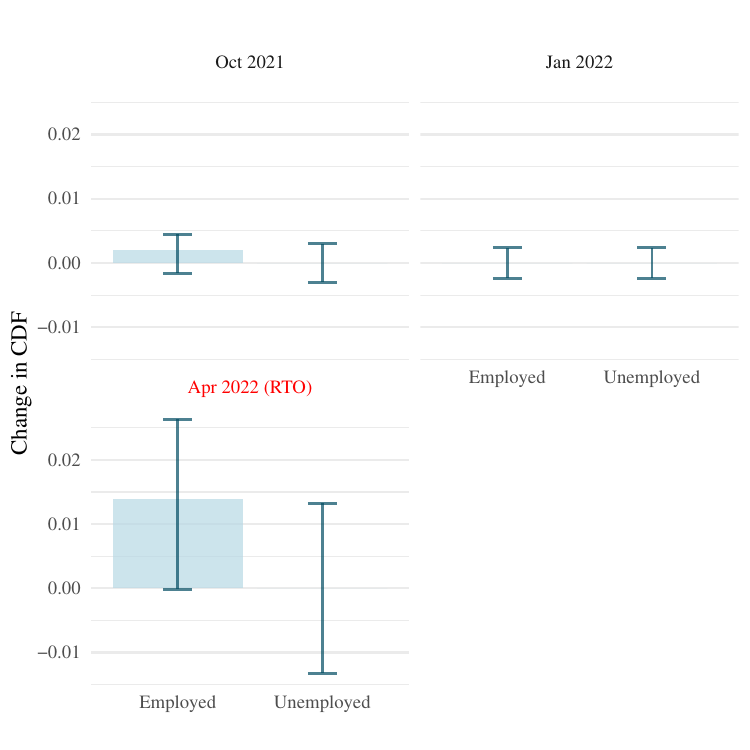}
\caption{New Employment Status}
\end{subfigure}
\begin{subfigure}[b]{0.49\textwidth}
\includegraphics[width=\textwidth]{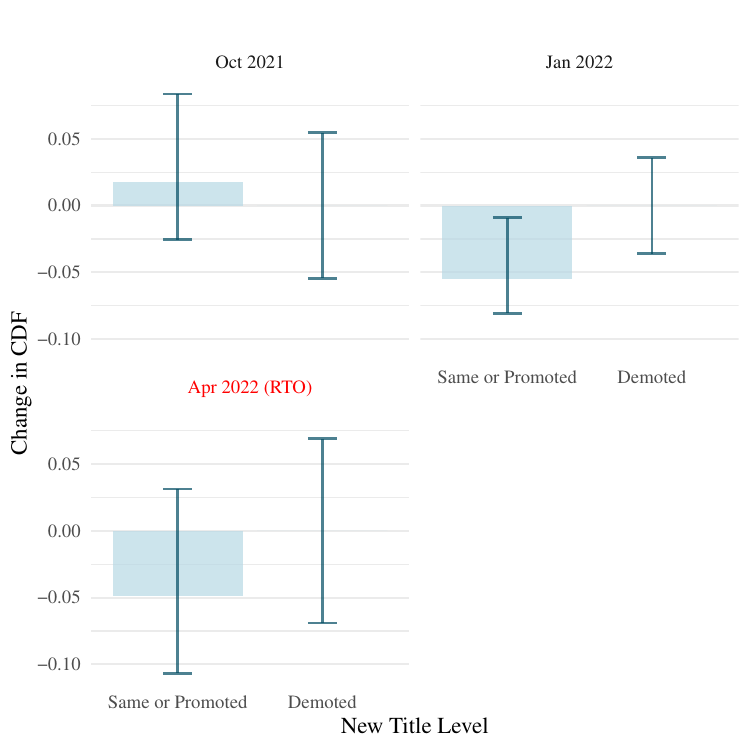}
\caption{Demotions}
\end{subfigure}
\begin{subfigure}[b]{0.49\textwidth}
\includegraphics[width=\textwidth]{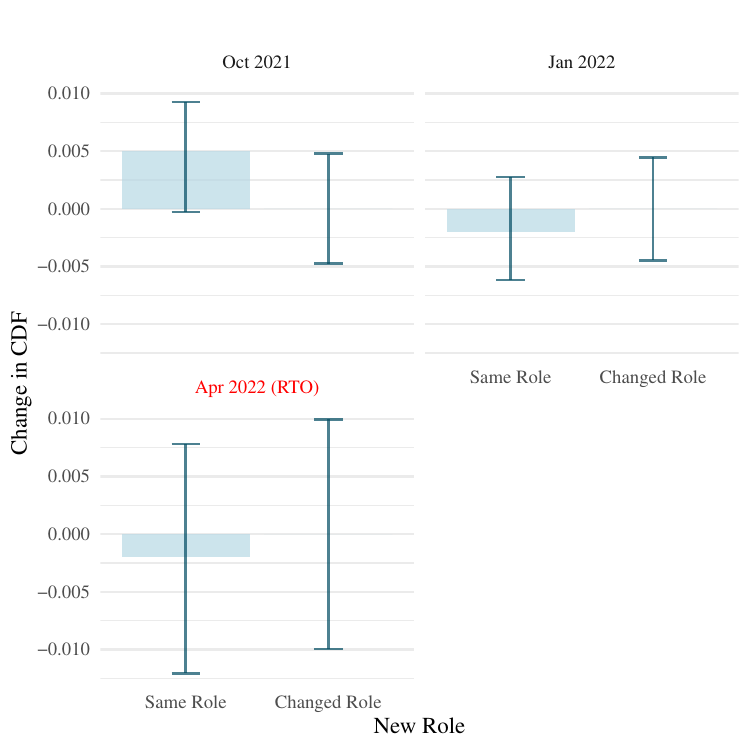}
\caption{Role Change}
\end{subfigure}
\caption{Gender, Employment, Demotions, and Role Changes Among Employees Leaving Microsoft}
\floatfoot{\textit{Note}: Figure shows counterfactual change in share of employees leaving Microsoft by gender, new employment status, whether they accepted a demotion in their new job, and whether they changed roles in their new job.  Estimated change is depicted by light blue solid bar and 95\% bootstrap bands by dark blue whiskers (1,000 bootstraps). Pre-treatment fits are depicted in the top two bar plots, while the post-treatment counterfactual is shown in the third bar plot in each graph. Permutation test p-values (see \eqref{eq:permutation}) are: gender 0.44, employment status 0.41, demotions 0.69, role change 0.88.}
\label{fig:null_results}
\end{figure}

\clearpage

\subsection{Tables}

\begin{table}[htbp!]
\centering
\footnotesize
\caption{Control Firms Used in Microsoft Analysis}
\label{tab:control_firms}
\begin{tabular}{lllll}
\toprule
Company & Ticker & RTO Status & \citet{flynn2024determinants} & Verification \\
\midrule
\multicolumn{5}{l}{\textit{Panel A: Public Companies}} \\[0.3em]
3M & MMM & Sep 2025$^\dag$ & \cmark$^a$ & \href{https://www.axios.com/local/twin-cities/2025/07/08/3m-calls-office-workers-back-four-days-a-week}{News (Oct 2021)} \\
Airbnb & ABNB & Never Treated & --- & \href{https://news.airbnb.com/airbnbs-design-to-live-and-work-anywhere/}{News (Apr 2022)} \\
AMD & AMD & Never Treated & \cmark & \href{https://www.teamblind.com/post/amd-return-to-office-policy-ramtvgka}{News} \\
Amazon & AMZN & May 2023$^\dag$ & \cmark$^a$ & \href{https://www.washingtonpost.com/us-policy/2023/02/17/amazon-rto-mandate/}{News (Feb 2023)} \\
Autodesk & ADSK & Never Treated & \cmark$^b$ & \href{https://adsknews.autodesk.com/en/news/flex-forward-hybrid-work/}{News (May 2021)} \\
Broadridge & BR & Never Treated & \cmark$^b$ & \href{https://www.glassdoor.com/Reviews/Broadridge-work-from-home-Reviews-EI_IE42559.0,10_KH11,25.htm}{Glassdoor} \\
Cisco & CSCO & Never Treated & \cmark$^b$ & \href{https://www.sdxcentral.com/articles/news/cisco-goes-all-hybrid-with-no-return-to-office-policy/2021/07/}{News (Jul 2021)} \\
Dell Technologies & DELL & May 2023$^\dag$ & \cmark$^a$ & \href{https://buildremote.co/return-to-office/dell/}{News (May 2023)} \\
DocuSign & DOCU & Jun 2023$^\dag$ & \cmark & \href{https://www.geekwire.com/2023/docusign-which-is-forcing-employees-back-to-the-office-touts-report-about-remote-work-benefits/}{News (May 2023)} \\
Gartner & IT & Never Treated & \cmark & \href{https://www.glassdoor.com/Benefits/Gartner-Work-From-Home-US-BNFT152_E2465_N1.htm}{Glassdoor} \\
HP & HPQ & Never Treated & \cmark$^b$ & \href{https://www.computerworld.com/article/1635006/hp-exec-theres-no-going-back-to-the-pre-pandemic-workplace.html}{News (Aug 2021)} \\
Intel Corporation & INTC & Sep 2025$^\dag$ & \cmark$^a$ & \href{https://www.intc.com/news-events/press-releases/detail/1738/lip-bu-tan-our-path-forward}{News (Apr 2025)} \\
Intuit & INTU & Never Treated & \cmark$^b$ & \href{https://www.intuit.com/careers/return-to-work-plan/}{News (Apr 2021)} \\
Molina Healthcare & MOH & Never Treated & \cmark & \href{https://www.healthcaredive.com/news/molina-permanent-remote-work-cut-real-estate-footprint/628291/}{News (Jul 2022)} \\
Motorola Solutions & MSI & Never Treated & --- & \href{https://www.glassdoor.com/Benefits/Motorola-Solutions-Work-From-Home-US-BNFT152_E427189_N1.htm}{Glassdoor} \\
NVIDIA & NVDA & Never Treated & --- & \href{https://fortune.com/2023/10/14/nvidia-skips-return-to-office-sticks-to-remote-work-among-hottest-tech-companies/}{News (Oct 2023)} \\
Oracle & ORCL & Apr 2023$^\dag$ & \cmark$^a$ & \href{https://www.bloomberg.com/news/articles/2023-04-13/oracle-s-netsuite-orders-return-to-office-as-tech-wavers-on-remote-work}{News (Apr 2023)} \\
Salesforce & CRM & Feb 2023$^\dag$ & \cmark & \href{https://www.salesforce.com/news/stories/hybrid-work-guidelines/}{News (Feb 2023)} \\
Splunk & SPLK & Never Treated & \cmark & \href{https://www.sfgate.com/bayarea/article/2021-05-SF-Splunk-permanent-remote-work-16160557.php}{News (May 2021)} \\
Texas Instruments & TXN & Never Treated & --- & \href{https://www.glassdoor.com/Benefits/Texas-Instruments-Work-From-Home-US-BNFT152_E651_N1.htm}{Glassdoor}$^c$ \\
\midrule
\multicolumn{5}{l}{\textit{Panel B: Private Companies}} \\[0.3em]
Avanade & --- & Never Treated & --- & \href{https://www.glassdoor.com/Benefits/Avanade-Work-From-Home-US-BNFT152_E11701_N1.htm}{Glassdoor} \\
Baker Tilly US & --- & Never Treated & --- & \href{https://www.glassdoor.com/Reviews/Baker-Tilly-work-from-home-Reviews-EI_IE249007.0,11_KH12,26.htm}{Glassdoor} \\
Citrix & --- & Oct 2022$^\dag$ & --- & \href{https://www.thelayoff.com/t/1iC4HYbr?reply=%404hkk+1iC4HYbr+}{TheLayoff.com} \\
Cox Automotive Inc. & --- & Never Treated & --- & \href{https://www.glassdoor.com/Benefits/Cox-Automotive-Work-From-Home-US-BNFT152_E981845_N1.htm}{Glassdoor} \\
Cox Communications & --- & Never Treated & --- & \href{https://www.glassdoor.ie/Benefits/Cox-Communications-Work-From-Home-US-BNFT152_E4063_N1.htm}{Glassdoor} \\
Deloitte  & --- & Never Treated & --- & \href{https://www.glassdoor.com/Benefits/Deloitte-Work-From-Home-US-BNFT152_E2763_N1.htm}{Glassdoor}  \\
Epsilon & --- & Never Treated & --- & \href{https://www.glassdoor.com/Benefits/Epsilon-Work-From-Home-US-BNFT152_E39997_N1.htm}{Glassdoor} \\
LinkedIn & --- & Oct 2023$^\dag$ & --- & \href{https://www.businessinsider.com/linkedin-return-to-office-hybrid-workers-2023-10}{News (Oct 2023)} \\
McKinsey \& Company & --- & Never Treated & --- & \href{https://tech.co/news/mckinsey-review-remote-working}{News} \\
Protiviti & --- & Never Treated & --- & \href{https://www.glassdoor.com/Benefits/Protiviti-Work-From-Home-US-BNFT152_E30849_N1.htm}{Glassdoor} \\
Slalom Consulting & --- & Never Treated & --- & \href{https://www.glassdoor.com/Benefits/Slalom-Work-From-Home-US-BNFT152_E31102_N1.htm}{Glassdoor} \\
\bottomrule
\end{tabular}
\floatfoot{\textit{Note}: Table lists all 31 control firms used in the Microsoft synthetic control analysis. ``Never Treated'' indicates the firm had not implemented an RTO mandate \revision{by early 2024 (time of this article's first draft)}. $^\dag$Firms that implemented RTO after our sample period; included as ``not-yet-treated'' controls. The \citet{flynn2024determinants} column indicates agreement (\cmark) with the RTO classification in that paper; ``---'' indicates the firm was not in their sample. $^a$RTO announced after their data collection ended. $^b$Classified as ``Hybrid'' in their data but with employee choice and no mandatory office days. $^c$Employee reports vary by role; semiconductor manufacturing requires on-site presence. All Glassdoor links accessed February 2026. All firms operate in NAICS codes 51 (Information), 54 (Professional Services), or 33 (Manufacturing).}
\end{table}

\begin{table}[htbp!]
\centering
\caption{Covariate Balance: Microsoft vs.\ Control Firms}
\label{tab:balance}
\begin{tabular}{lcccc}
\toprule
 & Microsoft & Control & \multicolumn{2}{c}{Synthetic Control} \\
\cmidrule(lr){4-5}
 & & Firms & Tenure Wts. & Title Wts. \\
\midrule
Median Tenure (days) & 884 & 968 & 902 & 754 \\
Mean Title & 5.2 & 5.3 & 5.1 & 5.2 \\
Number of Resumes & 49,083 & 10,584 & 20,392 & 10,253 \\
Quarterly Turnover & 0.053 & 0.062 & 0.069 & 0.070 \\
Share Female & 0.261 & 0.317 & 0.300 & 0.356 \\
Share to Startups & 0.032 & 0.043 & 0.043 & 0.037 \\
\bottomrule
\end{tabular}
\floatfoot{\textit{Note}: Observable characteristics of Microsoft and its 31 control firms in the quarter before the RTO (January--March 2022). The Control Firms column is the simple average across the 31 controls; the two Synthetic Control columns reweight them by the main tenure and title (seniority) DiSCo weights, with each cell the value for the resulting synthetic firm ($\sum_j w_j x_j$). Tenure is summarized by its median (robust to the right-skewed tenure distribution; the synthetic control's median is the DiSCo barycenter quantile at $0.5$), seniority by its mean (the median is an uninformative 5 for all firms). Title is on a 1--10 scale (5 = Senior).}
\end{table}

\clearpage 

\subsection{Extended Pre-Treatment Period}\label{app:extended_pretreatment}

A potential concern with synthetic control methods is whether the pre-treatment period is sufficiently long to credibly establish that the synthetic control captures the treated unit's counterfactual trajectory. Our main analysis uses two pre-treatment quarters (Q4 2021 and Q1 2022). Here we extend the analysis to include three pre-treatment quarters: Q3 2021, Q4 2021, and Q1 2022.

The extended pre-treatment period serves two purposes: (1) it provides stronger evidence that our synthetic control accurately captures Microsoft's counterfactual trajectory, and (2) it allows us to rule out pre-existing trends in the tenure distribution that might confound our estimates.

Figure \ref{fig:extended_pretreatment} presents the results. The synthetic control continues to replicate Microsoft's distributions nearly perfectly across all three pre-treatment quarters, with confidence intervals including zero at all quantile grid points. The post-treatment effect remains statistically significant and similar in magnitude to our original estimates.

We emphasize that the linear distributional model of \citet[Appendix A]{gunsilius2023distributional} \revision{is more restrictive than the linear factor model of \citet{abadie2010synthetic}, as it does not allow for multiple unobserved factor components but rather assumes strict linearity in a one-dimensional unobservable. As such, under this model, our estimator can achieve identification without requiring the pre-treatment duration to approach infinity. The extended pre-treatment period test still helps in assessing the realism of our single-dimensional factor assumption. For a multivariate distributional factor model, we again refer to \citet{zhang2024asymptotic}.}

\begin{figure}[htbp!]
\centering
\begin{subfigure}[b]{0.32\textwidth}
    \includegraphics[width=\textwidth]{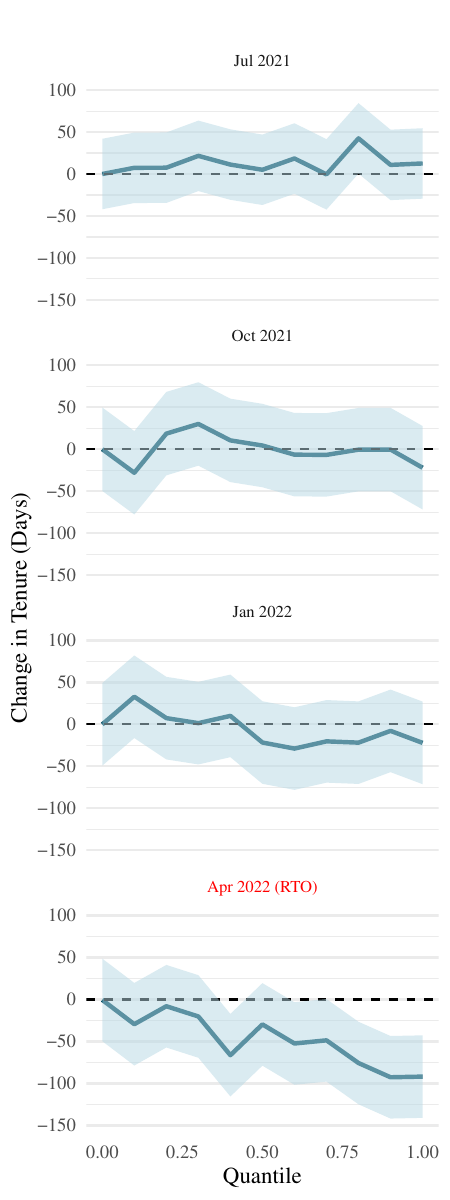}
    \caption{Tenure (Continuous)}
    \label{fig:extended_tenure_continuous}
\end{subfigure}
\begin{subfigure}[b]{0.32\textwidth}
    \includegraphics[width=\textwidth]{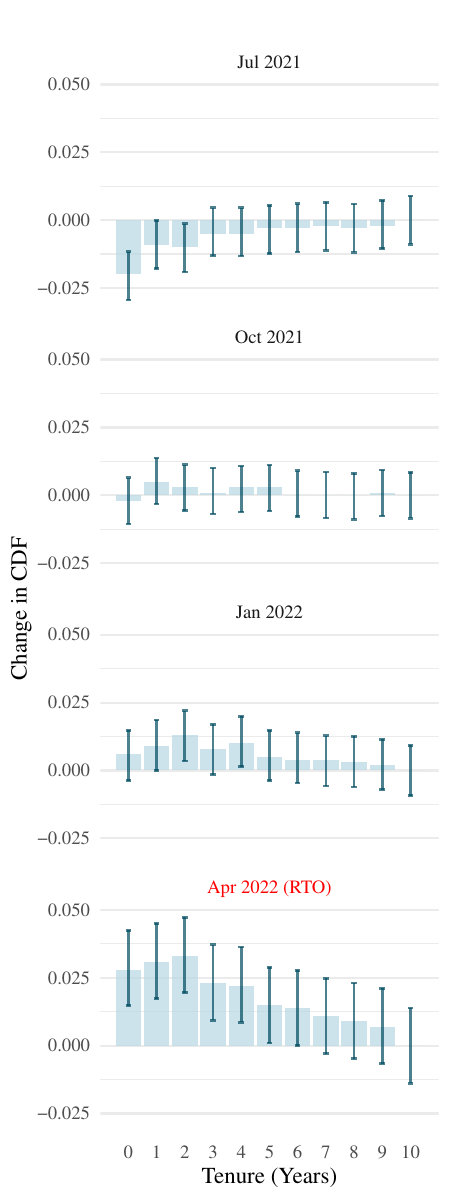}
    \caption{Tenure (Years)}
    \label{fig:extended_tenure_years}
\end{subfigure}
\begin{subfigure}[b]{0.32\textwidth}
    \includegraphics[width=\textwidth]{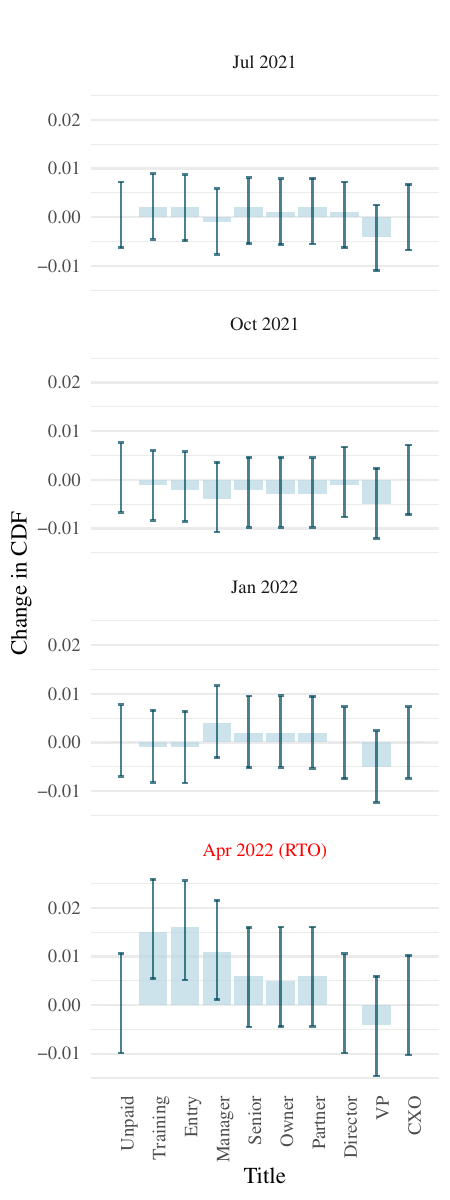}
    \caption{Titles}
    \label{fig:extended_titles}
\end{subfigure}
\caption{Extended Pre-Treatment Period: Three Quarters Before RTO}
\floatfoot{\textit{Note}: Figure shows distributional synthetic control estimates using three pre-treatment quarters (Q3 2021, Q4 2021, Q1 2022) and one post-treatment quarter (Q2 2022). Panel (a) shows counterfactual changes in quantile functions of tenure in days; panel (b) shows counterfactual changes in CDFs of tenure discretized into years; panel (c) shows counterfactual changes in CDFs of title levels. The synthetic control replicates Microsoft's distributions across all pre-treatment quarters, with confidence intervals including zero. Post-treatment effects remain statistically significant \revision{as shown by the uniform confidence bands} and similar in magnitude to the main specification. Permutation test p-values: tenure (continuous) 0.062 \revision{(0.03 for above-median quantiles)}, tenure (years) 0.156, titles 0.031.}
\label{fig:extended_pretreatment}
\end{figure}

\clearpage

\subsection{Spillover Robustness}\label{app:spillover}

A potential concern with the synthetic control approach is that control firms may experience ``indirect treatment'' if Microsoft employees who leave after the RTO join control firms, mechanically affecting their tenure distributions. This section quantifies the magnitude of such spillovers and demonstrates that our results are robust to their exclusion.

Table \ref{tab:spillover_magnitude} reports the destinations of Microsoft employees who left after the RTO and joined one of our 31 control firms. The spillover magnitude is negligible: only 78 Microsoft employees (0.16\% of Microsoft's workforce) joined any control firm in the post-treatment quarter. The maximum share of any control firm's post-RTO workforce that came from Microsoft was 0.30\% (Docusign), and no control firm received more than 1\% of its workforce from Microsoft.

Importantly, the direction of any spillover bias is toward zero (i.e., conservative). If senior Microsoft employees join control firms, this mechanically increases the average tenure at those firms, shifting the synthetic control's tenure distribution rightward. This would reduce the estimated gap between Microsoft and its synthetic control, making our treatment effect estimates conservative.

Table \ref{tab:spillover_robustness} presents a bounding exercise where we re-estimate our main specification after excluding the five control firms that received the most Microsoft employees post-RTO. The permutation test p-value \revision{remains stable with Microsoft ranked second ($2/32\approx0.06, 2/27\approx 0.07$)}, confirming that spillovers do not drive our results. 

\begin{table}[ht]
\centering
\begin{tabular}{lrrr}
  \hline
Company & MSFT Hires & Total Emp. & MSFT Share (\%) \\
  \hline
Docusign &   9 & 2,980 & 0.30 \\
  Airbnb &   9 & 3,073 & 0.29 \\
  Avanade &   3 & 2,696 & 0.11 \\
  Autodesk &   3 & 3,203 & 0.09 \\
  Splunk &   3 & 3,447 & 0.09 \\
  Amazon &  35 & 79,375 & 0.04 \\
  Salesforce &   5 & 22,422 & 0.02 \\
  Intel &   4 & 18,251 & 0.02 \\
  Oracle &   4 & 25,246 & 0.02 \\
  Dell &   2 & 15,800 & 0.01 \\
   \hline
\end{tabular}
\caption{Microsoft Employee Share in Control Firms (Post-RTO)}
\label{tab:spillover_magnitude}
\floatfoot{\textit{Note}: table shows top 10 control firms by share of post-RTO workforce that joined from Microsoft. MSFT Share is the percentage of the control firm's workforce that came from Microsoft in the quarter following the RTO mandate.}
\end{table}

\begin{table}[ht]
\centering
\begin{tabular}{lrr}
  \hline
Model & P-value (overall) & N Controls \\ 
  \hline
Full Sample & 0.06 & 31 \\
  Excluding Top 5 Spillover Recipients & 0.07 & 26 \\ 
   \hline
\end{tabular}
\caption{Robustness to Spillover Effects} 
\label{tab:spillover_robustness}
\end{table}

\clearpage

\section{Worker Separations and Joins}\label{app:flows}

This appendix provides a closer look at the RTO mandate's effects on worker separations and joins, to assess whether the compositional changes in the tenure distribution in the main results could reflect differential hiring rather than senior employee departures. We apply synthetic control methods to separation and hiring rates, analyze the composition of separating employees, and examine whether the synthetic control replicates seasonal hiring patterns.

\textit{Note:} For the aggregate flow rate analyses in this appendix, we apply the standard \citet{abadie2010synthetic} synthetic control methodology to scalar outcomes (e.g., separation rate, hiring rate), with weights optimized separately for each outcome. This differs from our main distributional synthetic controls (DiSCo) analysis, which estimates treatment effects on entire distributions. The standard SC approach is appropriate here because we are analyzing aggregate rates rather than distributional outcomes.

  \subsection{Aggregate Separation and Hiring Rates}\label{sec:aggregate_flows}                            
   
  Table \ref{tab:worker_flows_descriptive} provides descriptive context on the composition of worker flows before and after Microsoft's RTO mandate. Several patterns emerge. In Panel A, we show the tenure composition of leavers, i.e., what seniority employees who left their companies in the given time period were. At Microsoft, the share of leavers with 5+ years tenure increased from \revision{15.2\%} to 15.9\% post-RTO, while the share with 0--1 years \textit{decreased} from \revision{29.9\%} to 22.4\%. This shift towards fewer junior leavers and more senior leavers is consistent with our main finding that separations tilted toward more experienced employees. In contrast, control firms' tenure composition remained stable (35.0\% $\to$ 35.1\% for 0--1 years; 10.2\% $\to$ 9.2\% for 5+ years).
   Notably, Microsoft's separators are more senior on average than controls' even pre-RTO (\revision{15.2\% vs. 10.2}\% with 5+ years), reflecting differences in workforce composition. In Panel B, we report the net flows by title. Both Microsoft and control firms experienced strong positive net flows post-RTO, reflecting the broader tech hiring boom in early 2022. Entry-level net flows swung from negative to sharply positive at both Microsoft and controls, capturing the June new graduate hiring spike.              
  \begin{table}[htbp!]
\centering
\caption{Descriptive Statistics: Composition of Worker Flows}
\label{tab:worker_flows_descriptive}
\begin{tabular}{lcccc}
\toprule
& \multicolumn{2}{c}{Microsoft} & \multicolumn{2}{c}{Control Average} \\
\cmidrule(lr){2-3} \cmidrule(lr){4-5}
& Pre-RTO & Post-RTO & Pre-RTO & Post-RTO \\
\midrule
\multicolumn{5}{l}{\textit{Panel A: Tenure Composition of Separators (\%)}} \\[0.3em]
\quad 0--1 years & 29.9 & 22.4 & 35.0 & 35.1 \\
\quad 1--2 years & 21.4 & 23.8 & 25.0 & 25.7 \\
\quad 2--5 years & 33.5 & 37.9 & 29.7 & 30.0 \\
\quad 5+ years & 15.2 & 15.9 & 10.2 & 9.2 \\
\midrule
\multicolumn{5}{l}{\textit{Panel B: Net Flows by Title (Joins $-$ Separations)}} \\[0.3em]
\quad Entry (1--3) & $-$344 & $+$1,494 & $-$941 & $+$5,452 \\
\quad Mid-level (4--5) & $+$1,211 & $+$1,784 & $+$5,555 & $+$12,639 \\
\quad Senior (6--8) & $+$145 & $+$207 & $+$444 & $+$1,104 \\
\quad Executive (9--10) & $+$13 & $+$19 & $-$102 & $+$230 \\
\bottomrule
\end{tabular}
\floatfoot{\textit{Notes}: Panel A reports the tenure composition of employees separating from each firm, where a separation is the last month an individual is observed employed at the firm; Microsoft is compared to the pooled average across control firms. Panel B reports net worker flows by job-title level (joins minus separations), where a join is the first month an individual is observed at the firm. Title levels: 1--3 = entry/junior; 4--5 = mid-level/manager; 6--8 = senior/director; 9--10 = VP/executive. Pre-RTO is the six months through February 2022; post-RTO is March--June 2022 (March is grouped post-treatment to capture departure flows effective end of March ahead of the mandate). Control columns pool all 31 control firms.}
\end{table}

  While informative, these numbers do not tell us much about the counterfactual behavior of Microsoft's workforce in the absence of their mandate, and it is clear that the numbers are heavily influenced by seasonal patterns such as entry-level hiring of new graduates (see below).                                      
  To get a better sense of the causal effects of the mandate on Microsoft's aggregate hiring and firing flows, we apply the standard \citet{abadie2010synthetic} synthetic control methodology separately to separation rates and hiring      
  rates (inflows and outflows as a percent of total workforce). We also report p-values for the classical permutation test for synthetic controls. Table     
  \ref{tab:worker_flows} reports the results, which indicate that the total amount of hiring and firing taking place after Microsoft's RTO was not significantly different from its synthetic control, though we see a spike in hiring for both Microsoft and the control firms in the post-treatment period, driven by hiring of new graduates (see Section \ref{sec:graduate_hiring_spike}).

  Indeed, as we show in Section \ref{sec:separation_composition} below, the RTO effect operates through the \textit{composition} of            
  separations (i.e., who leaves) rather than through differential aggregate rates. This distinction is         
  important: even with similar overall separation rates, the mix of leavers shifted toward more senior     
  employees at Microsoft. The RTO effect operates not through differential \textit{volumes} of hiring or   
  separation, but through the \textit{composition} of who leaves, as documented in Table \ref{tab:separation_composition}.                                                                                                                
  Figure \ref{fig:flows_sc} shows the synthetic control estimates for separation and hiring rates over     
  time. The left panel shows Microsoft's separation rate relative to its synthetic control; the right panel
   shows the same for hiring rates. The synthetic control achieves excellent pre-treatment fit for both    
  outcomes.                                                                            
  \begin{figure}[hb!]             
      \centering                          
      \begin{minipage}{0.48\textwidth}            
          \centering              \includegraphics[width=\textwidth]{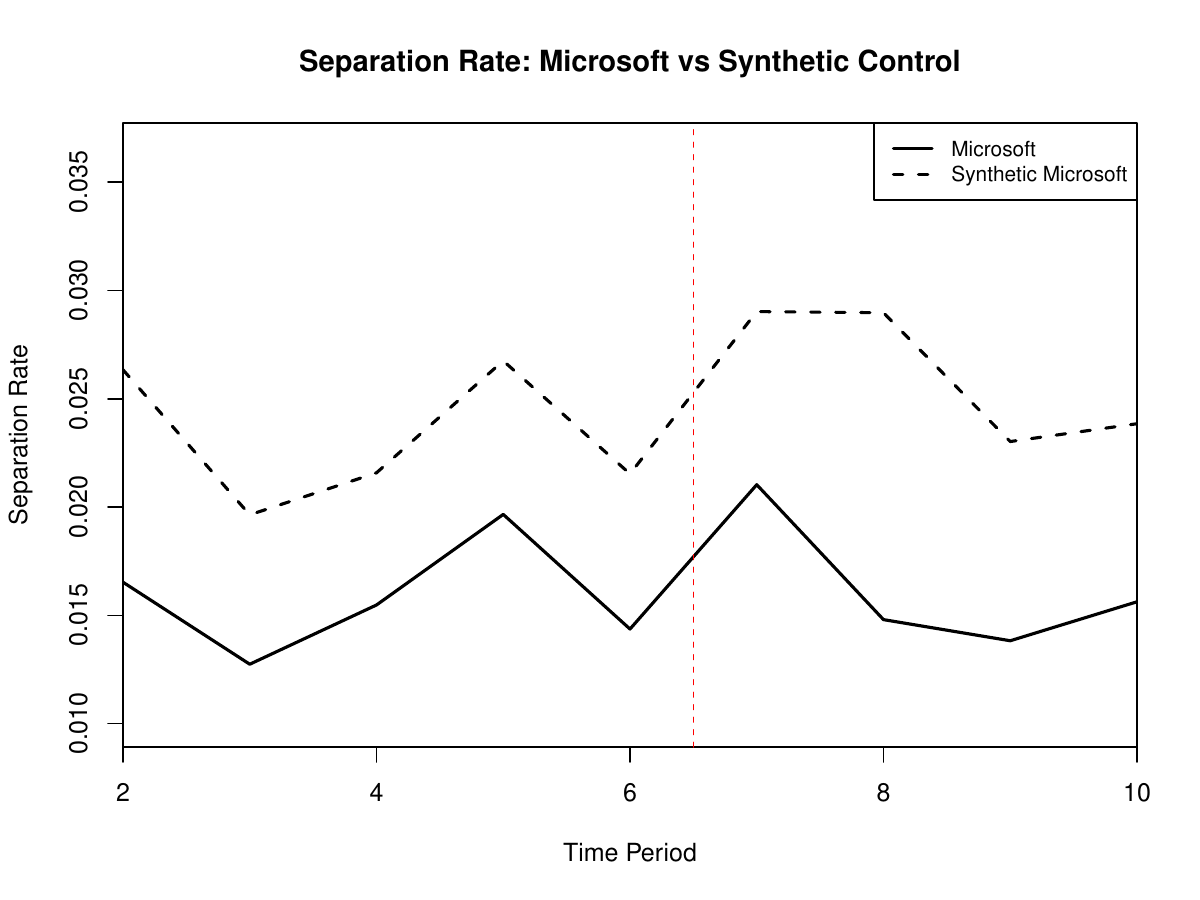}                        
      \end{minipage}                                                                                       
      \hfill                                                                                               
      \begin{minipage}{0.48\textwidth}                                                                     
          \centering                                                                                       
          \includegraphics[width=\textwidth]{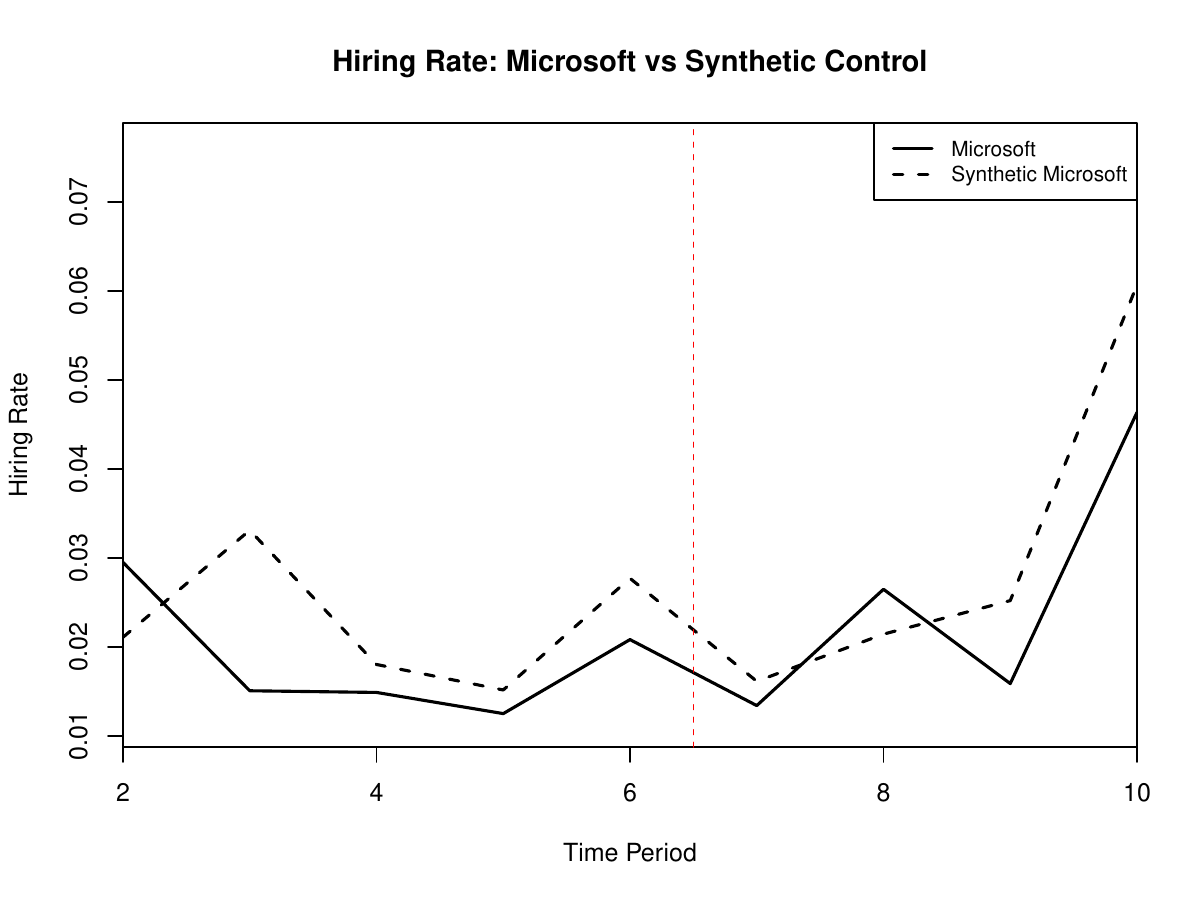}                            
      \end{minipage}                                                                                       
      \caption{Synthetic Control Analysis of Worker Flow Rates}                                            
      \label{fig:flows_sc}                                                                                 
      \floatfoot{\textit{Note}: Left panel shows Microsoft's monthly separation rate (solid) versus        
  synthetic control (dashed). Right panel shows hiring rates. Vertical line indicates RTO implementation   
  (April 2022). Shaded area indicates post-treatment period.}                                              
  \end{figure} 

\subsection{Composition of New Hires and Departures}\label{sec:separation_composition}

To directly test whether senior employees left at elevated rates, we analyze the composition of new and separating employees---that is, the characteristics of those who joined and left, conditional on joining or leaving. 

Table \ref{tab:separation_composition} reports synthetic control estimates for the share of separating employees with high tenure or senior titles.

\begin{table}[htbp!]
\centering
\footnotesize
\caption{Composition of Separating Employees: Synthetic Control Estimates}
\label{tab:separation_composition}
\begin{tabular}{lcc}
\toprule
Outcome & Treatment Effect & $p$-value \\
\midrule
Mean Tenure of Leavers (days) & $+110.8$ & \textbf{0.032}** \\
Share Senior Title ($\geq$5) & $+8.4$ pp & 0.097* \\
Share Manager+ Title ($\geq$6) & $+0.4$ pp & 0.387 \\
Mean Title Level & $+0.26$ & 0.129 \\
\bottomrule
\end{tabular}
\floatfoot{\textit{Note}: Table reports synthetic control treatment effect estimates for the composition of employees separating from Microsoft. Outcomes measure characteristics of employees who left, conditional on leaving. ``Mean Tenure of Leavers'' is measured in days; ``Share Senior Title'' and ``Share Manager+ Title'' are the percentage of leavers with title level $\geq 5$ (Senior IC and above) and $\geq 6$ (Owner/Director-level and above), respectively, in percentage points (pp); ``Mean Title Level'' is on the 1--10 title scale. $p$-values are computed using permutation inference with placebo firms; the mean-tenure effect uses the RMSPE-ratio statistic of \citet{abadie2010synthetic}. Estimates use quarterly aggregation. * $p < 0.10$; ** $p < 0.05$; *** $p < 0.01$.}
\end{table}

The key finding is that the mean tenure of Microsoft employees who left increased by \revision{110.8 days} relative to the synthetic control ($p = 0.032$). The share of separating employees with senior titles (level 5+) increased by \revision{8.4} percentage points ($p = 0.097$). These results provide direct evidence that the RTO mandate shifted the composition of departures toward more senior, higher-tenure employees.

If Microsoft were simultaneously pursuing a junior hiring strategy, we would expect to see changes in the composition of new hires. Table \ref{tab:hire_composition} reports synthetic control estimates for the share of new hires with senior titles.

\begin{table}[htbp!]
\centering
\footnotesize
\caption{Composition of New Hires: Synthetic Control Estimates}
\label{tab:hire_composition}
\begin{tabular}{lcc}
\toprule
Outcome & Treatment Effect & $p$-value \\
\midrule
Share Senior Title ($\geq$5) & $-4.1$ pp & 0.903 \\
Share Manager+ Title ($\geq$6) & $-2.3$ pp & 0.806 \\
Mean Title Level & $-0.36$ & 0.871 \\
\bottomrule
\end{tabular}
\floatfoot{\textit{Note}: Table reports synthetic control treatment effect estimates for the composition of new hires at Microsoft. Outcomes measure characteristics of employees who joined, conditional on joining. ``Share Senior Title'' and ``Share Manager+ Title'' are the percentage of new hires with title level $\geq 5$ (Senior IC and above) and $\geq 6$ (Owner/Director-level and above), respectively, in percentage points (pp); ``Mean Title Level'' is measured on the 1--10 title scale. $p$-values are computed using permutation inference with placebo firms. None of the effects are statistically significant, indicating that the composition of new hires was unchanged after the RTO mandate.}
\end{table}

We find no significant changes in the composition of new hires for any specification. The share of senior employees among new hires did not change significantly ($-4.1$ pp, $p = 0.903$), nor did mean title level ($-0.36$, $p = 0.871$). This null result rules out the alternative explanation that compositional shifts in the workforce were driven by differential hiring of junior employees rather than departures of senior employees.

\subsection{Seasonal Hiring Patterns: June New Graduate Spike} \label{sec:graduate_hiring_spike}

A potential concern is whether our main distributional synthetic control adequately captures seasonal hiring patterns, particularly the June spike in entry-level hiring associated with new college graduates. To address this, we use the same DiSCo weights from our main tenure distribution analysis to construct a synthetic control for hiring patterns.

Table \ref{tab:disco_june} shows the validation results. We compare entry-level hiring shares for Microsoft versus the DiSCo-weighted synthetic controls.

\begin{table}[htbp!]
\centering
\footnotesize
\caption{June Hiring Spike: Microsoft vs. DiSCo-Weighted Synthetic Control}
\label{tab:disco_june}
\begin{tabular}{lccc}
\toprule
& Microsoft & DiSCo-Weighted SC & Gap \\
\midrule
June entry-level share (\%) & 62.6 & 46.4 & $+$16.2 \\
Other months share (\%) & 3.1 & 2.7 & $+$0.4 \\
June spike (pp) & $+$59.5 & $+$43.7 & $+$15.8 \\
\midrule
Placebo $p$-value (June gap) & \multicolumn{3}{c}{0.22 \ (rank 7/32)} \\
\bottomrule
\end{tabular}
\floatfoot{\textit{Note}: Table compares the share of new hires at entry level (title levels 1--3) for Microsoft versus a synthetic control constructed using the same DiSCo weights from our main tenure distribution analysis.}
\end{table}

We find that the DiSCo-weighted synthetic control exhibits a \revision{similarly large} June spike as Microsoft. Microsoft's entry-level share jumps from 3.1\% in non-June months to 62.6\% in June (a 59.5 pp spike), while the DiSCo-weighted SC jumps from 2.7\% to 46.4\% (a 43.7 pp spike). \revision{Moreover, the difference in the precise magnitudes is not statistically significant: a permutation test that builds an analogous synthetic control for each control firm and treats its June entry-level prediction error as a placebo yields a $p$-value of 0.22 for Microsoft, so the seasonal spike is reproduced  within sampling error.}

This directly addresses the concern about differential junior hiring: the June new graduate hiring pattern is accounted for by the synthetic control counterfactual, and so does not drive our results. 

Figure \ref{fig:entry_hiring} visualizes this pattern, showing that Microsoft's entry-level hiring spike in June is matched by its DiSCo-weighted synthetic control.

\begin{figure}[ht!]
    \centering
    \includegraphics[width=0.8\textwidth]{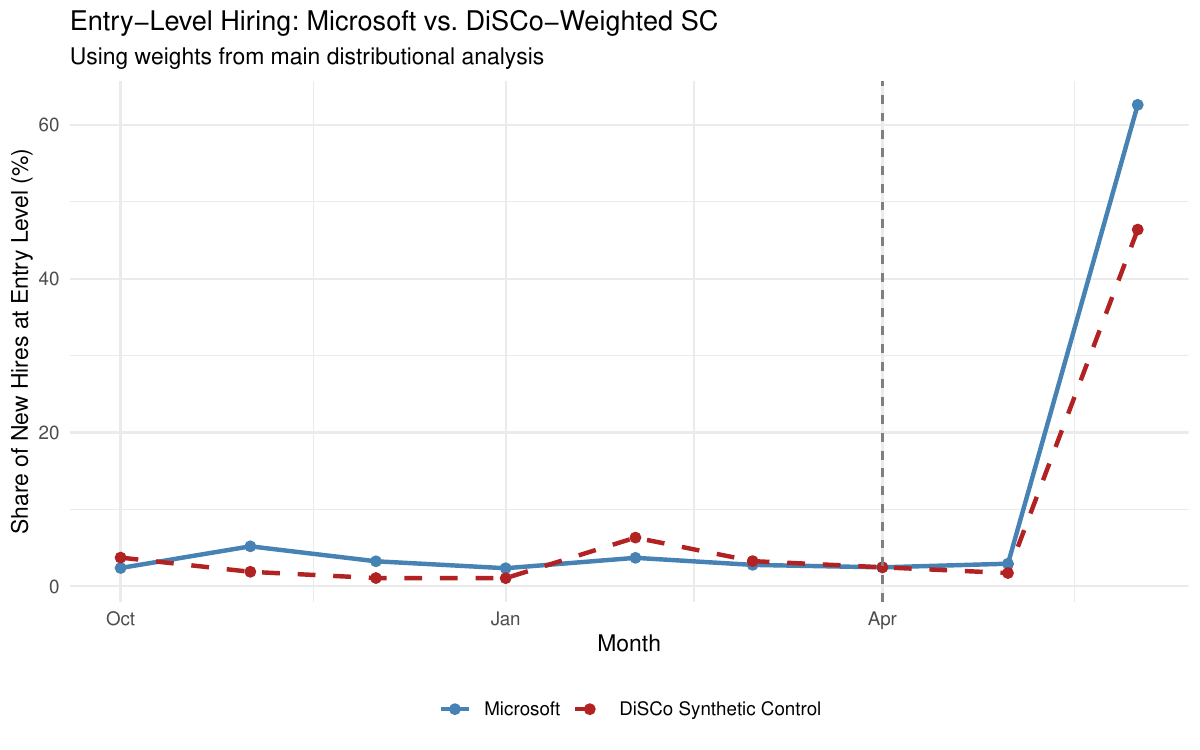}
    \caption{Entry-Level Hiring: Microsoft vs. DiSCo-Weighted Synthetic Control}
    \label{fig:entry_hiring}
    \floatfoot{\textit{Note}: Figure shows the share of new hires at entry level (title levels 1--3) by month for Microsoft (solid blue) and its DiSCo-weighted synthetic control (dashed red), constructed using the same weights from the main tenure distribution analysis. Both series exhibit pronounced June spikes corresponding to new graduate hiring. Vertical dashed line indicates RTO implementation (April 2022).}
\end{figure}

\subsection{Discussion}

The analyses in this appendix and Table \ref{tab:worker_flows} establish several key facts:
\begin{enumerate}
    \item \textit{Aggregate rates show null effects:} Neither Microsoft's separation rate ($p = 0.219$) nor its hiring rate ($p = 1$) differs significantly from controls. The RTO did not cause a detectable change in how many employees left or were hired in aggregate.

    \item \textit{Composition of leavers shifted toward seniors:} The mean tenure of Microsoft employees who separated increased by 110.8 days relative to the synthetic control ($p = 0.032$), and the share with senior titles increased by 8.4 pp ($p = 0.097$).

    \item \textit{No change in new hire composition:} The share of senior employees among new hires did not change significantly, providing no evidence of a concurrent junior hiring strategy. Together with the previous bullet, this suggests that senior employees left at higher rates, but this was offset by a decreased departure rate among junior employees, explaining the null effects on the overall hiring and separation rates. This finding is in line with recent literature which found that the gains from office proximity are concentrated among less-tenured and younger employees who are building human capital \citep{emanuel2023power}, providing a potential explanation for the differential departures we estimate after a return to office.

    \item \textit{Seasonal patterns accounted for:} The DiSCo-weighted synthetic control exhibits a \revision{large} June new graduate spike, \revision{close to that of} Microsoft (46.4\% vs.\ 62.6\%), demonstrating that seasonal hiring patterns do not drive our distributional effects.
\end{enumerate}

Together, these results support our interpretation that RTO mandates changed \textit{who} left but not \textit{how many} left. The composition of departures shifted toward more senior, higher-tenure employees, while aggregate separation and hiring rates remained similar to control firms.

\section{Layoff Comparison: Voluntary vs. Forced Attrition}
\label{app:layoff_comparison}

A natural question is whether the distributional shifts we document are specific to RTO-induced voluntary attrition or would arise from any workforce reduction. To address this, we apply our methodology to Microsoft's January 2023 layoff, in which the company announced the elimination of approximately 10,000 positions \citep{Microsoft2023Layoffs}.

We construct a control sample of 38 firms that did not conduct layoffs within two months of Microsoft's announcement, and estimate distributional synthetic controls for both tenure and title distributions. Table \ref{tab:layoff_comparison} compares the results.

\begin{table}[htbp!]
\centering
\begin{tabular}{lcc}
\hline
 & RTO (Apr 2022) & Layoff (Jan 2023) \\
\hline
Tenure p-value & 0.062 & 0.051 \\
Title p-value & 0.031 & 0.103 \\
N controls & 31 & 38 \\
\hline
Nature of attrition & Voluntary & Forced \\
\hline
\end{tabular}
\caption{Distributional Effects: RTO vs Layoff at Microsoft}
\label{tab:layoff_comparison}
\floatfoot{\textit{Note}: Table compares distributional effects of Microsoft's April 2022 RTO mandate with its January 2023 layoff of 10,000 employees. P-values are from permutation tests using distributional synthetic controls \citep{gunsilius2023distributional}. The RTO produced a significant shift in the title distribution and a tenure effect concentrated above the median. The layoff effects, by contrast, reflect employer-determined separations.
}
\end{table}

The comparison reveals a particularly clear difference in the title distribution. While the RTO produced a significant shift in titles ($p=0.03$), the layoff did not ($p=0.103$). The evidence for tenure is less distinct based on the overall permutation tests, with $p=0.06$ for the RTO and $p=0.051$ for the layoff, although the RTO effect is concentrated above the median ($p=0.03$) and the layoff confidence bands include zero nearly everywhere. Figure \ref{fig:layoff_comparison} visualizes these differences.

\begin{figure}[htbp!]
\centering
\begin{subfigure}[b]{0.48\textwidth}
    \includegraphics[width=\textwidth]{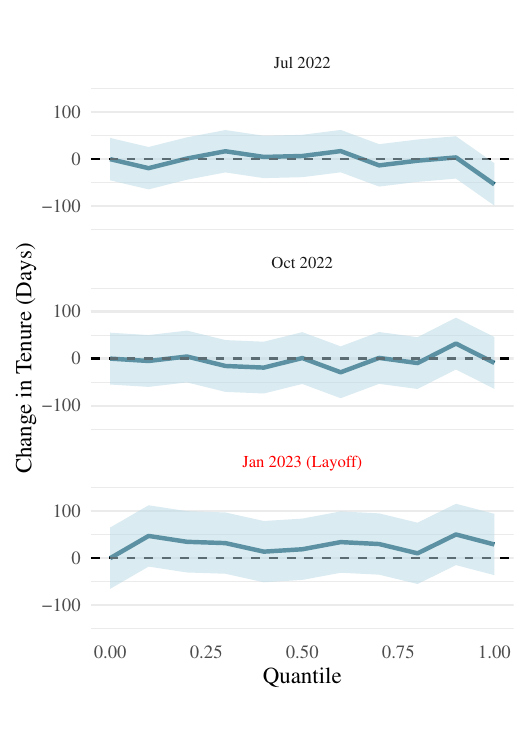}
    \caption{Tenure Distribution}
\end{subfigure}
\begin{subfigure}[b]{0.48\textwidth}
    \includegraphics[width=\textwidth]{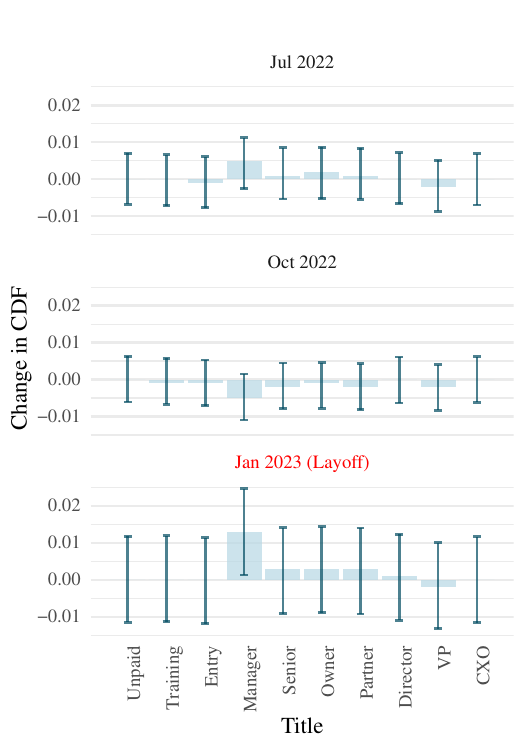}
    \caption{Title Distribution}
\end{subfigure}
\caption{Microsoft Tenure and Title Distributions After January 2023 Layoff}
\label{fig:layoff_comparison}
\floatfoot{\textit{Note}: Figure shows distributional synthetic control estimates for Microsoft's January 2023 layoff. Left panel shows changes in tenure quantile functions; right panel shows changes in title CDFs. Unlike the RTO (Figures \ref{fig:main_tenure}--\ref{fig:main_titles}), the layoff does not produce significant shifts in the title distribution, suggesting that layoffs affected workers more evenly across seniority levels.}
\end{figure}

This comparison provides several insights. First, it demonstrates that our methodology can distinguish between different types of workforce composition shocks. Second, the lack of significant title effects from the layoff---in contrast to the strong RTO effects---supports our interpretation that the RTO caused a selective ``brain drain'' of senior employees rather than generic workforce turnover. Layoffs, being employer-determined, appear to affect workers more evenly across the seniority distribution, while voluntary departures following the RTO were concentrated among senior staff with better outside options.

\clearpage

\section{External Validity: 2023 RTOs in a Slack Labor Market}
\label{app:external_validity}

A natural concern is whether our results from 2022 generalize to different labor market conditions. To address this, we replicate our analysis for two major tech RTOs that occurred in 2023: Amazon (May 2023) and Google (June 2023). These RTOs provide a useful contrast because they occurred in a markedly different economic environment.
\revision{This exercise applies our main sample-construction logic with one deliberate exception: we drop the treated-unit requirement of no surrounding layoffs, which Amazon and Google necessarily violate, while building each firm's donor pool with the same restrictions used throughout (NAICS 51/54/33, never- or not-yet-treated donors, and a minimum resume count relative to the treated firm).
Among 2023 adopters meeting these criteria, we focus on Amazon and Google because
they have by far the largest within-firm samples and thus the most power; other
contemporaneous adopters such as Dell and DocuSign (whose mandates fell in the same
two months) have much smaller samples and would yield even noisier estimates.}

\subsection{Labor Market Context}

Table \ref{tab:labor_market} documents the dramatic shift in labor market conditions between our 2022 RTOs and the 2023 RTOs.

\begin{table}[htbp!]
\centering
\caption{Labor Market Conditions: 2022 vs 2023}
\label{tab:labor_market}
\begin{tabular}{lccc}
\toprule
Metric & Q2 2022 & Q2 2023 & Change \\
\midrule
Tech Unemployment Rate & 1.8\% & 2.3\% & +28\% \\
Total Job Openings (millions) & 11.5 & 9.5 & $-$17\% \\
Quits Rate & 2.8\% & 2.4\% & $-$14\% \\
Tech Layoffs (preceding quarter) & 10,000 & 167,000 & +1,570\% \\
Share of Russell 3000 with RTO & 16.2\% & 20.1\% & +24\% \\
\bottomrule
\end{tabular}
\floatfoot{\textit{Note}: Tech unemployment from CompTIA analysis of BLS Current Population Survey data (June of each year). Job openings and quits rate from JOLTS (quarterly averages). Tech layoffs from layoffs.fyi (Q1 of each year, preceding the RTOs). RTO adoption calculated as share of Russell 3000 firms with Hybrid or In-Person policies announced by end of Q2, from \citet{flynn2024determinants}.}
\end{table}

The 2023 labor market was characterized by substantially higher unemployment, fewer job openings, lower quit rates, and a wave of tech layoffs that had begun in late 2022. Crucially, by mid-2023, a \revision{large share} of tech firms had already implemented some form of RTO \revision{(approximately 35\% of the tech firms in the Russell 3000 had any type of RTO policy in 2023 according to the data used in Table \ref{tab:labor_market})}, reducing the availability of remote-friendly alternatives for dissatisfied employees.

In particular, both Amazon and Google conducted massive layoffs in close temporal proximity to their RTO mandates (and so did all other major tech firms in our sample), as shown in Table \ref{tab:layoff_timeline}.

\begin{table}[htbp!]
\centering
\caption{Layoff Timeline for 2023 RTO Firms}
\label{tab:layoff_timeline}
\begin{tabular}{llr}
\toprule
Company & Date & Layoffs \\
\midrule
\multicolumn{3}{l}{\textit{Amazon (RTO: May 1, 2023)}} \\
& November 16, 2022 & 10,000 \\
& January 4, 2023 & 8,000 \\
& March 20, 2023 & 9,000 \\
& \textbf{Total} & \textbf{27,000} \\
\midrule
\multicolumn{3}{l}{\textit{Google (RTO: June 7, 2023)}} \\
& January 20, 2023 & 12,000 \\
& \textbf{Total} & \textbf{12,000} \\
\bottomrule
\end{tabular}
\floatfoot{\textit{Note}: Layoff data from layoffs.fyi. Amazon's RTO came just 6 weeks after its last major layoff round. Google's RTO came 5 months after its major layoff.}
\end{table}

These layoffs likely shifted the tenure and title distributions independently of any RTO effect, making it impossible to cleanly identify the causal impact of the RTO mandate alone. Nonetheless, to get a sense of the external validity of our findings, we replicate our synthetic controls analysis for these two firms. \revision{As in the main results, we exclude control firms that have a layoff in the sample window or already had an RTO active.}

\subsection{Results}

Table \ref{tab:rto_comparison} reports the permutation test p-values alongside our 2022 results.

\begin{table}[htbp!]
\centering
\caption{RTO Effects: 2022 (Tight Labor Market) vs 2023 (Slack Labor Market)}
\label{tab:rto_comparison}
\begin{tabular}{llccc}
\toprule
Company & Year & Labor Market & Tenure P-value & Title P-value \\
\midrule
Microsoft & 2022 & Tight & 0.06 & 0.03 \\
Apple & 2022 & Tight & --- & 0.07 \\
SpaceX & 2022 & Tight & --- & 0.05 \\
\midrule
Amazon$^*$ & 2023 & Slack & 1 & 0.75 \\
Google$^*$ & 2023 & Slack & 0.48 & 0.24 \\
\bottomrule
\end{tabular}
\floatfoot{\textit{Note}: Permutation test p-values for distributional treatment effects on tenure and title. ``Tight'' and ``Slack'' refer to labor market conditions: Q2 2022 had 1.8\% tech unemployment and 10K tech layoffs; Q2 2023 had 2.3\% tech unemployment and 167K tech layoffs. $^*$Amazon and Google 2023 RTOs are confounded by concurrent layoffs (27,000 and 12,000 employees, respectively, in the six months preceding their RTOs), making clean identification impossible.}
\end{table}

Strikingly, we find \textit{no} significant distributional effects for either Amazon or Google in 2023. The tenure p-values are 1.00 (Amazon) and 0.48 (Google); title p-values are 0.75 and 0.24, respectively, far from being statistically  significant. Figure \ref{fig:amazon_2023} shows the estimated quantile treatment effects with confidence bands.

\begin{figure}[htbp!]
\centering
\subfloat[Amazon]{\includegraphics[width=0.48\textwidth]{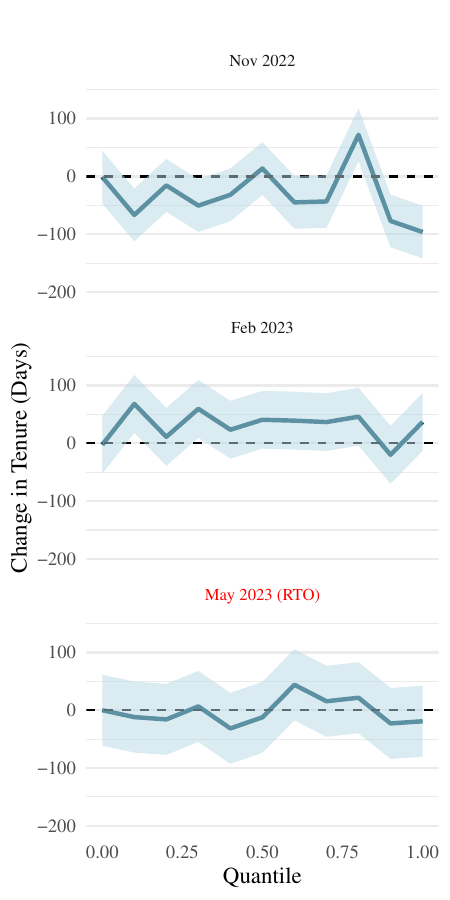}}
\hfill
\subfloat[Google]{\includegraphics[width=0.48\textwidth]{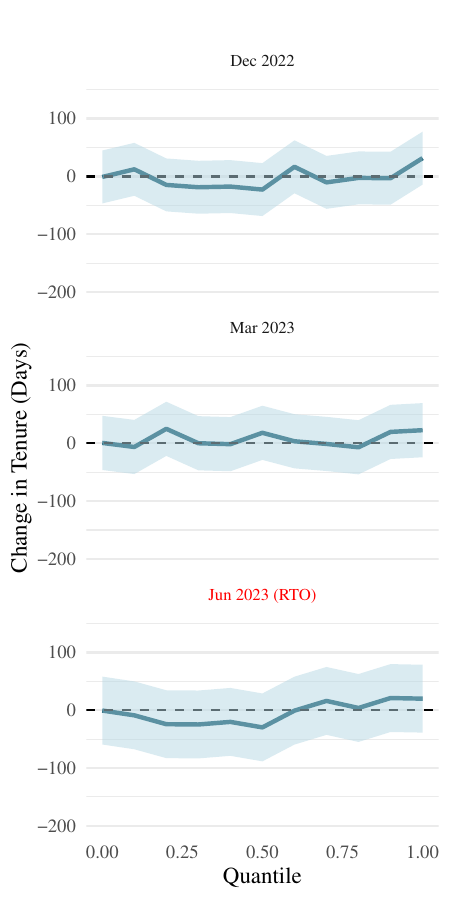}}
\caption{External Validity: Amazon \& Google 2023 RTO}
\label{fig:amazon_2023}
\floatfoot{\textit{Note}: Estimated quantile treatment effects for Amazon \& Google's 2023 RTO with 95\% confidence bands. Neither tenure distributions show significant effects (p = 1.00 and 0.48, respectively). Results are confounded by 27,000 \& 12,000 layoffs, respectively, in the preceding six months.}
\end{figure}

\subsection{Interpretation}

We interpret the null results for 2023 cautiously, as they admit two distinct explanations.

\begin{enumerate}
    \item \textit{Layoff confounding}: The massive layoffs at Amazon and Google in early 2023 likely already shifted the workforce composition before the RTO took effect. If layoffs disproportionately affected certain tenure or seniority groups, this would contaminate any RTO effect.
    \item \textit{Labor market channel}: In the more slack 2023 labor market, employees had fewer outside options. With higher unemployment, fewer job openings, and more firms requiring office attendance, dissatisfied senior employees may have been unable to \revision{leave for remote-friendly alternatives.}
\end{enumerate}

The labor-market interpretation is consistent with the 2022 effects being larger in settings with stronger outside options, although the confounding from contemporaneous layoffs prevents a formal comparison of effect magnitudes across the two periods.
\revision{The 2022 RTOs occurred in a high-outside-option environment: labor markets were tight, remote jobs were relatively abundant, and many major technology firms had not yet returned to the office. The significant effects we estimate should therefore be interpreted as the realized compositional response to an RTO when employees could actually leave their firms. As that margin becomes constrained---through layoffs, labor-market slack, or industry-wide RTO adoption---the magnitude of the response is likely to diminish.}
 While Amazon's pre-treatment tenure distributions in Figure \ref{fig:amazon_2023} are not well-replicated, likely because of the layoffs, Google's are. In line with the expectation that a lack of outside options attenuates our paper's main effects, the treatment effects post-RTO are not significantly different from zero across the distribution.

\subsection{RTO Stringency Gradient}

Additional support for this interpretation comes from examining the relationship between RTO stringency and effect magnitude across our 2022 firms:

\begin{table}[htbp!]
\centering
\caption{RTO Stringency and Effect Magnitude (2022 Firms)}
\label{tab:stringency}
\begin{tabular}{lccc}
\toprule
Company & Days Required & Title Effect (\%) & P-value \\
\midrule
Apple & 1 day/week & $\sim$2\% & 0.07 \\
Microsoft & 2.5 days ($\sim$50\%) & $\sim$2\% & 0.03 \\
SpaceX & 5 days (full-time) & $\sim$5\% & 0.046 \\
\bottomrule
\end{tabular}
\end{table}

The effect magnitude increases with the stringency of the RTO mandate. This pattern suggests that distributional effects scale with the ``cost'' imposed on employees, consistent with a model where employees trade off commuting costs against the value of their current position. As the cost of staying increases (more required office days), more employees---particularly those with better outside options---choose to leave.

\section{Bootstrapping confidence intervals for the constructed quantile functions} \label{app:bootstrap}
In this section we provide assumptions on the data under which a canonical bootstrap procedure is valid for obtaining statistical uniform confidence intervals on the constructed quantile functions 
\begin{equation}\label{eq:count}
F^{-1}_{0t,N} \coloneqq \sum_{j=1}^J \lambda^*_j F^{-1}_{jt},
\end{equation} 
where $\lambda^*_j$ are the optimal weights \citep{gunsilius2023distributional}. Throughout this appendix, let $\mathcal{Q}\subseteq[0,1]$ denote the
quantile-index interval over which the constructed quantile functions are
considered. In the compact-support case we take
$\mathcal{Q}=[0,1]$, while in the unbounded-support case we take
$\mathcal{Q}=[q_1,q_2]\subset(0,1)$.

In the case where $t\leq t^*$, this is the ``replicated'' quantile function of the observable $F^{-1}_{0t}$; in the case $t>t^*$, it is the counterfactual quantile function. The bootstrap procedure is captured in Algorithm \ref{algo:sub}. The novelty of this result is that we consider the quantile functions as functional data themselves, which is different from standard approaches such as \citet{fang2019inference}, which are defined for \emph{fixed} quantile functions, not random quantile processes.

\begin{algorithm}[htbp!]
\caption{The Bootstrap for Distributional Synthetic Controls}
\label{algo:sub}
\begin{algorithmic}[1]
\STATE \textbf{Input}: Outcome data $\left\{ \left(Y_{ijt}\right)_{i=1,\ldots,N_{jt}} \right\}_{j=0,\ldots,J;\: t=1,\ldots,T},$ quantile grid $\mathcal{Q} \coloneqq \{ \frac{q}{G} : q = 1,\ldots, G \}$,  empirical quantile functions $\hat{\mathbf{F}}^{-1}_{tn} \coloneqq \left(\hat{F}_{jtn}^{-1}(q) : q \in \mathcal{Q} \right)_{j=0,\ldots,J}$, index vector $\mathcal{I}_{jt} \coloneqq \{1, \ldots, N_{jt}\}$, number of bootstrap draws $B$, significance level $\alpha \in(0,1)$, time-specific weights estimators $\hat{\vec{\lambda^*}}_{tn}\left(\hat{\mathbf{F}}^{-1}_{tn}\right) \coloneqq \left(\hat{\lambda}^*_{jtn}(\hat{F}_{jtn}^{-1})\right)_{j=1,\ldots,J}$, the derived average weights $\hat{\vec{\lambda^*_{n}}} \coloneqq (\hat{\lambda^*_{1n}},\ldots,\hat{\lambda^*_{Jn}}) \coloneqq \frac{1}{t^*}\sum_{t=1}^{t^*} \hat{\vec{\lambda^*}}_t$, and counterfactual quantile estimator $\hat{F}^{-1}_{0tn,N} \coloneqq \sum_{j=1}^J \hat{\vec{\lambda^*}}_{jn} \hat{F}^{-1}_{jtn} $
\STATE \textbf{Output:} Confidence bands for the average weights $\hat{\vec{\lambda}}^*_n$, counterfactual quantile function $\left\{\hat{F}^{-1}_{0tn,N}\right\}_{t=1,\ldots,T}$, and counterfactual quantile differences $\left\{\hat{F}^{-1}_{0tn} - \hat{F}^{-1}_{0tn,N}\right\}_{t=1,\ldots,T}$.
    \FOR{$b = 1$ to $B$} 
        \FOR{$t=1$ to $T$}
            \STATE For each $j=0,\ldots,J$, draw a new sample $\tilde{\mathcal{I}}^b_{jt}$ of size $N_{jt}$ \revision{with} replacement from $\mathcal{I}_{jt}$ 
            \STATE Estimate the new empirical quantile functions $\tilde{\mathbf{F}}^{-1,b}_{tn}$ from the resampled data $\left\{\left(Y_{ijt}\right)_{i\in \tilde{\mathcal{I}}_{jt}} \right\}_{j=0,\ldots,J}$
            \STATE Compute the optimal time-specific weights $\tilde{\vec{\lambda^{*,b}}}_{tn}\left(\tilde{\mathbf{F}}^{-1,b}_{tn}\right)$ from the resampled quantile functions for $t \leq t^*$
        \ENDFOR
        \STATE Calculate the resampled average weights $\tilde{\vec{\lambda^{*,b}}}_{n}$ from the resampled time-specific weights
        \STATE Construct the counterfactual quantile function $\left\{\tilde{F}^{-1,b}_{0tn,N}\right\}_{t=1,\ldots,T}$ and quantile differences $\left\{\tilde F^{-1,b}_{0tn}-\tilde F^{-1,b}_{0tn,N}\right\}_{t=1,\ldots,T}$ from the resampled average weights and quantile functions
    \ENDFOR
    \STATE Construct confidence intervals for each object of interest $\mathbf{X}$ and its corresponding estimate $\hat{\mathbf{X}}$,
    \[ 
    C(\alpha) \coloneqq \left[ \hat{\mathbf{X}} - t_{1-\alpha}, \hat{\mathbf{X}} + t_{1-\alpha}   \right],
    \] 
    where $t_{1-\alpha}$ is the $1-\alpha$ quantile of the maximum absolute difference between the bootstrapped estimates and the main estimate $\mathbf{X}^b$.
\end{algorithmic}
\end{algorithm}

In the following, to simplify notation, we assume that for all elements $j=0,\ldots, J$ and for all time periods $t$ we have the same amount of observations $n_{jt} = n$. This restriction serves only to save on notation and the results can be extended to the case where we allow for different number of observations $n_{jt}$ as long as they are asymptotically balanced, i.e., as long as $\frac{n_{ks}}{n_{jt}}\to\eta_{jkst}$ with $0<\eta_{jkst}<+\infty$ for all $j,k\in\{0,\ldots,J\}$ and all $s,t\leq T$. 

Formally, validity of the bootstrap means that the bootstrapped empirical process
\[\tilde{\mathbb{G}}_{n} = \sqrt{n}\left(\tilde{F}^{-1}_{0tn,N} - \hat{F}^{-1}_{0tn,N}\right) = \sqrt{n}\left(\sum_{j=1}^J \tilde{\lambda}_{jn}^*\tilde{F}^{-1}_{jtn} - \sum_{j=1}^J \hat{\lambda}_{jn}^*\hat{F}^{-1}_{jtn}\right),\] converges weakly to the standard empirical process
\[\mathbb{G}_{n} = \sqrt{n}\left(\hat{F}^{-1}_{0tn,N} - F^{-1}_{0tn}\right) = \sqrt{n}\left(\sum_{j=1}^J \hat{\lambda}_{jn}^*\hat{F}^{-1}_{jtn} - \sum_{j=1}^J \lambda^*_jF^{-1}_{jtn}\right),\]
where $\tilde{\lambda}_{jn}$ and $\tilde{F}^{-1}_{jtn}$ are the bootstrapped optimal weights and quantile functions; $\hat{\lambda}_n$ and $\hat{F}^{-1}_{jtn}$ are the estimated optimal weights and quantile functions; and $\lambda^*_j$ and $F^{-1}_{jt}$ are the optimal weights and quantile functions in the population. 

Our approach to proving bootstrap validity is not to assume that the population objects $F^{-1}_{jt}$ are fixed. In this case, we would have to work with points of non-differentiability directly, like in, e.g., \citet{dumbgen1993nondifferentiable, fang2019inference}, but in abstract Banach spaces. Instead of this, we think about a more abstract sampling setting \emph{that includes randomness at the level of quantile functions}: a functional data approach. This means that we assume the quantile functions are random objects themselves, which is the correct setup for our setting. This simplifies statement and proofs but only allows us to make statements that the bootstrap works ``almost everywhere'' in the abstract space of random quantile functions we consider. The space we consider is the Hilbert product space $\bigtimes_{j=0}^J L^2(\mathcal{Q})$ equipped with the Euclidean product topology. We turn it into a Polish space, by equipping this space with the Borel $\sigma$-algebra induced by the corresponding norm-and product topology, which we denote by $\bigotimes_{j=0}^J \mathcal{B}(L^2(\mathcal{Q}))$.

For randomly sampled quantile functions $\{F_{jt}^{-1}\}_{j=0}^J$ whose joint law $\mu$ on
$\bigtimes_{j=0}^J L^2(\mathcal{Q})$ is absolutely continuous with respect to some non-degenerate Gaussian measure $\gamma$, the probability that the $\vec{\lambda}^*_t$ are not Hadamard-differentiable is zero. Indeed, the exceptional set is $\gamma$-null and hence also $\mu$-null. It therefore makes sense to allow randomness at the \emph{functional data level}, which is a novel methodological contribution.

Since we work in infinite dimensional spaces, this almost everywhere statement is different from classical ``almost everywhere'' statements, as an analogue to Lebesgue measure does not exist in infinite dimensional spaces. We rely on existing results showing that the exceptional set is Gauss-null for any non-degenerate Gaussian measure $\gamma$ on $\bigtimes_{j=0}^J L^2(\mathcal{Q})$. Consequently, it also has probability zero under any law $\mu$ that is absolutely continuous with respect to $\gamma$. We require the following assumption on the data generating process, the first part of which is a standard assumption for bootstrap validity of quantile estimators \citep[Lemma 3.9.23]{van1996weak}. 
\begin{assumption}\label{ass:boot}
For all $j=0,\ldots, J$ and all time periods $t\leq T$, we have the same number of observations $n$ and all observations $Y_{ijt}$ are iid draws from the distribution $F_{jt}$. In addition, the data-generating process is such that the corresponding distribution functions $F_{jt}$ either 
\begin{enumerate}
    \item[(i)] have compact support $[a,b]$ and are continuously differentiable with strictly positive density or 
    \item[(ii)] have infinite support and for every $0<q_1< q_2<1$ there exists an $\varepsilon>0$ such that $F_{jt}$ is continuously differentiable on the interval $[a,b]\equiv [F^{-1}_{jt}(q_1)-\varepsilon, F^{-1}_{jt}(q_2)+\varepsilon]$ with strictly positive density $f$.
\end{enumerate}
Moreover, the quantile functions are square-integrable in the sense that 
\[\int \left\lvert F_{jt}^{-1}(q)\right\rvert^2\dd q<+\infty\] for all $j,t$ and the $J\times J$ Hessian matrix $\mathbf{H}_t$ with $jk$-th entry 
\[\mathbf{H}_t(j,k) = \int F_{kt}^{-1}(q)F_{jt}^{-1}(q)\dd q\] is strongly positive definite for all $t$ in the sense that there exists $\alpha>0$ such that
\[\langle \mathbf{H}_t x, x\rangle\geq \alpha\|x\|_2^2\] for any $x\in\mathbb{R}^J$, where $\|\cdot\|_2$ denotes the Euclidean norm.
\end{assumption}
This assumption allows us to prove validity of the bootstrap for obtaining confidence intervals on the constructed quantile function; in the compact support case we can do this for the entire quantile range $[0,1]$, while in the case of unbounded support we can only do it for any compact subinterval of the open interval $(0,1)$. 
One can interpret the assumption on the matrix $\mathbf{H}_t$ as essentially a ``variance'' condition. It requires that there is enough variation in the quantile functions and mirrors the classical variance condition for the invertibility of the design matrix in classical linear regression. 

The following proposition is key to proving bootstrap validity. It
shows that the constructed quantile function is Hadamard differentiable
outside a null set with respect to any non-degenerate Gaussian measure
$\gamma$ on $\bigtimes_{j=0}^J L^2(\mathcal{Q})$. Consequently, the
same conclusion holds almost surely when the joint law of the random
quantile functions is absolutely continuous with respect to $\gamma$. Hadamard differentiability lets us apply the Delta method \citep[chapter 3.9]{van1996weak} to derive bootstrap validity for $\gamma$-almost every constellation of quantile functions. Since $F^{-1}_{0t,N}$ is a function of $\mathbf{F}_t^{-1}$, we work in the product spaces
$\bigtimes_{j=0}^J L^2(\mathcal{Q})$ and $\mathbf{C}(\mathcal{Q})\coloneqq
\bigtimes_{j=0}^J C(\mathcal{Q})$, the product space of square-integrable functions on $\mathcal{Q}$ and the product space of continuous functions on $\mathcal{Q}$, respectively, which we equip with the standard Euclidean product norm.

\begin{proposition}\label{lem:hada}
Let Assumption \ref{ass:boot} hold. For any non-degenerate Gaussian
measure $\gamma$ on
\[
\left(
\bigtimes_{j=0}^J L^2(\mathcal{Q}),
\bigotimes_{j=0}^J
\mathcal{B}\bigl(L^2(\mathcal{Q})\bigr)
\right),
\]
the constructed quantile functions \eqref{eq:count} are Hadamard
differentiable in $\bigtimes_{j=0}^J L^2(\mathcal{Q})$ tangentially to $\mathbf{C}(\mathcal{Q})$ outside a $\gamma$-null set. Consequently, if the joint law $\mu$ of $\mathbf{F}_t^{-1}\coloneqq (F^{-1}_{0t},\ldots,F^{-1}_{Jt})$ is absolutely continuous with respect to $\gamma$, the constructed quantile functions are Hadamard differentiable $\mu$-almost surely.
\end{proposition}
\begin{proof}
    Under Assumption \ref{ass:boot} the quantile functions $F^{-1}_{jt}$ are all Hadamard differentiable tangentially to $C(\mathcal{Q})$ in the space $L^2(\mathcal{Q})$ of square integrable functions with respect to Lebesgue measure \citep[Lemma 3.9.23]{van1996weak}. Since $C(\mathcal{Q})$ and $L^2(\mathcal{Q})$ are Banach spaces, we can use the chain rule for Hadamard derivatives in Banach spaces \citep[Lemma 3.9.3]{van1996weak} and can therefore focus on the Hadamard differentiability of the optimal $\vec{\lambda}_t^*(\mathbf{F}^{-1}_t)$ in $\mathbf{F}^{-1}_t = (F^{-1}_{0t},\ldots, F^{-1}_{Jt})$ in each time period $t$, obtained via
    \[\vec{\lambda}^*_t(\mathbf{F}_{t}^{-1}) = \argmin_{\lambda\in\Delta^J} f(\vec{\lambda}_t,\mathbf{F}^{-1}_t)\equiv \argmin_{\lambda\in\Delta^J}\int_0^1 \left\lvert \sum_{j=1}^J\lambda_{jt}F^{-1}_{jt}(q)-F^{-1}_{0t}(q)\right\rvert^2\dd q.\]

\noindent \emph{Step 1: Rewriting the constrained problem:}\\
To show Hadamard differentiability ``almost everywhere'', we use the fact that the optimization problem is strictly convex. Indeed, under Assumption \ref{ass:boot}, the Gram matrix $\mathbf{H}_t$ is strongly positive definite. Hence the objective is strictly convex in $\vec{\lambda}$, so the minimizer $\vec{\lambda}_t^*$ is unique. 
Then we write the constrained problem as
\[\vec{\lambda}^*_t(\mathbf{F}_{t}^{-1}) = \argmin_{\lambda\in\mathbb{R}^J} \int_0^1\left\lvert \sum_{j=1}^J\lambda_{jt}F^{-1}_{jt}(q)-F^{-1}_{0t}(q)\right\rvert^2\dd q + \chi_{\Delta^J}(\vec{\lambda}),\]
where $\chi_A(x)$ is the indicator function that is zero if $x\in A$ and $+\infty$ otherwise. 

The weights $\vec{\lambda}_t^*$ are optimal if and only if they satisfy the first order condition \citep[p.~149]{parikh2014proximal}
\[0\in \partial_{\vec{\lambda}} f(\vec{\lambda}_t^*,\mathbf{F}_t^{-1})+ \partial_{\vec{\lambda}_t^*}\chi_{\Delta^J}(\vec{\lambda}_t^*),\] where $\partial_{\vec{\lambda}}$ denotes the vector subgradient with respect to $\vec{\lambda}\in \mathbb{R}^J$. This is equivalent to the fixed point condition \citep[p.~150]{parikh2014proximal}
\begin{equation}\label{eq:FOC}
\vec{\lambda}_t^*(\mathbf{F}^{-1}_t) = \pi_{\Delta^J}\left(\vec{\lambda}_t^*(\mathbf{F}^{-1}_t) - \partial_{\vec{\lambda}}f(\vec{\lambda}_t^*,\mathbf{F}^{-1}_t)\right), 
\end{equation}
where $\pi_{\Delta^J}$ is the standard metric projection onto $\Delta^J$ using the Euclidean norm, which follows from e.g.~\citet[Proposition 2.4.12]{clarke1990optimization}.

Since $f$ is differentiable in $\vec{\lambda}$ its subgradient $\partial_{\vec{\lambda}}$ is single-valued and coincides with the standard gradient. It is a $1\times J$-vector, with $j$-th entry computed via 
\begin{align*}
    &\partial_{\lambda_j} \int \left\lvert \sum_{j=1}^J \lambda_{jt}F_{jt}^{-1}(q)-F_{0t}^{-1}(q)\right\rvert^2\dd q\\
    =&\int \partial_{\lambda_j}\left\lvert \sum_{j=1}^J \lambda_{jt}F_{jt}^{-1}(q)-F_{0t}^{-1}(q)\right\rvert^2\dd q\\
    =&2\int \left(\sum_{k=1}^J \lambda_{kt}F_{kt}^{-1}(q)-F_{0t}^{-1}(q)\right) F_{jt}^{-1}(q)\dd q,
\end{align*}
where the second line follows from the dominated convergence theorem in conjunction with the boundedness of the quantile functions under Assumption \ref{ass:boot}.

We now want to take the Hadamard derivative $D_{\mathbf{F}_t^{-1}}$ of this first-order condition at an optimal $\vec{\lambda}_t^*$ with respect to the $J+1$-dimensional vector $\mathbf{F}_t^{-1} = (F_{0t}^{-1},\ldots, F_{Jt}^{-1})$. The projection $\pi_{\Delta^J}$ is only Lipschitz continuous in its argument \citep[Lemma 6.54d]{aliprantis2006infinite} and hence by Rademacher's theorem only differentiable almost everywhere. 

We therefore analyze the right-hand side of \eqref{eq:FOC}. Formally, using the chain rule of Hadamard derivatives in Banach spaces \citep[Theorem 3.9.3]{van1996weak}, we have
\begin{multline*}
    D_{\mathbf{F}_t^{-1}}\left(\vec{\lambda}_t^*(\mathbf{F}^{-1}_t) - \partial_{\vec{\lambda}}f(\vec{\lambda}_t^*,\mathbf{F}_t^{-1})\right) \\= D_{\mathbf{F}_t^{-1}}\vec{\lambda}_t^*(\mathbf{F}_t^{-1}) - D_{\mathbf{F}_t^{-1}}\partial_{\vec{\lambda}} f(\vec{\lambda}_t,\mathbf{F}_t^{-1}) - \partial_{\vec{\lambda}}^2 f(\vec{\lambda}_t,\mathbf{F}_t^{-1}) D_{\mathbf{F}_t^{-1}}\vec{\lambda}(\mathbf{F}_t^{-1}),
\end{multline*}
where $D_{\mathbf{F}_t^{-1}}\partial_{\vec{\lambda}} f(\vec{\lambda}_t,\mathbf{F}_t^{-1})$ is arranged as a Jacobian of dimension $J\times (J+1)$, $\partial_{\vec{\lambda}}^2 f(\vec{\lambda}_t,\mathbf{F}_t^{-1})$ is a Hessian of dimension $J\times J$ and $D_{\mathbf{F}_t^{-1}}\vec{\lambda}(\mathbf{F}_t^{-1})$ is again a Jacobian of dimension $J\times (J+1)$. 

To compute the derivative with respect to the quantile functions, define
\[ R_t(q) \coloneqq \sum_{k=1}^J \lambda_{kt}F^{-1}_{kt}(q)-F^{-1}_{0t}(q).
\]
For a direction
$\mathbf{h}=(h_0,\ldots,h_J)\in\mathbf{C}(\mathcal{Q})$, the Hadamard
derivative of the $j$-th component of the gradient is
\begin{align*}
&D_{\mathbf{F}_t^{-1}} \partial_{\lambda_j} f(\vec{\lambda}_t,\mathbf{F}_t^{-1}) [\mathbf{h}]
\\
&\qquad= 2\int_{\mathcal{Q}} \left( \sum_{k=1}^J\lambda_{kt}h_k(q)-h_0(q) \right) F^{-1}_{jt}(q)\dd q + 2\int_{\mathcal{Q}}R_t(q)h_j(q)\dd q.
\end{align*}
Equivalently, the derivative with respect to the $k$-th quantile
function in a direction $h_k$ is, for $k=1,\ldots,J$,
\[ D_{F^{-1}_{kt}} \partial_{\lambda_j} f(\vec{\lambda}_t,\mathbf{F}_t^{-1}) [h_k] = 2\lambda_{kt} \int_{\mathcal{Q}} h_k(q)F^{-1}_{jt}(q)\dd q + 2\mathds{1}\{j=k\} \int_{\mathcal{Q}} R_t(q)h_k(q)\dd q,
\]
while, for $k=0$,
\[
D_{F^{-1}_{0t}} \partial_{\lambda_j} f(\vec{\lambda}_t,\mathbf{F}_t^{-1}) [h_0]
= -2 \int_{\mathcal{Q}} h_0(q)F^{-1}_{jt}(q)\dd q.
\]
These derivatives follow directly from the product rule and continuity
of the $L^2$ inner product. The $jk$-th entry of the Hessian with respect
to $\vec{\lambda}$ is
\[
\partial_{\lambda_k}\partial_{\lambda_j}
f(\vec{\lambda}_t,\mathbf{F}_t^{-1})
=
2\int_{\mathcal{Q}}
F^{-1}_{kt}(q)F^{-1}_{jt}(q)\dd q
=
2\mathbf{H}_t(j,k).
\]

The object of interest is $D_{\mathbf{F}_t^{-1}}\vec{\lambda}(\mathbf{F}_t^{-1})$, which is the Hadamard derivative of $\vec{\lambda}_t$ with respect to the quantile functions at the optimal value $\vec{\lambda}_t^*$. We want to use the fixed-point condition \eqref{eq:FOC} to show that it is Hadamard differentiable up to a negligible set. \\

\noindent \emph{Step 2: Analyzing $D_{F^{-1}_t}\vec{\lambda}_t(\mathbf{F}_t^{-1})$ at $\vec{\lambda}_t^*$.}\\
To fix ideas, consider the case where the term inside the projection in \eqref{eq:FOC} lies inside $\Delta^J$. In this case the projection is the identity mapping and a rearranging of \eqref{eq:FOC} gives
\[D_{\mathbf{F}_t^{-1}}\partial_{\vec{\lambda}} f(\vec{\lambda}_t,\mathbf{F}_t^{-1}) + \partial_{\vec{\lambda}}^2 f(\vec{\lambda}_t,\mathbf{F}_t^{-1}) D_{\mathbf{F}_t^{-1}}\vec{\lambda}(\mathbf{F}_t^{-1})=0.\] 
Since
$\partial_{\vec{\lambda}}^2 f(\vec{\lambda}_t,\mathbf{F}_t^{-1}) = 2\mathbf{H}_t$
and $\mathbf{H}_t$ is strongly positive definite under Assumption \ref{ass:boot}, the Hessian is invertible. Hence, in this interior case, $D_{F^{-1}_t}\vec{\lambda}_t(\mathbf{F}_t^{-1})$ is well defined.

We now consider the general case. We first show that $\vec{\lambda}_t^*(\mathbf{F}_t^{-1})$ defined via the fixed-point relation \eqref{eq:FOC} is Lipschitz continuous in $\mathbf{F}_t^{-1}$ \citep{dafermos1988sensitivity}. Then we rely on infinite-dimensional analogues of Rademacher's theorem \citep{aronszajn1976differentiability, lindenstrauss2003frechet, zajivcek2015hadamard} that guarantee that the optimal weights are Hadamard differentiable in $\mathbf{F}_t^{-1}$ up to a small set in the Banach space $\bigtimes_{j=0}^JL^2(\mathcal{Q})$, which will be $\gamma$-null.

The Lipschitz continuity of $\vec{\lambda}_t^*$ in $\mathbf{F}_t^{-1}$ follows from the exact same argument as Lemma 2.4 in \citet{dafermos1988sensitivity}, but in a Banach space setting; since the proof of Lemma 2.4 only works with norms, this argument directly extends to the Banach space setting. We only have to check the corresponding (Lipschitz-) continuity conditions. For this, define $g(\vec{\lambda}_t^*, \mathbf{F}_t^{-1})\coloneqq \partial_{\vec{\lambda}}f(\vec{\lambda}_t^*,\mathbf{F}_t^{-1})$ and note that the restriction on $\mathbf{H}_t$ implies that $g(\vec{\lambda}_t^*,\mathbf{F}_t^{-1})$ is strongly monotone in $\vec{\lambda}_t^*$ in the sense of equation (1.3) in \citet{dafermos1988sensitivity}. Indeed,
\begin{align*}
    \langle g(\vec{\lambda}_t^*,\mathbf{F}_t^{-1}) - g(\tilde{\lambda}_t^*,\mathbf{F}_t^{-1}), \vec{\lambda}_t^* - \tilde{\lambda}_t^*\rangle
    =& 2\sum_{j=1}^J\sum_{k=1}^J (\lambda_k - \tilde{\lambda}_k)(\lambda_j-\tilde{\lambda}_j)\int F_{kt}^{-1}(q)F_{jt}^{-1}(q)\dd q\\
    =& 2(\vec{\lambda}_t^* - \tilde{\lambda}_t^*)\mathbf{H}_t(\vec{\lambda}_t^* - \tilde{\lambda}_t^*)^\top \\
    \geq& \alpha\|\vec{\lambda}_t^* - \tilde{\lambda}_t^*\|_2^2
\end{align*}
for some $\alpha>0$. 

Analogously, the Lipschitz continuity of $g$ in $\vec{\lambda}_t^*$ for fixed $\mathbf{F}_t^{-1}$ as required in equation (1.4) in \citet{dafermos1988sensitivity} follows from
\begin{align*}
    &\left\|g(\vec{\lambda}_t^*,\mathbf{F}_t^{-1}) - g(\tilde{\lambda}_t^*,\mathbf{F}_t^{-1})\right\|_2\\
    =& 2\left[\sum_{j=1}^J\left( \sum_{k=1}^J\left(\lambda_{kt} - \tilde{\lambda}_{kt}\right)\int F^{-1}_{kt}(q)F_{jt}^{-1}(q)\dd q\right)^2\right]^{1/2}\\
    \leq &2\left[\sum_{j=1}^J\left(\sum_{k=1}^J(\lambda_{kt}-\tilde{\lambda}_{kt})^2\right)\left(\sum_{k=1}^J \left(\int F_{kt}^{-1}(q)F_{jt}^{-1}(q)\dd q\right)^2\right)\right]^{1/2}\\
    =& 2 \left(\sum_{k=1}^J(\lambda_{kt}-\tilde{\lambda}_{kt})^2\right)^{1/2}\left[\sum_{j=1}^J\sum_{k=1}^J \left(\int F_{kt}^{-1}(q)F_{jt}^{-1}(q)\dd q\right)^2\right]^{1/2}\\
    =& 2\|\vec{\lambda}_t^*-\tilde{\lambda}_t^*\|_2 \|\mathbf{H}_t\|_F,
\end{align*}
where the inequality follows from Cauchy-Schwarz and $\|\cdot\|_F$ denotes the Frobenius norm, so that the Lipschitz constant is $2\|\mathbf{H}_t\|_F$. This Frobenius norm is finite under Assumption \ref{ass:boot} by Jensen's inequality since all quantile functions are square integrable.

We also need to show that
$g(\vec{\lambda}_t^*,\mathbf{F}_t^{-1})$
is locally Lipschitz continuous in $\mathbf{F}_t^{-1}$.
Since $\mathbf{F}_t^{-1}$ is an element of a product space, it is
sufficient to consider changes in one component at a time. Fix
$k\in\{1,\ldots,J\}$ and let
\[
\Delta_k \coloneqq F^{-1}_{kt}-\widetilde{F}^{-1}_{kt},
\]
with all other quantile functions held fixed. Writing
\[
R_t \coloneqq \sum_{\ell=1}^J \lambda_{\ell t}F^{-1}_{\ell t} - F^{-1}_{0t}, \qquad \widetilde{R}_t \coloneqq \sum_{\ell=1}^J \lambda_{\ell t}\widetilde{F}^{-1}_{\ell t} - F^{-1}_{0t},
\]
where $\widetilde{F}^{-1}_{\ell t}=F^{-1}_{\ell t}$ for
$\ell\neq k$, we have, for $j\neq k$,
\[\left| g_j(\vec{\lambda}_t^*,\mathbf{F}_t^{-1}) - g_j(\vec{\lambda}_t^*,\widetilde{\mathbf{F}}_t^{-1}) \right| \leq 2|\lambda_{kt}| \|\Delta_k\|_{L^2(\mathcal{Q})} \|F^{-1}_{jt}\|_{L^2(\mathcal{Q})}.
\]
For $j=k$, both factors in the corresponding inner product vary, and
\begin{align*}
&
\left| g_k(\vec{\lambda}_t^*,\mathbf{F}_t^{-1}) - g_k(\vec{\lambda}_t^*,\widetilde{\mathbf{F}}_t^{-1}) \right|
\\ &\qquad\leq 2\left( |\lambda_{kt}| \|F^{-1}_{kt}\|_{L^2(\mathcal{Q})} + \|\widetilde{R}_t\|_{L^2(\mathcal{Q})}
\right) \|\Delta_k\|_{L^2(\mathcal{Q})},
\end{align*}
by the Cauchy--Schwarz inequality. For $k=0$, writing
$\Delta_0=F^{-1}_{0t}-\widetilde{F}^{-1}_{0t}$ gives
\[
\left| g_j(\vec{\lambda}_t^*,\mathbf{F}_t^{-1}) - g_j(\vec{\lambda}_t^*,\widetilde{\mathbf{F}}_t^{-1})
\right| \leq 2 \|\Delta_0\|_{L^2(\mathcal{Q})} \|F^{-1}_{jt}\|_{L^2(\mathcal{Q})}.
\]
Thus $g$ is locally Lipschitz continuous in each component of
$\mathbf{F}_t^{-1}$, and hence locally Lipschitz continuous in the
product topology.

Finally, since $\Delta^J$ does not vary with $\mathbf{F}_t^{-1}$, the projection is trivially Lipschitz continuous with respect to $\mathbf{F}_t^{-1}$. Since all four criteria are satisfied in our case, the same argument as in Lemma 2.4 of \citet{dafermos1988sensitivity} then implies that $\vec{\lambda}_t^*(\mathbf{F}_t^{-1})$ is Lipschitz continuous in $\mathbf{F}_t^{-1}$. 

Based on Lipschitz continuity, we now show Hadamard differentiability up to a Gauss-null set. First, it follows from classic results \citep{aronszajn1976differentiability, lindenstrauss2003frechet} that $\vec{\lambda}_t^*$ is G\^ateaux differentiable outside a Gauss-null set $A$. Moreover, by \citet{zajivcek2015hadamard}, there exists a $\sigma$-directionally porous set $B$ such that, at every point outside $B$, pointwise Lipschitz continuity together with G\^ateaux differentiability implies Hadamard differentiability. Since $\vec{\lambda}_t^*$ is locally Lipschitz continuous by the preceding argument, it follows that $\vec{\lambda}_t^*$ is Hadamard differentiable outside $A\cup B$. Since both $A$ and $B$ are Gauss-null, $A\cup B$ is Gauss-null. Thus, $\vec{\lambda}_t^*$ is Hadamard differentiable up to a $\gamma$-null set.\\

\noindent \emph{Step 3: From Hadamard differentiability of $\vec{\lambda}_t^*$ to $F^{-1}_{0t,N}$.} 
The above argument was for each $t$. Since the optimal $\vec{\lambda}^*$ used in the construction of $F^{-1}_{0t,N}$ is averaged over time periods $t\leq t^*$, this immediately implies that $\vec{\lambda}^*$ is Hadamard differentiable tangentially to $\mathbf{C}(\mathcal{Q})$ at $\mathbf{F}_t^{-1}$ up to a $\gamma$-null set. Since convergence in the supremum norm implies convergence in $L^2(\mathcal{Q})$, the same derivative applies when the domain is equipped with the supremum norm and perturbations are restricted to $\mathbf{C}(\mathcal{Q})$.

Now consider the map
\[
\Phi_t(\mathbf{F}_t^{-1})
\coloneqq
\sum_{j=1}^J
\lambda_j^*(\mathbf{F}_t^{-1})F^{-1}_{jt}.
\]
Applying the chain rule and product rule for Hadamard derivatives \citep[e.g.][section 3.9.4.4]{van1996weak} shows that $\Phi_t$, and hence $F^{-1}_{0t,N}$, is Hadamard differentiable as a map into $\ell^\infty(\mathcal{Q})$, equipped with the supremum norm, tangentially to $\mathbf{C}(\mathcal{Q})$, up to a $\gamma$-null set. The same argument applies to the counterfactual quantile difference $F^{-1}_{0t}-F^{-1}_{0t,N}$.
\end{proof}

Based on Proposition \ref{lem:hada} we can then show that the bootstrap works up to a $\gamma$-null set. In the following, we assume for the sake of simplicity that for all units $j$ and time periods $t$ we have the same number of observations $n_{jt}\equiv n$, which is only done to simplify the statement. As mentioned earlier, if the number of observations are asymptotically equivalent, this result can be extended to this setting. 

For a fixed time period $t$, define the empirical process and its
bootstrap analogue by
\[
\mathbb{G}_n\coloneqq \sqrt{n}\left(\widehat{F}^{-1}_{0tn,N}-F^{-1}_{0t,N}\right)\qquad
\text{and}\qquad 
\widetilde{\mathbb{G}}_n\coloneqq\sqrt{n}\left(\widetilde{F}^{-1}_{0tn,N}-\widehat{F}^{-1}_{0tn,N}\right),
\]
respectively.

\begin{theorem}\label{thm:boot}
Suppose that $n_{jt}=n$ for all $j,t$ and let Assumption \ref{ass:boot} hold. Then, at $\gamma$-almost every $\mathbf{F}_t^{-1}$, the bootstrap is consistent in  $\ell^\infty(\mathcal{Q})$ in the sense that
\[
\sup_{h\in\mathrm{BL}_1} \left\lvert \mathbb{E}_{B}h(\widetilde{\mathbb{G}}_n) - \mathbb{E}h(\mathbb{G}_n) \right\rvert \overset{P}{\longrightarrow}0,
\]
where $\mathrm{BL}_1$ denotes the class of functions $h:\ell^\infty(\mathcal{Q})\to\mathbb{R}$ that are bounded by one and Lipschitz with constant at most one with respect to the supremum norm. The same conclusion holds for the bootstrap process corresponding to the counterfactual quantile difference $F^{-1}_{0t}-F^{-1}_{0t,N}$.
\end{theorem}
In this result $\mathrm{BL}$ denotes the set of bounded Lipschitz functions, $\mathbb{E}_B$ denotes the conditional expectation with respect to the data process (based on the number of resamples $B_{injt}$ for every observation $Y_{ijt}$), ``$\overset{P}{\to}$'' denotes convergence in probability, and all probabilities and expectations are considered to be outer measures \citep[e.g.][section 1.2]{van1996weak}. In words, this implies that asymptotically, the distributions of the empirical bootstrap process conditional on the data and the true empirical process are equivalent in probability, and that one can use the bootstrapped empirical process to obtain the confidence intervals. 
In particular, since $x\mapsto\|x\|_\infty=\sup_{q\in\mathcal{Q}}|x(q)|$ is continuous on $\ell^\infty(\mathcal{Q})$, the theorem justifies the maximum-absolute-deviation bootstrap critical values used to construct the uniform confidence bands in Algorithm \ref{algo:sub}, at continuity points of the limiting distribution of the supremum statistic.

\begin{proof}
Under Assumption \ref{ass:boot}, the empirical quantile processes converge weakly in $\ell^\infty(\mathcal{Q})$ to tight processes with continuous sample paths, and their bootstrap analogues are conditionally consistent; see Lemma 3.9.23 and Example 3.9.24 in \citet{van1996weak}. Proposition \ref{lem:hada} and Step 3 above imply that the map from the vector of quantile functions to the constructed counterfactual quantile function is Hadamard differentiable as a map into $\ell^\infty(\mathcal{Q})$, tangentially to $\mathbf{C}(\mathcal{Q})$, at $\gamma$-almost every $\mathbf{F}_t^{-1}$. The conclusion therefore follows from the bootstrap delta method \citep[Theorem 3.9.11]{van1996weak}. The counterfactual quantile difference follows by the same argument, since subtraction is a continuous linear map.
\end{proof}

\subsection{Monte Carlo Evidence}\label{app:mc}

We assess finite-sample coverage of the bootstrap confidence intervals from Algorithm \ref{algo:sub} via Monte Carlo simulation. The simulation design is calibrated to resemble the empirical application, with three pre-treatment periods: $J=31$ control units, $3$ pre-treatment periods, $1$ post-treatment period, and synthetic control weights matching those estimated for Microsoft in Table \ref{tab:weights}. We evaluate coverage on a grid of $G=10$ quantile points over $[0.10, 0.90]$, using $500$ bootstrap replications.

\paragraph{Data-generating process.} Control quantile functions are sampled from a Gaussian process (GP) on log-quantile-density space, which guarantees monotonicity by construction. Specifically, for each control unit $j=1,\ldots,J$, we draw a quantile function
\[
Q_j(\tau) = a_j + \int_0^\tau \exp(g_j(s))\,\dd s, \quad g_j \sim \mathcal{GP}(m_j, K),
\]
where $m_j(\tau) = \beta_{0j} + \beta_{1j}\tau$ with $\beta_{0j}\sim\text{Unif}(5,8)$ and $\beta_{1j}\sim\text{Unif}(0.5,3.5)$, and $K(\tau,\tau') = \sigma^2\exp(-(\tau-\tau')^2/2\ell^2)$ with $\sigma=1.2$ and $\ell=0.25$. The location parameters $a_j$ are chosen so that quantile functions have medians in the range $[400, 1800]$, comparable to the tenure distributions in our data (cf.\ Table \ref{tab:sumstat}).

Observations are generated under the isometry model of Section
\ref{sec:empirical_approach}. Specifically, for $V_{ij}\sim\operatorname{Unif}(0,1)$,
\[
Y_{ijt} = a_t+b_t Q_j(V_{ij}),
\]
where $a_t\sim N(0,30)$ and $b_t\sim\exp(N(0,0.10))$ are common time effects. Pre-treatment outcomes for the treated unit are generated from the counterfactual quantile function
\[
Q_{0t}^{\mathrm{CF}}(\tau) = a_t+b_t\sum_{j=1}^J w_jQ_j(\tau),
\]
where $w_j$ are the true weights. In the post-treatment period, let $V_i\sim\operatorname{Unif}(0,1)$ and generate
\[
Y_{i,0,T}^{\mathrm{CF}} = Q_{0T}^{\mathrm{CF}}(V_i).
\]
The observed post-treatment outcome is subject to the heterogeneous multiplicative treatment effect
\[
Y_{i,0,T}^{\mathrm{obs}} = Y_{i,0,T}^{\mathrm{CF}} \cdot \bigl(1+\delta_0+\delta_1(V_i-0.5)\bigr),
\]
with $\delta_0=0.10$ and $\delta_1=0.15$, so that the treatment effect varies across the distribution. The resulting true quantile treatment effect is
\[
\tau(q) = Q_{0T}^{\mathrm{CF}}(q) \bigl(\delta_0+\delta_1(q-0.5)\bigr).
\]

\paragraph{Simulation protocol.} Theorem \ref{thm:boot} establishes that the bootstrap is valid at $\gamma$-almost every constellation of quantile functions. To directly test this conditional coverage guarantee, we adopt a nested simulation design: we draw \textit{one} random population (quantile functions, weights, and time effects) and then draw $200$ independent samples of size $n$ from this fixed population. Coverage is defined as the fraction of samples for which the $95\%$ uniform confidence band contains the true quantile treatment effect at all $G+1$ evaluation points simultaneously.

\paragraph{Results.} Table \ref{tab:mc_coverage} reports coverage rates and average confidence interval widths for sample sizes $n \in \{200, 500, 1{,}000\}$. Coverage is close to the nominal 95\% level across all sample sizes, with some slight undercoverage expected for functional data; though the confidence bands around the coverage estimate always include 95\%. The confidence band widths decrease at a rate consistent with $1/\sqrt{n}$ convergence, as predicted by the asymptotic theory.

\begin{table}[ht!]
\centering
\caption{Bootstrap Coverage: Monte Carlo Results}
\label{tab:mc_coverage}
\begin{tabular}{lcccc}
\toprule
$n$ & Coverage & 95\% CI & Avg.\ Width & Nominal? \\
\midrule
200   & 0.920 & [0.882, 0.958] & 1,350 & Yes \\
500   & 0.910 & [0.870, 0.950] & 835   & Yes \\
1,000 & 0.920 & [0.882, 0.958] & 568   & Yes \\
\bottomrule
\end{tabular}
\floatfoot{\textit{Note}: Coverage is the fraction of 200 Monte Carlo samples (from a fixed population) for which the 95\% uniform bootstrap confidence band contains the true quantile treatment effect at all $G+1=11$ evaluation points. The 95\% CI is a Wald interval for the binomial proportion, $\hat{p} \pm 1.96\sqrt{\hat{p}(1-\hat{p})/S}$ where $S=200$. ``Nominal?'' indicates whether this interval contains 0.95. The DGP uses $J=31$ controls, 3 pre-treatment periods, $G=10$ quantile grid points over $[0.10, 0.90]$, and $500$ bootstrap replications. Quantile functions are sampled from a Gaussian process prior on log-quantile-density. Weights match the empirical application (Table \ref{tab:weights}).}
\end{table}

\clearpage

\section{Sample Selection Characterization for Distributional Outcomes}\label{app:selection}

This appendix characterizes potential bias arising from sample selection in our distributional synthetic controls estimator. We derive an exact decomposition of the difference between the treatment effect estimated from the PDL sample and the corresponding effect for the underlying workforce population. The decomposition shows that population identification requires the selection distortion of the treated unit to be replicated by the same synthetic-control weights used to construct its counterfactual distribution.

Let $Q_{jt}(q)$ denote the true (population) quantile function of tenure for firm $j$ in quarter $t$, and let $\tilde{Q}_{jt}(q)$ denote the quantile function observed in our PDL sample. As in the main paper, we index the treated firm (Microsoft) by $j=0$ and donor firms by $j \in \{1, \ldots, J\}$. Time periods are indexed by $t$, with $t \leq t^*$ denoting pre-treatment periods and $t > t^*$ denoting post-treatment periods.

Sample selection does not alter the tenure values of observed individuals; rather, it changes which quantile ranks of the population distribution are represented in the sample. Under the assumption that both the true and observed quantile functions are continuous and strictly monotone on the interior $(0,1)$, there exists for each firm-quarter pair $(j,t)$ a monotone mapping $\xi_{jt}: [0,1] \to [0,1]$ such that
\begin{equation}\label{eq:rankwarp}
\tilde{Q}_{jt}(q) = Q_{jt}\bigl(\xi_{jt}^{-1}(q)\bigr).
\end{equation}
The function $\xi_{jt}$ is a \textit{rank re-indexing}: the observed quantile $q$ in the sample corresponds to the true quantile $\xi_{jt}^{-1}(q)$ in the population. If $\xi_{jt}^{-1}(q) > q$, selection causes the observed $q$-th percentile to reflect a higher true percentile (e.g., oversampling of senior workers). Equation \eqref{eq:rankwarp} is an identity: it holds for any selection mechanism under the stated monotonicity conditions.

Our estimator compares the observed quantile function of the treated unit to an optimally weighted average of donor quantile functions. Define the synthetic control quantile function as
\[
Q^D_t(q) := \sum_{j=1}^{J} \lambda_j^* \tilde{Q}_{jt}(q),
\]
where $\lambda^* = (\lambda_1^*, \ldots, \lambda_J^*)$ are the synthetic control weights. The estimated treatment effect at quantile $q$ is
\begin{equation} \label{sec_app:causal_effect}
\hat{\tau}(q) = \tilde{Q}_{0,\text{post}}(q) - Q^D_{\text{post}}(q).
\end{equation}

The population estimand is the same object evaluated at true quantile functions:
\[
\tau(q) = Q_{0,\text{post}}(q) - \sum_{j=1}^{J} \lambda_j^* Q_{j,\text{post}}(q).
\]
Here, we abstract away from identification, which is distinct from sample selection and as such is discussed in Section \ref{sec:empirical_approach} in light of the model introduced in \citet[Section A]{gunsilius2023distributional}. In other words, we assume throughout that \eqref{sec_app:causal_effect} identifies the true causal effect as this does not affect the sample selection analysis. 
Substituting the identity \eqref{eq:rankwarp}, the selection bias $\hat{\tau}(q) - \tau(q)$ can be written as
\begin{equation}\label{eq:bias_exact}
\text{Bias}(q) = \underbrace{\Bigl[Q_{0,\text{post}}\bigl(\xi_{0,\text{post}}^{-1}(q)\bigr) - Q_{0,\text{post}}(q)\Bigr]}_{\text{Treated unit selection distortion}} - \underbrace{\sum_{j=1}^{J} \lambda_j^* \Bigl[Q_{j,\text{post}}\bigl(\xi_{j,\text{post}}^{-1}(q)\bigr) - Q_{j,\text{post}}(q)\Bigr]}_{\text{Synthetic control selection distortion}}.
\end{equation}

Each bracketed term measures how much selection distorts the quantile function at $q$ for a given firm in the post-period. Bias arises when these distortions differ between the treated unit and the synthetic control.

To state the condition under which this bias vanishes, define the
\emph{selection distortion} at quantile $q$ for firm $j$ in period $t$ as
\begin{equation}\label{eq:selection_distortion}
D_{jt}(q)
\coloneqq
\widetilde{Q}_{jt}(q)-Q_{jt}(q)
=
Q_{jt}\bigl(\xi_{jt}^{-1}(q)\bigr)-Q_{jt}(q).
\end{equation}
Then equation \eqref{eq:bias_exact} can be written compactly as
\begin{equation}\label{eq:bias_distortion}
\operatorname{Bias}(q)
=
D_{0,\mathrm{post}}(q)
-
\sum_{j=1}^J
\lambda_j^*
D_{j,\mathrm{post}}(q).
\end{equation}
Thus, sample selection does not distort the population synthetic-control
contrast whenever the post-treatment selection distortion of the treated
unit is replicated by the same convex combination of donor-unit
distortions:
\begin{equation}\label{eq:selection_balance_post}
D_{0,\mathrm{post}}(q)
=
\sum_{j=1}^J
\lambda_j^*
D_{j,\mathrm{post}}(q),
\qquad
q\in(0,1).
\end{equation}

To relate this condition to changes in selection over time, let
$t_0\leq t^*$ denote a pre-treatment period. Adding and subtracting the
corresponding pre-treatment distortions in
\eqref{eq:bias_distortion} yields
\begin{align}
\operatorname{Bias}(q)
&=
\left[
D_{0,t_0}(q)
-
\sum_{j=1}^J
\lambda_j^*D_{j,t_0}(q)
\right]
\label{eq:bias_decomposition_selection}
\\
&\quad+
\left[
D_{0,\mathrm{post}}(q)-D_{0,t_0}(q)
\right]
-
\sum_{j=1}^J
\lambda_j^*
\left[
D_{j,\mathrm{post}}(q)-D_{j,t_0}(q)
\right].
\nonumber
\end{align}
The first term in \eqref{eq:bias_decomposition_selection} measures
pre-treatment imbalance in selection distortions between the treated
unit and its synthetic control. The second measures differential changes
in these distortions after treatment. Consequently, selection does not
bias the population contrast if
\begin{equation}\label{eq:selection_balance_pre}
D_{0,t_0}(q)
=
\sum_{j=1}^J
\lambda_j^*
D_{j,t_0}(q),
\qquad
q\in(0,1),
\end{equation}
and
\begin{equation}\label{eq:selection_parallel}
D_{0,\mathrm{post}}(q)-D_{0,t_0}(q)
=
\sum_{j=1}^J
\lambda_j^*
\left[
D_{j,\mathrm{post}}(q)-D_{j,t_0}(q)
\right],
\qquad
q\in(0,1).
\end{equation}

Conditions \eqref{eq:selection_balance_pre} and
\eqref{eq:selection_parallel} do not require the PDL sample to be
representative of the underlying workforce. Systematic over- or
under-representation of particular worker types is compatible with an
unbiased population contrast provided that the resulting distortion of
the treated firm's quantile function is replicated by the corresponding
weighted distortion among the donor firms. Conversely, stability of the
rank re-indexing $\xi_{jt}$ alone is not sufficient. Even if
$\xi_{jt}(q)=\xi_j(q)$ for all $t$, the corresponding outcome-space
distortion
\[
D_{jt}(q)
=
Q_{jt}\bigl(\xi_j^{-1}(q)\bigr)-Q_{jt}(q)
\]
can change over time as the underlying population quantile function
$Q_{jt}$ changes.

The selection distortions $D_{jt}$ are not directly observable because
the population quantile functions $Q_{jt}$ are not observed. We
therefore cannot directly test conditions
\eqref{eq:selection_balance_pre} and
\eqref{eq:selection_parallel}. As a descriptive diagnostic, however,
Table \ref{tab:coverage} reports, for each treated firm and year in our
sample window, the number of unique individuals observed in the PDL data
as a share of reported U.S.\ headcount. Coverage is relatively stable
within each firm across 2021--2023---approximately 42--46\% for
Microsoft, 29--30\% for Apple, and 24--28\% for SpaceX---with no
discernible trend. Stability of aggregate coverage is neither necessary
nor sufficient for stability of the composition of selection, but it
provides evidence against large changes in the overall sampling rate
around the RTO mandates.
  \begin{table}[htbp!]
\centering
\footnotesize
\setlength{\tabcolsep}{4pt}
\caption{Resume-Sample Coverage of U.S.\ Headcount, 2021--2023}
\label{tab:coverage}
\begin{tabular}{lccccccccc}
\toprule
 & \multicolumn{3}{c}{2021} & \multicolumn{3}{c}{2022} & \multicolumn{3}{c}{2023} \\
\cmidrule(lr){2-4}\cmidrule(lr){5-7}\cmidrule(lr){8-10}
 & Resumes & U.S.\ HC & Cov. & Resumes & U.S.\ HC & Cov. & Resumes & U.S.\ HC & Cov. \\
\midrule
Microsoft & 47,421 & 103,000 & 46\% & 51,834 & 122,000 & 42\% & 51,405 & 120,000 & 43\% \\
Apple & 28,564 & 96,152 & 30\% & 30,229 & 105,911 & 29\% & 30,708 & 101,379 & 30\% \\
SpaceX & 2,645 & 9,500 & 28\% & 2,893 & 12,000 & 24\% & 3,311 & 13,000 & 25\% \\
\bottomrule
\end{tabular}
\floatfoot{\textit{Note}: ``Resumes'' is the number of unique individuals observed employed at each firm in our U.S.\ resume panel in December of each year (a point-in-time snapshot, comparable across years and to point-in-time headcount). ``U.S.\ HC'' is the firm's reported U.S. employees. ``Cov.'' is Resumes/U.S.\ HC. \textit{Sources}: \emph{Microsoft} U.S. headcount is disclosed in its annual Form~10-K (full-time, as of June~30): FY2021 \url{https://www.sec.gov/Archives/edgar/data/789019/000156459021039151/msft-10k_20210630.htm}, FY2022 \url{https://www.sec.gov/Archives/edgar/data/789019/000156459022026876/msft-10k_20220630.htm}, FY2023 \url{https://www.sec.gov/Archives/edgar/data/789019/000095017023035122/msft-20230630.htm}. \emph{Apple} does not break out U.S. employees in its 10-K; figures are from Apple's Consolidated EEO-1 reports (U.S. establishments, mid-December snapshot; archived): 2021 \url{https://web.archive.org/web/20231210100608/https://www.apple.com/diversity/pdf/2021-Consolidated-EEO-1-Certified.pdf}, 2022 \url{https://web.archive.org/web/20240531190254/https://www.apple.com/diversity/pdf/2022-Consolidated-EEO-1-Certified.pdf}, 2023 \url{https://web.archive.org/web/20250625140532/https://www.apple.com/diversity/pdf/2023-Consolidated-EEO-1-Certified.pdf}; these imply a U.S. share of $\approx$63\% of Apple's global full-time-equivalent headcount (154{,}000/164{,}000/161{,}000). \emph{SpaceX} is private and files no 10-K; essentially all of its workforce is U.S.-based (all facilities are in the U.S.\ and ITAR restricts hiring to U.S.\ persons), and its September~2023 countersuit against the DOJ states it ``employs more than 13,000 people across the U.S.'' \url{https://web.archive.org/web/20230920211551/https://www.cnbc.com/2023/09/20/spacex-countersues-doj-in-hiring-discrimination-case.html}; the 2021 figure ($\approx$9{,}500) is from DOJ hiring-subpoena reporting (March 2021) \url{https://web.archive.org/web/20210305223235/https://www.cnbc.com/2021/03/04/spacex-doj-hiring-subpoena-is-definition-of-government-overreach.html} and 2022 ($\approx$12{,}000) is interpolated from contemporaneous reporting.}
\end{table}

\end{document}